\newif\ifsubmit     
\newif\ifllncs      
\newif\ifexabs      
\newif\ifblind  

\submittrue

\ifllncs
\documentclass[runningheads,a4paper]{llncs}

  \usepackage[style=alphabetic,minalphanames=3,maxalphanames=4,maxnames=99,backref=true]{biblatex}

\else \documentclass[letterpaper,11pt,pdfa]{article}
  \usepackage[in]{fullpage}
\fi

\usepackage{iftex}
\ifPDFTeX
  \usepackage[utf8]{inputenc}
  \usepackage[noTeX]{mmap}
  \usepackage[T1]{fontenc}
\fi
\ifLuaTeX
  \usepackage{luatex85}
  \usepackage[noTeX]{mmap}
\fi

\ifllncs
\else

\fi

\usepackage{mdframed}
\usepackage{amsmath}
\usepackage{amsfonts}
\usepackage{amssymb}
\usepackage{amsthm}
\usepackage{color}
\usepackage[bookmarks]{hyperref}
\usepackage[nameinlink]{cleveref}

\usepackage{mdframed}
\usepackage{amsmath}
\usepackage{amsfonts}
\usepackage[nottoc,numbib]{tocbibind}
\usepackage{amssymb}
\usepackage{amsthm}
\usepackage{color}
\usepackage[bookmarks]{hyperref}
\usepackage[nameinlink]{cleveref}
\ifllncs

  \spnewtheorem{claim}{Claim}{\bfseries}{\rmfamily}
  \crefname{claim}{claim}{claims}
  \Crefname{claim}{Claim}{Claims}
   \spnewtheorem{fact}{Fact}{\bfseries}{\rmfamily}
  \crefname{fact}{fact}{facts}
  \Crefname{fact}{fact}{Facts}
\else
  \newtheorem{theorem}{Theorem}[section]
  \newtheorem{definition}[theorem]{Definition}
  \newtheorem{remark}[theorem]{Remark}
 
  \newtheorem{lemma}[theorem]{Lemma}

  \newtheorem{claim}[theorem]{Claim}
  \newtheorem*{remark*}{Remark}
  
  \newtheorem{fact}[theorem]{Fact}

  \newtheorem*{theorem*}{Theorem}
  \newtheorem*{lemma*}{Lemma}

  \newcommand{\email}[1]{\href{mailto:#1}{\texttt{#1}}}

\ifllncs
\else
\fi

\fi

\usepackage{appendix}
\usepackage{algorithm}
\usepackage{algpseudocode}
\usepackage{braket}
\usepackage{comment}
\usepackage{url}
\usepackage{multicol}
\usepackage{tikz}
\usetikzlibrary{arrows,automata,positioning}
\usepackage{caption}
\usepackage[caption=false]{subfig}
\usepackage[font=small,labelfont=bf]{caption}
\usepackage{comment}
\usepackage{pifont}
\usetikzlibrary{patterns}

\usepackage{enumitem}
\setlist[description]{noitemsep}
\setlist[enumerate]{noitemsep}
\setlist[itemize]{noitemsep}

\usepackage{soul, xcolor, xparse}
\makeatletter
  \ExplSyntaxOn
    \cs_new:Npn \white_text:n #1
    {
      \fp_set:Nn \l_tmpa_fp {#1 * .01}
      \llap{\textcolor{white}{\the\SOUL@syllable}\hspace{\fp_to_decimal:N \l_tmpa_fp em}}
      \llap{\textcolor{white}{\the\SOUL@syllable}\hspace{-\fp_to_decimal:N \l_tmpa_fp em}}
    }
    \NewDocumentCommand{\whiten}{ m }
    {
      \int_step_function:nnnN {1}{1}{#1} \white_text:n
    }
  \ExplSyntaxOff
  
  \NewDocumentCommand{ \varul }{ D<>{5} O{0.2ex} O{0.1ex} +m } {%
    \begingroup
    \setul{#2}{#3}%
    \def\SOUL@uleverysyllable{%
      \setbox0=\hbox{\the\SOUL@syllable}%
      \ifdim\dp0>\z@
      \SOUL@ulunderline{\phantom{\the\SOUL@syllable}}%
      \whiten{#1}%
      \llap{%
        \the\SOUL@syllable
        \SOUL@setkern\SOUL@charkern
      }%
      \else
      \SOUL@ulunderline{%
        \the\SOUL@syllable
        \SOUL@setkern\SOUL@charkern
      }%
      \fi}%
    \ul{#4}%
    \endgroup
  }
\makeatother

\usepackage{braket}

\ifsubmit    
    \newcommand{\jiahui}[1]{}
    \newcommand{\markz}[1]{}
\else 
    \newcommand{\jiahui}[1]{{\color{purple} Jiahui: #1}}
    \newcommand{\markz}[1]{{\color{blue} Mark: #1}}
\fi

\newcommand{\Enc}{\mathsf{Enc}}
\newcommand{\KeyGen}{\mathsf{KeyGen}}

\newcommand{\Dec}{\mathsf{Dec}}

\newcommand{\td}{\mathsf{td}}
\newcommand{\pk}{\mathsf{pk}}
\newcommand{\mpk}{\mathsf{mpk}}
\newcommand{\sk}{\mathsf{sk}}
\newcommand{\ct}{\mathsf{ct}}

\newcommand{\id}{\mathsf{id}}

\newcommand{\wibe}{\mathsf{wIBE}}

\newcommand{\ibe}{\mathsf{IBE}}

\newcommand{\aux}{\mathsf{aux}}

\newcommand{\ver}{{\sf Ver}}

\renewcommand{\st}{{\sf st}}
\newcommand{\prove}{{\sf Prove}}
\newcommand{\sign}{{\sf Sign}}

\newcommand{\sig}{{\sf sig}}

\newcommand{\OTS}{{\sf OTS}}

\newcommand{\DS}{{\sf DS}}

\newcommand{\negl}{\mathsf{negl}}

\newcommand{\A}{\mathcal{A}}

\newcommand{\adv}{{\sf Adv}}

\newcommand{\poly}{{\sf poly}}

\newcommand{\cA}{\mathcal{A}}
\newcommand{\cD}{{\mathcal{D}}}
\newcommand{\cC}{{\mathcal{C}}}

\newcommand{\cM}{{\mathcal{M}}}

\newcommand{\cS}{{\mathcal{S}}}

\newcommand{\cO}{{\mathcal{O}}}
\newcommand{\cQ}{{\mathcal{Q}}}
\newcommand{\cT}{{\mathcal{T}}}
\newcommand{\cF}{{\mathcal{F}}}

\newcommand{\cY}{\mathcal{Y}}
\newcommand{\cX}{\mathcal{X}}

\newcommand{\sX}{\mathcal{X}}
\newcommand{\sY}{\mathcal{Y}}

\newcommand{\cB}{{\mathcal{B}}}

\newcommand{\N}{\mathbb{N}}

\newcommand{\constrain}{\mathsf{Constrain}}

\newcommand{\csig}{\mathsf{CSig}}

\newcommand{\SecAlg}{\mathsf{SecEval}}

\newcommand{\PubAlg}{\mathsf{PubEval}}

\newcommand{\CSig}{\mathsf{CSig}}

\newcommand{\watermark}{\mathsf{Mark}}
\newcommand{\extract}{\mathsf{Extract}}

\newcommand{\tildesk}{\mathsf{\tilde{sk}}}

\newcommand{\vk}{\mathsf{vk}}
\newcommand{\xk}{\mathsf{xk}}
\newcommand{\mk}{\mathsf{mk}}

\newcommand{\WP}{\mathsf{WP}}
\newcommand{\wsecEval}{\mathsf{wSecEval}}

\newcommand{\wpubEval}{\mathsf{wPubEval}}

\newcommand{\eval}{\mathsf{Eval}}

\newcommand{\pke}{\mathsf{PKE}}

\newcommand{\ske}{\mathsf{SKE}}
\newcommand{\msk}{\mathsf{msk}}

\newcommand{\wmsetup}{\mathsf{WMSetup}}

\newcommand{\view}{\mathsf{view}}

\newcommand{\Sim}{\mathsf{Sim}}

\newcommand{\out}{\mathsf{out}}

\newcommand{\inp}{\mathsf{inp}}

\newcommand{\istar}{{i^*}}

\newcommand{\prf}{\mathsf{PRF}}
\newcommand{\wprf}{\mathsf{wPRF}}

\newcommand{\prp}{\mathsf{PRP}}

\newcommand{\mac}{\mathsf{MAC}}
\newcommand{\wmac}{\mathsf{wMAC}}

\newcommand{\whe}{\mathsf{wHE}}

\newcommand{\wpke}{\mathsf{wPKE}}

\newcommand{\ccapke}{\mathsf{CCA-PKE}}

\newcommand{\cca}{\mathsf{CCA}}

\newcommand{\verify}{\mathsf{Verify}}

\newcommand{\crs}{\mathsf{CRS}}

\newcommand{\nizk}{\mathsf{NIZK}}

\newcommand{\NP}{\mathsf{NP}}

\newcommand{\delegate}{\mathsf{Delegate}}

\newcommand{\ind}{\mathsf{index}}

\newcommand{\setup}{\mathsf{Setup}}

\newcommand{\ConstrainSign}{\mathsf{ConstrainSign}}

\newcommand{\nm}{\mathsf{NM}}

\newcommand{\abe}{\mathsf{ABE}}

\newcommand{\dabe}{\mathsf{DABE}}

\newcommand{\mfe}{\mathsf{MFE}}

\newcommand{\fe}{\mathsf{FE}}

\newcommand{\fhe}{\mathsf{FHE}}

\newcommand{\qtrace}{\mathsf{QTrace}}

\newcommand{\comp}{\mathsf{comp}}

\newcommand{\mode}{\mathsf{mode}}

\newcommand{\skenc}{\mathsf{SKEnc}}

\newcommand{\MSK}{\mathsf{MSK}}

\newcommand{\Marked}{\mathsf{Marked}}

\newcommand{\Unmarked}{\mathsf{Unmarked}}

\newcommand{\SKF}{\mathsf{SKF}}

\usepackage{authblk}

\title{Composability in Watermarking Schemes}
\ifblind
\author{}
\else
\ifllncs
\author{Jiahui Liu\inst{1}\orcidID{0000-0003-4380-8168}, 
\and Mark Zhandry\inst{2}\orcidID{0000-0001-7071-6272}}

\institute{Massachusetts Institute of Technology, \email{jiahuiliu@csail.mit.edu}
\and 
NTT Research, \email{mark.zhandry@ntt-research.com}
}

\index{Liu, Jiahui}
\index{Zhandry, Mark}
\else
\author[1]{Jiahui Liu}
\author[2]{Mark Zhandry}
\affil[1]{Massachusetts Institute of Technology, \email{jiahuiliu.crypto@gmail.com}}
\affil[2]{NTT Research, \email{mark.zhandry@ntt-research.com}}
\fi
\fi
\date{}

\begin{document}

\maketitle

\begin{abstract}Software watermarking allows for embedding a mark into a piece of code, such that any attempt to remove the mark will render the code useless. Provably secure watermarking schemes currently seems limited to programs computing various cryptographic operations, such as evaluating pseudorandom functions (PRFs), signing messages, or decrypting ciphertexts (the latter often going by the name ``traitor tracing''). Moreover, each of these watermarking schemes has an ad-hoc construction of its own.

We observe, however, that many cryptographic objects are used as building blocks in larger protocols. We ask: just as we can compose building blocks to obtain larger protocols, can we compose watermarking schemes for the building blocks to obtain watermarking schemes for the larger protocols? We give an affirmative answer to this question, by precisely formulating a set of requirements that allow for composing watermarking schemes. We use our formulation to derive a number of applications.
\end{abstract}

\section{Introduction}
\jiahui{TODO:
1. add an outline section
2. review all typos according
3. rewrite the general watermarking algorithm, to handle the "honest" circuits
4. Tech overview CCA MAC 

5. Give more intuition on the 2 types of reduction in the technical section!

6. change all watermarking reductions to use the watermarking key to answer queries in stage 1}
Watermarking is an old idea, which aims to embed a mark in some object, such that any attempt to remove the mark destroys the object. In software watermarking, this means embedding a mark into program code, such that any attempt to remove the code will make the code useless. Such watermarking aims to deter piracy by identifying the source of pirated software. Recently, software watermarking has become an active area of research within cryptography, with numerous positive results for watermarking cryptographic functionalities, such as trapdoor functions~\cite{EC:Nishimaki13}, pseudorandom functions~\cite{STOC:CHNVW16,C:KimWu17,goyal2021beyond}, decryption~\cite{CFNP00} (under the name ``traitor tracing''), and more~\cite{goyal2019watermarking}.

In this work, we initiate the study of \emph{composing} watermarked functionalities. That is, if a cryptographic primitive $A$ (or perhaps several primitives) is used to build a primitive $B$, can we use a watermarking scheme for $A$ to realize a watermarking scheme for $B$? Our aim is to show when such a composition is possible, based on properties of the construction and security proof for $B$ using $A$. 

\subsection{Motivation}

Abstractions are central to cryptography, as they allow for decomposing various tasks into smaller building blocks, which can then be instantiated independently. The literature is full of results that show how to generically realize one abstraction assuming solutions to one or more input abstractions.

Given that cryptographic primitives are often composed, and that most watermarking schemes with provable security are for cryptographic primitives, an interesting question is whether watermarking schemes can be composed. To the best of our knowledge, this question has not been previously asked. In contrast, existing watermarking schemes are each developed ``from scratch.'' Even if the underlying techniques are similar, watermarking schemes for different primitives must go through separate constructions and security proofs. This can be a time-consuming process.

\paragraph{What does it mean to watermark a PRF?} As a running example throughout this this introduction, we will use pseudorandom functions (PRFs), which are one of the main workhorses in symmetric key cryptography, and are used as a central component in many higher-level protocols. In~\cite{STOC:CHNVW16}, the authors show how to watermark the evaluation procedure for a certain class of PRFs. We point out, however, that PRFs are typically not considered a cryptographic end goal, but rather a tool used to build other cryptographic notions. So \emph{what, then, is the utility of watermarking a PRF?}

Another fundamental question is the following: for important reasons that we will not get into here, the watermarking guarantee proved by~\cite{STOC:CHNVW16} (and all subsequent work on watermarking PRFs) is weaker than one may expect. Namely, they show that it is impossible to remove the mark without causing the program to fail on \emph{random} inputs. Certain PRFs called ``weak PRFs'' are only guaranteed secure when the adversary sees evaluations on random inputs, and so for this reason may authors (e.g.~\cite{goyal2021beyond,PKC:MaiWu22,EC:KitNis22}) refer to such a scheme as watermarking ``weak PRFs.'' Weak PRFs can be used as building blocks for many applications, though not as many as ordinary PRFs. In the context of using PRFs as a building block, what is implication of watermarking a \emph{weak} PRF?

\paragraph{Composition of Watermarking Schemes.} Our thesis is that watermarking, at least in many cases, should be defined and executed in such a way as to be composable, allowing watermarking schemes for building blocks to generically imply watermarking schemes for higher-level protocols. Watermarking a weak PRF, for example, should generically enable watermarking for many applications of weak PRFs, such as CPA-secure symmetric encryption. Naturally, a more ambitious goal is: watermarking a message-authentication code scheme or a digital signatures scheme, when composed with a watermarkable PRF scheme, should lead to a watermarkable CCA-secure SKE scheme.

\paragraph{Not all compositions support watermarking.} Certainly, not all composition results from cryptography can be used to compose watermarking schemes. For example, pseudorandom functions can be built from any pseudorandom generator (PRG)~\cite{GGM86}. But this does not seem to yield a viable path toward constructing a watermarking scheme for PRFs. After all, what would it even mean to watermarking a PRG, given that the evaluation algorithm for the PRG is public?

\subsection{Overview of Our Results}

\paragraph{Defining Watermarking Security.} Most cryptographic primitives are defined by an interactive game between an adversary $A$ and challenger.  Given a security game, we can translate the game into a watermarking definition, as follows. The attacker gets a watermarked secret key, and then tries to produce a program $A$. We say that $A$ is ``good'' if it can win the security game with non-negligible advantage. We then require the existence of a tracing algorithm, which can extract the mark from any good program $A$ produced by the adversary. In the case of watermarking public key decryption functionalities, our notion corresponds exactly to existing notions of traitor tracing, since the security game is non-interactive. However, for other primitives, our definition is potentially stronger than existing notions: existing notions only ask for mark extraction for non-interactive programs $A$ that can evaluate some function, whereas we must extract from any $A$ which wins a security experiment, which is potentially interactive. The strengthened definitions will be crucial for our composition theorem, which we now describe.

\paragraph{Composition Theorem.} Our first result is a composition theorem, which gives conditions under which a construction of a target primitive $P$ from input primitives $P_1,\cdots,P_k$ can be turned into a compiler for watermarking schemes. 
\begin{theorem}[Main Theorem (Informal)]\label{thm:maininf}
If the construction of a target primitive $P$ from input primitives $P_1,\cdots,P_k$ satisfies some given conditions and the watermarking schemes for $P_1,\cdots,P_k$ satisfy the above watermarking security definition, then we can compose the construction to watermarking schemes for $P_1,\cdots,P_k$, black-boxly into a watermarking scheme for $P$ that satisfies the above watermarking definition.
\end{theorem}
\noindent Our conditions apply to a large class of constructions. They are, very roughly, as follows:
\begin{itemize}
    \item The construction of $P$ from $P_1,\cdots,P_k$ is black-box. Moreover, the secret key $\sk$ for $P$ is $\sk=(\sk_1,\cdots,\sk_k)$ where $\sk_i$ is the secret key for $P_i$.
    \item The security proof for $P$ turns an adversary $A$ for $P$ into adversaries $A_1,\cdots,A_k$ for $P_1,\cdots,P_k$, respectively, such that if $A$ has non-negligible advantage, so does \emph{at least} one of the $P_i$. Typical proofs in cryptography utilizing hybrid proofs will usually have this form.
    \item Moreover, we require a property of the security proof, which we call a ``watermarking-compatible reduction''. This is a rather technical definition, but roughly we allow the security games for $P,P_1,\cdots,P_k$ to consist of two phases, where the adversary's winning condition is only dependent on the second phase. We require that the reduction respects the phases, in the sense that the second stages of $A_1,\cdots,A_k$ only depend on the second stage of $A$.

    \item Depending on the construction and reduction, \emph{not all} input primitives $P_1,\cdots,P_k$ have to be watermarkable to give a watermarking construction for the target primitive $P$. That is, we can compose the "plain" constructions of these input primitives with the watermarkable versions of the other input primitives to give a watermarkable construction for $P$.
    \item If all input primitives (that need to be watermarkable) have collusion-resistant watermarking security, the resulting watermarkable $P$ also satisfies collusion-resistant security. If all input primitives can be traced using only a public key, the same is true of $P$.
\end{itemize}

Many reductions in the literature are watermarking-compatible. For example, consider constructing CPA-secure symmetric encryption from weak PRFs. Here, the first stage for CPA-security consists of all queries occurring prior to the challenge query, and the second stage consists of the challenge query and all subsequent queries. The winning condition does not depend on any first stage CPA queries (or second stage). The first and second stages for the weak PRF are just two rounds of queries to the weak PRF oracle, and here again the win condition does not depend on the actually queries. Thus, this reduction is watermarking compatible.

A non-example would be many proofs involving signatures as the target primitive $P$. The issue is that, the winning condition for a signature scheme security game is that the adversary produces a ``new'' signature on a message that was not seen in a previous query. But checking this win condition requires knowing all the queries. So if there is a first phase where the attacker can make signature queries, then the win condition is not solely dependent on the second stage. If we let the winning condition depend only on the second stage, the adversary can trivially win by querying a signature in the first stage and give it to the second stage adversary as the final output.

\paragraph{Applications.}
We demonstrate that many well-known cryptographic constructions have watermarking-compatible reductions, and therefore we can compose the watermarkable constructions of the input primitives to obtain watermarkable constructions of the target primitive. Within the scope of this work, we give the following examples:
\begin{itemize}
\item Two most simple examples are: 
\begin{itemize}
    \item Watermarkable CPA-secure secret-key encryption scheme from watermarkable weak PRF
    
    \item Watermarkable CCA2-secure secret-key encryption scheme from watermarkable weak PRF and watermarkable MACs.
\end{itemize}

\item Some more advanced examples are: 
\begin{itemize}
    \item Watermarkable CCA2-secure public-key encryption from watermarkable selectively secure identity based encryption and strong one-time signatures, which can in turn be based on LWE. Here, the strong one-time signature we need is a "plain scheme" which does not need to be watermarkable.
    \item Watermarkable CCA2-secure PKE from watermarkable CPA-secure PKE and NIZK (without watermarking), which can in turn be based on LWE too.

    \item Watermarkable weak pseudorandom permutation from watermarkable weak PRF.

    \item Watermarkable CCA2-secure hybrid encryption scheme from watermarkable CCA2-secure PKE and CCA2-secure SKE (without watermarking).
\end{itemize}

\end{itemize}

Additionally, we show that all the input primitives (weak PRF, CPA-secure PKE, selective IBE, signatures, etc.) in the above examples have constructions that satisfy our security definition for watermarking. We can obtain these constructions by using or modifying the existing watermarking schemes in \cite{goyal2019watermarking}, \cite{goyal2021beyond,PKC:MaiWu22}. Therefore, the above composed schemes all have concrete constructions. 

We also briefly discuss how the functional encryption 
 construction in \cite{goldwasser2013reusable} is also watermarking compatible. We cannot possibly elaborate all the concrete examples of watermarking compatible reductions within the scope of this work, but given our generic framework of composition, one can easily verify whether a construction is watermarking compatible by looking into its security proof.

\subsection{Other Related Work}

\cite{TCC:Nis20} presents a general framework for constructing watermarking schemes for any primitive which admits a certain ``all-but-one'' reduction. The work and ours focus on different aspects of watermarking: theirs is focused on constructing watermarking schemes in the first place, whereas our composition theorem shows how to combine watermarking schemes. We also note that the framework in~\cite{TCC:Nis20} seems very much tied to the collusion-free setting, whereas ours is much more general, and can accommodate collusions if the input tracing mechanisms are collusion resistant.

\subsection{Technical Overview}
\label{sec:tech_overview}

\paragraph{Watermarking} We first briefly recall the definition of watermarking a cryptographic primitive: in a watermakrable cryptographic scheme, apart from the usual evaluation algorithms (such as key generation, encrypt, decrypt or
sign and verify), it additionally has a $\watermark$ algorithm and an $\extract$ algorithm (as well as the corresponding marking and extraction keys). The $\watermark$ algorithm allows one to embed a mark into the secret key used to evaluate the cryptographic functionality; the $\extract$ key allows one to extract a mark from an allegedly marked key. The watermarking security, usually referred to as "unremovability", states that given a marked secret key with an adversarially requested mark $\tau$, the adversary should not be able to produce a  circuit that has the same functionality as the secret key such that the above mark 
$\tau$ cannot be extracted from this adversarial circuit.

As explored in this work and some previous works (\cite{goyal2019watermarking,goyal2021beyond}), one important definitional aspect in the above security lies in what we mean by the "adversary's output circuit has the same functionality as the original secret key". We will elaborate in the following sections of the overview.

\paragraph{Watermarking-Compatible Reduction}
Before we go into how we compose watermarking schemes, we expand more on the type of reductions that allow us to build a watermarking composition upon. We call these reductions watermarking-compatible. We will then give a concrete example to help comprehension.

Consider a black-box construction of a target primitive $P$ from input primitives $P_1, \cdots, P_k$ and consider a
reduction algorithm $\cB$ from security of $P$ to the security of $P_i$: we let both the adversary $\A$ and the reduction $\cB$ be divided into two stages. In the first stage, stage-1 adversary $\A_1$, for the security game of $P$, receives some public parameters from stage-1 $\cB_1$ and makes some queries; $\cB_1$ answers these queries by making queries to the oracle provided by the challenger of the security game for $P_i$. 

Entering the second stage, as one can expect, $\A_1$ can give an arbitrary state  it to stage-2 adversary $\A_2$. However, what $\cB_1$ can give to the second stage reduction $\cB_2$ is more restricted: $\cB_1$ can give all the public parameters 
to the second stage but \emph{none of the queries} made by $\A_1$, to $\cB_2$.
$\cB_2$ continues to simulate the query stage to answer $\A$'s queries. In particular, $\cB$ records $\A_2$'s queries. 
Note that the challenge phase of the security game where $\A_2$ (and resp. $\cB$) receives its challenge from the challenger always happens in stage 2 \footnote{But as we will see in some concrete examples, $\A$ can commit to some "challenge messages" in stage 1, which will be taken into stage-2 as an "auxiliary input" and later used by the challenger to prepare the challenge.}. 

In the final output phase, 
$\cB_2$'s answer to the challenger in game $P_i$ will be dependent on (some of) the following pieces of information: challenger's challenge input to $\cB_2$, $\A_2$'s queries made during stage 2, $\cB_2$'s randomness
used to prepare the challenge input for $\A_2$ and $\A_2$'s final output.

We give some intuition on why we need $B$ to be "oblivious" about $\A$'s queries in stage 1 of the above reduction. Looking forward, 
in the actual watermarking (unremovability) security game, we can replace "answering  
 $\A$'s queries in stage 1" with a 
giving out a watermarked secret key to $\A$, where $\cB$ will have no idea what inputs $\A$ has evaluated on using the key. We therefore model the reduction as the above to capture this scenario.

\paragraph{Example: CCA2-Secure Secret-Key Encryption}
To make the above abstract description concrete, we take an example of the reduction from CCA2-security 
to weak PRF and MAC.

We briefly recall the textbook construction: the scheme's secret key consists of the PRF's secret key $\sk_1$ and MAC's secret key $\sk_2$. The encryption algorithm, on input message $m$, computes ciphertext $\ct = (r, \ct' = \prf.\eval(\sk_1, r) \oplus m, \sig)$, where $\sig \gets \mac.\sign(\sk_2, (r,\ct'))$.
The decryption algorithm will first verify the signature $\sig$ on $(r,\ct')$, if yes continue to decrypt using $\sk_1$, else abort.

\jiahui{first divide stages}

It's not hard to see how the CCA2 security game can  fit into the format of a two-stage game we have depicted above: stage 1 consists simply of letting the adversary making queries to the encryption and decryption oracles. Entering stage 2, the adversary $\A$ is allowed to make more queries and then submits the challenge plaintexts; then $\A$ receives challenge ciphertexts from the challenger. $\A$ continues to make more admissible queries and finally outputs its answer.

By viewing the security game in two stages, we will recall the security proof for CCA2 security. The proof can go through a case-by-case analysis: 
\begin{enumerate}
    \item In case 1, in stage 2 of the CCA2 security game, there is a decryption query from $\A$ of the form $\ct = (\ct_0, \sigma)$, such that it has never been the output of an encryption query. Also, the decryption oracle did not output $\bot$ on this query. 

    \item In case 2, there is no such decryption query in stage 2.
\end{enumerate}
In the first case, we can do a reduction to unforgeability of $\mac$:
\begin{itemize}

    \item $\A_\mac$ samples its own $\prf$ secret key $\sk_1$. For stage 1: For any encryption query from 
    stage-1 adversary, $\A^1$, 
 stage-1 reduction $\A^1_\mac$ will simulate the response as follows: it will compute the $\ct_0$ part of ciphertext and then query the signing oracle $\mac.\sign(\sk_2, \cdot)$ in the security game $G_\mac$.
    For any decryption query, $\A_\mac^1$ will query the verification oracle $\mac.\verify(\sk_2, \cdot)$ first and decrypt those whose response is not $\bot$.
    \item After entering stage 2, for any encryption or decryption query from $\A^2$, $\A^2_\mac$ will still simulate the response as the above. Recall that $\A^2_\mac$ will not get to see any queries made in stage 1 and $\A^2_\mac$ will record all queries from $\A_2$.
    
    \item In the challenge phase, 
    $\A^2$ sends in challenge messages $(m_0, m_1); $  $\A^2_\mac$ flips a coin $b \gets \{0,1\}$, prepares and sends the signed encryption $\ct^*$ of $m_b$ to $\A^2$.  $A^2_\mac$ continues to simulate the encryption and decryption oracles for $\A^2$, on queries $\ct \neq \ct^*$.
    
    \item In the end, $\A^2_\mac$ will look up  $\A^2$'s queries: find a decryption query  $\ct = (\ct', \sigma)$ such that it has never been the output of an encryption query and the decryption oracle did not output $\bot$ on this query. $\A^2_\mac$ output $(\ct', \sigma)$ as its forgery.
\end{itemize}

In the second case, we consider a reduction $\A_\prf$ to the  weak pseudorandomness of the PRF. In the following reduction, for he sake of simplicity, let's consider a variant of the weak PRF game: the adversary $\A_\prf$ is given adaptive query access to an oracle $F: \{0,1\}^n \to \{0,1\}^m$ that computes a PRF function; there will be a challenge phase where a challenge input $r^*$ is sampled at random;  $\A_\prf$ will receive one of $(r^*, F(r^*))$ or $(r^*, y \gets \{0,1\}^m)$ at random and try to tell which one it receives \footnote{Watermarkable weak PRF with this type of security can be constructed in \cite{goyal2019watermarking}. 
}.
\begin{itemize}
      \item In stage 1: $\A_\prf$ samples its own $\mac$ key $\sk_2$. For any encryption query from $\A^1$, $\A^1_\prf$ will simulate the response as follows: it will query the PRF oracle $F(\cdot) := \prf.\eval(\sk_1, \cdot)$ on  a fresh $r$ of its own choice. 
    After obtaining $(r,F( r))$, $\A_\prf$ signs the message $(r,F(r) \oplus m)$ using $\sk_2$ and sends the entire ciphertext to $\A$. 

    Similarly, for decryption queries, $\A_\prf$ first verifies the signature and then queries $F( \cdot)$ oracle to decrypt.
    \item After entering stage 2, for any encryption query from $\A^2$, $\A^2_\prf$ will still simulate the response as the above. Recall that $\A^2_\prf$ will not get to see any queries made in stage 1.
    
    \item In the challenge phase, $\A^2_\prf$ will receive a challenge input-output pair $(r^*, y^*)$ from the weak $\prf$ challenger: 
    $r^* \gets \{0,1\}^\ell, y^* \in \{0,1\}^\ell$ where $y^*$ is either $\prf.\eval(\sk_1, r^*)$ or uniformly random. 

    \item $\A^2$ sends in challenge messages, $(m_0, m_1); $  $\A^2_\prf$ flips a coin $b \gets \{0,1\}$ and sends $\ct = (r^*, y^* \oplus m_b, \sigma^*)$ to $\A^2$.


    \item If $\A^2$ guesses the correct $b$, then $\A^2_\prf$ output 0, for "$y^*$ is $\prf.\eval(\sk_1, r^*)$", else $\A^2_\prf$ output 1, for "$y^*$ is uniformly random".
\end{itemize}

Note that in the above reductions, both $\A_\prf$ and $\A_\mac$ do not need the stage-1 queries from $\A$ to help them win their own games.

\paragraph{Definition and Composition Framework of Watermarking}
Now we relate the above reduction to the watermarking goal:
in our watermarking (i.e. unremovability)  security game for primitive $P$, the watermarking adversary $\A$ acts like a "stage-1" adversary in the original security game for the underlying primitive P.
 Then $\A$ produces a program $C$ ----this signifies the end of "stage 1" of the security game.

 A notable feature of our extraction algorithm is the following. The adversarial program $C$ produced by $\A$ will be treated as a stage-2 adversary by the extraction algorithm: the extraction algorithm will try to extract a watermark by having only black-box access to $C$. Since $C$ is supposed to function like a stage-2 adversary that wins the security game of $P$, the extraction procedure ensures $C$ operates properly by simulating the stage-2 security game for $C$.
 
Finally, $\A$ wins if $C$ is a "good" program in winning the corresponding security game for $P$, but no valid watermarks can be extracted from $C$.

Suppose a target primitive $P$ can be built from input primitives $P_1, \cdots, P_k$ via watermarking-compatible reductions, and each watermarking scheme for $P_i$ satisfies our definition, then we can compose the watermarkable versions of $P_1, \cdots, P_k$ to give a watermarkable $P$.

The composition construction works roughly as follows: 
\begin{enumerate}
    \item A watermarkable $P$'s key generation algorithm is similar to the unwatermarkable (plain) construction of $P$: by generating all keys of $P_1, \cdots, P_k$ and concatenate them, for secret/public keys respectively. The extra step is that now we also concatenate the marking keys and extraction keys generated from the watermarkble $P_1, \cdots, P_k$.

    \item How $P$ evaluates is exactly the same as in the plain construction of $P$ form  $P_1, \cdots, P_k$: by running the $P_1, \cdots, P_k$ algorithms as black-box subroutines.

    \item To watermark $P$, we simply watermark the secret keys of the $P_1,\cdots,P_k$ respectively. The watermarked key for $P$ is then just the concatenation of the watermarked keys for $P_1,\cdots,P_k$. Since the construction is black-box and assumes the key is just a tuple of keys for $P_1,\cdots,P_k$,  the evaluation with watermarked keys is still the same, running the evaluation algorithm with black-box from subroutines $P_1,\cdots,P_k$, except now with the marked keys respectively. We obtain correctness (functionality-preserving property), following from the correctness of the plain construction and functionality preserving properties of watermarkable $P_1,\cdots,P_k$.

    \item In order to extract a mark, we turn any pirated algorithm $A$ for $P$ into pirated algorithm $A_1,\cdots,A_k$ for $P_1,\cdots,P_k$. We do this by applying the security reduction for $P$ to the pirated algorithm $A$, and let $A_1,\cdots,A_k$ be the adversaries produced by the reduction. We then interpret $A_1,\cdots,A_k$ as pirated programs for $P_1,\cdots,P_k$, and attempt to extract the mark from each of them.
    Note that any mark extraction algorithm is supposed to use only the input-output behavior of the pirate program.
    
    The guarantee of the reduction is that one of these pirated programs $A_i$ must actually be a "good" program in the security game for the corresponding primitive $P_i$, which means that mark extraction must succeed for that $A_i$. The security proof therefore shows that we are guaranteed to extract some mark \footnote{The $\extract$ algorithm assumes the (input) adversarial circuit to be possibly stateful and interactive, but does not require any input circuit to be stateful and interactive. For example, an honestly generated watermarked circuit is not stateful or interactive in our construction.}.
\end{enumerate}

Formalizing this composition takes some care, as we need to ensure that the mark extraction algorithm has the ability to actually transform $A$ into each of the $A_i$. 
There are a couple issues with getting this to work:

\begin{itemize}
    \item In a reduction, it is typically assumed that, when the reduction is attacking $P_i$, it has access to the keys for $P_j,j\neq i$. This access is often used to simulate the view of $A$ when constructing $A_i$. In our watermarking construction, we therefore need to ensure that the tracing algorithm has this information.

    \item Security proofs often work via hybrid arguments that change various terms. Importantly, the hybrid arguments get to simulate the \emph{entire} game. For us, however, we only get to simulate the second stage of the game; the first stage has actually already been fixed by the time the adversary produces its pirated program $A$. So we have to ensure that the tracing algorithm can carry out the reduction to create a program $A_i$, even though it cannot simulate the entire security experiment from the very beginning.
\end{itemize}

Our definition and proof take care of these issues and more, to give a composition theorem that applies any time our criteria are met.

Note that the restriction to games where the winning condition is independent of the first stage seems necessary. This is because the adversary actually gets in its possession a watermarked key, which lets it compute the various cryptographic functions of the key. For example, in the case of digital signatures, the adversary can actually compute signatures for itself. In the case of encryption, the adversary can encrypt and decrypt messages. Moreover, there is no way to track what the adversary is doing, since it is all happening internally to the adversary rather than being explicitly queried (suppose for example, that the adversary signs a few messages, but has its entire state encrypted under a fully homomorphic encryption (FHE) scheme so that the signature and message being signed are all ``under the hood'' of the FHE scheme). We model this ability by having a first stage of the security experiment where the adversary can freely query the functionalities, but then do not have the win condition depend on those queries.



\paragraph{Watermarking Composition Example: CCA2-Secure SKE}

Now we demonstrate how the above abstract watermarking composition works by using the concrete example of a CCA2-secure SKE built from weak PRF and MAC again.

Consider having input constructions as watermarkable weak PRF and watermarkable MAC which satisfies our syntax and unremovability security requirements,
a watermarking CCA-secure scheme consists of several algorithms: $\wmsetup, \Enc, \Dec, \watermark, \extract$.
The $\wmsetup$ generates the secret key, marking key and extraction key by concatenating the corresponding keys from the PRF and MAC.
The marking algorithm is on input $\sk = (\sk_\prf, \sk_\mac)$ and watermarking message $\tau$, 
output $\tildesk = (\tildesk_\prf \gets \prf.\watermark(\sk_\prf, \tau),  \tildesk_\mac \gets \prf.\watermark(\sk_\mac, \tau)).$

Finally, the $\extract$ algorithm is slightly more complex:
on input extraction key $\xk = (\xk_\mac, \xk_\prf)$ and a circuit $C$, do \emph{both} of the following:
\begin{enumerate}

\item Create a circuit $C_\prf$: 
\begin{itemize}
    \item  Simulate the stage-2 CCA2-security for $C$ (as described in our previous paragraph on watermarking-compatible reduction from CCA2 security to PRF). Note that just as a real reduction $C_\prf$ is only hardcoded with the MAC extraction key $\xk_\mac$ (unless the underlying watermarkable PRF scheme has a public extraction key), which enables it to simulate the signing and verification of MAC; to compute the part that involves the PRF evaluation oracle, $C_\prf$ need to make external queries to some "challenger" to get answers.

    How $C_\prf$ uses $C$ to get its final output is exactly the same as in the watermarking-compatible reduction from CC SKE to weak PRF we described.          
\end{itemize}
After $C_\prf$ is created, the extraction algorithm uses the watermarkable $\prf$'s extraction algorithm to extract a mark from $C_\prf$: $\tau/\bot \gets \prf.\extract(\xk_\prf, C_\prf)$. Note that the "external queries" made by the circuit will be answered now by the $\prf.\extract(\xk_\prf, \cdot)$ algorithm, because it treats $C_\prf$ as a stage-2 adversary in the weak pseudorandomness security game and uses the extraction key $\xk_\prf$ to simulate the game.

\item Create a circuit $C_\mac$: 
\begin{itemize}
    \item  Simulate the stage-2 CCA2-security for $C$ (as described in our previous paragraph on watermarking-compatible reduction from CCA2 security to MAC). Note that just as a real reduction $C_\mac$ is only hard-coded with the PRF extraction key $\xk_\prf$, which enables it to simulate the oracles of PRF; to compute the part that involves the MAC signing/verifying oracles, $C_\mac$ need to make external queries to some "challenger" to get answers.

    How $C_\mac$ uses $C$ to get its final output is exactly the same as in the watermarking-compatible reduction from CCA SKE to MAC we described.          
\end{itemize}
After $C_\mac$ is created, the  extraction algorithm will use the watermarkable $\mac$'s extraction algorithm to extract a mark from $C_\mac$: $\tau/\bot \gets \mac.\extract(\xk_\mac, C_\prf)$. Note that the "external queries" made by the circuit will be answered by the algorithm $\mac.\extract(\xk_\mac, \cdot)$ algorithm, because it treats $C_\mac$ as a stage-2 adversary in the MAC security game and uses the extraction key $\xk_\mac$ to simulate the game.
\end{enumerate}

We say that the adversary $\A$ wins the watermarkable CCA2-secure SKE's unremovability game if no (previously queried) watermark can be extracted from either of the above circuits $C_\prf, C_\mac$ created during the extraction algorithm and $C$ is a "good" program in winning the CCA2 security game. The intuition for the security proof is relatively straightforward: if $C$ is a "good" program for winning CCA2-security game, then at least one of $C_\prf$ and $C_\mac$ must be a "good" program for winning its corresponding security game pf weak PRF or MAC, by the property of the reduction. 

For instance, imagine a reduction algorithm $\cB_\prf$ to the unremovability of watermarkable $\prf$: $\cB_\prf$ will simulate the marking query stage for $\A$   by making marking queries to the marking oracle of the weak PRF challenger, along with the $\mac$ key it sampled on its own.
Note that the evaluation queries to the $\prf$ evaluation oracle can now be replaced by having access to a marked $\prf$ key.
After $\A$ outputs circuit $C$, $\cB_\prf$ simply creates the same circuit $C_\prf'$ as the  $C_\prf$ created in the above extraction algorithm. If the  $C_\prf$ circuit created in the extraction is "good" at winning the weak pseudorandomness game, then so is  $C_\prf'$. Since we cannot extract a watermark from $C_\prf$, neither can we extract from $C_\prf'$ because they have the same input-output behavior when using the same $C$ as its subroutine black-boxly. The same argument applies to the reduction to unremovability of watermarkable MAC.

\paragraph{More Advanced Watermarking Constructions with "Unwatermarkable" Building Blocks}


The watermarking reduction in the above example of CCA2-secure SKE is relatively straightforward to acquire. 
Some similar construction examples include a watermarkable weak PRP from watermarkable weak PRF (\ifllncs\else see \Cref{sec:wm_wprp_from_wprf}\fi).

However,
there exist various constructions that look markedly distinct from the above example of CCA2-secue SKE and may be much more involved.

More importantly, some constructions have building blocks which are cryptographic primitives that don't make sense to watermark, for instance, a non-interactive zero knowledge proof scheme.
Nevertheless, we show that a large class of them can be watermarking-compatible and in many scenarios, we simply do \emph{not} need to watermark all the building blocks.

We briefly discuss two examples on a high level:
\begin{itemize}
    \item \textbf{CCA-secure PKE from IBE and One-Time Signatures} The first example is the CCA2-secure PKE from selectively secure IBE and one-time signatures by \cite{boneh2007chosen}.

    In the original scheme,  the encryption samples a fresh signature signing/verification key pair $(\sk, \vk)$, compute an IBE ciphertext by using the $\vk$ as the identity, then signs this ciphertext using signing key $\sk$, and finally outputs the IBE ciphertext, the signature and $\vk$.

    To decrypt, one first verifies the signature under $\vk$, if valid then derives an identity-embedded decryption key from the IBE master secret key using $\vk$ as the identity and decrypt the ciphertext.

    Note that in this scheme, the keys for one-time signatures are only sampled "online" upon every encryption algorithm and how to watermark such keys is not well-defined. However, as it turns out---we don't have to watermark the signing key in our composed construction. A watermarkable CCA2-secure PKE using the above construction only has to watermark the IBE secret key. The security proof would show that a successful unremovability adversary can help us either break the unremovability security of the watermark IBE scheme or break the strong one-time unforgeability security of the signature scheme. In other words, leveraging the "plain" security of the one-time signature scheme suffices for the composed watermarking scheme.

    In more detail, we can show the following: the pirated algorithm $C$ produced by an adversary is a "good" program for breaking the CCA-security and meanwhile no watermark can be extracted from $C$. However, it does not help us win the selective IBE CPA-security game. Then, simply following the hybrid argument in the security analysis of the CCA-secure PKE itself, we observe that $C$ must be breaking the security of the one-time signature scheme. Thereby, we can use $C$ black-boxly to create a reduction for the one-time signature scheme.

\item  \textbf{CCA-secure PKE from CPA-secure PKE and NIZK} A second example is the \cite{naor1990public} CCA-secure PKE scheme built from CPA-secure PKE and a non-interactive zero-knowledge proof scheme (NIZK). 
The encryption algorithm encrypts a message twice under two different public keys and uses a NIZK proof to prove that they encrypt the same message.

The watermarkable version of the above construction does not watermark the NIZK scheme, but only the two decryption keys of the PKE scheme. How the NIZK scheme is used in the watermarking scheme is less obvious to see: intuitively, the extraction algorithm 
will try to create two circuits that break the underlying CPA security game of the PKE, with the corresponding keys. If we look carefully into the proof for the \cite{naor1990public} scheme, to make this reduction go through, one has to first work in a hybrid game where the real NIZK proof in the challenge ciphertext is replaced with a simulated proof. Our extraction algorithm thereby uses such a simulated proof when interacting with the input circuit $C$. The security of NIZK helps us say that if $C$ is "good" in the original CCA-security game, so will $C$ be good in the game simulated by the extraction algorithm.
\end{itemize}
To summarize, we only have to watermark the secret keys of input primitives $P_i$ which have their keys generated in the main Key Generation algorithm of the target primitive $P$ and have their secret keys used during the evaluation algorithm of $P$.
If an input primitive has its keys sampled "freshly" during every invocation of $P$'s evaluation algorithms or sampled in key generation, but not used in $P$'s evaluation algorithm (e.g. only used in the security proof instead), then we don't need a watermarkable version of $P_i$ to build a watermarkable $P$. Leveraging $P_i$'s original security property suffices.

More advanced examples include the functional encryption from attribute-based encryption, FHE and garbled circuits in \cite{goldwasser2013reusable}, where we only have to watermark the ABE scheme; a hybrid CCA-secure encryption  scheme from CCA-secure PKE and CCA secure SKE, where we only have to watermark the PKE scheme. We refer the readers to the construction and discussions of these examples in \ifllncs\else \Cref{sec:naor_yung_wm}, \ref{sec:wm_cca2_pke_ibe}, \ref{sec:wm_wprp_from_wprf},\ref{sec:wm_ccahybrid_enc}, \ref{sec:wm_fe_from_abe}. \fi

\ifllncs
\else
\section{Organizations}

\fi

\section{Definitions: General Cryptographic Primitive and Watermarking-Compatible Constructions}

\subsection{General Cryptographic Primitive}
\label{sec:general_crypto_primitive}

In this section, we present syntax and definitions for a general cryptographic primitive.
The notations and definitions formalized in this section will assist the demonstration of our generic watermarking framework.

\paragraph{General Cryptographic Primitive Syntax}
A cryptographic primitive $P = (\KeyGen, \SecAlg,\PubAlg)$ consists of the following algorithms:
\begin{itemize}
    \item $\KeyGen(1^\lambda) \to (\sk, \pk)$: is a (randomized) algorithm that takes a security parameter $\lambda$ and interacts with an adversary $\A$: $(\sk, \pk) \leftarrow (\A \Leftrightarrow \KeyGen(1^\lambda
))$ where $\sk$ is some secret information unknown to $\A$, and $\A$ can get some public information $\pk$ from
the interaction. 
    
    \item $\SecAlg(\sk, \pk, x \in \sX) \to y \in \sY_s:$ a secret-evaluation algorithm that takes in the secret information $\sk$, public information $\pk$ and some input $x$ from input space $\sX$, and output a value $y \in \sY_s$. 

    \item $\PubAlg(\pk, x \in \sX) \to y \in \sY_p:$ 
     a public-evaluation algorithm that takes in the public information $\pk$ and some input $x$ from input space $\sX$, and output a value $y \in \sY_p$.

\end{itemize}
More generally, the $\KeyGen$ algorithm can generate several secret keys $\sk_1, \cdots, \sk_\ell$ for some polynomial $\ell$. There will be $\ell$
different secret algorithms $\{\SecAlg_i(\sk_i, \pk, x \in \sX_{s,i})\}_{i \in [\ell]}$ and (some constant) $m$ different public algorithms $\{\PubAlg_j( \pk, x \in \sX_{p,j})\}_{j \in [m]}$. For example, an identity-based encryption scheme can have two decryption algorithms, one using the master secret key, the other with a secret key embedded with an identity.

However, without loss of generality, we will make two simplifications: 
\begin{itemize}
    \item All algorithms have their input space padded to be the same length as $\sX$; 
    \item We view all secret algorithms $\{\SecAlg_i(\sk_i, \pk, x \in \sX_{s,i})\}_{i \in [\ell]}$ as one algorithm $\SecAlg(\sk,\pk, \cdot)$ that will take in an index $i \in [\ell]$ to decide which mode to use, similarly for public algorithms. 
    But we will assume such an index specification to be implicit and omit the use of indices in the algorithm to avoid an overflow of letters.
\end{itemize}

Occasionally, we denote all the algorithms $(\{\SecAlg_i(\sk_i, \pk, x \in \sX_{s,i})\}_{i \in [\ell]}, \{\PubAlg_j( \pk, x \in \sX_{p,j})\}_{j \in [m]})$ as a combined functionality hardcoded with the corresponding keys $\eval(\pk, \sk)$. We call it $\eval$ for short. 

\paragraph{Correctness for Predicate $F_R$} 

Before going into the correctness property, we first define a notion important to our generic definition:
\begin{definition}
[Predicate]
\label{def:predicate}
    A predicate $F(C, x, z_1, \cdots , z_k, r)$ is a binary outcome function that runs a program $C$ on a some input $x$ to get output $y$, and outputs 0/1 depending on whether $(x, y, z_1, \cdots , z_k, r) \in R$
for some binary relation defined by $R$. The randomness of input $x$, program $C$ 
 both depend on randomness $r$.
. $z_1, · · · , z_k$ are auxiliary inputs that specify the
relation.
\end{definition}

A correctness property of a primitive $P$ with respect to predicate $F_R$ says that 
\begin{definition}[Correctness for Predicate $F_R$]
\label{def:correctness_p}
    $P$ satisfies correctness if there exists some function $\epsilon = \epsilon(\lambda) \in [0,1]$ so that for all $\lambda \in \N, x \in \sX$ \footnote{For a cryptographic primitive $P$, there can be many different correct properties, each defined with respect to a different predicate.}:
$$ \Pr_{r \gets D_r, (\sk,\pk)\gets \KeyGen(1^\lambda)}[F_R(\eval, \pk, \sk, x,r) = 1] \geq 1-\epsilon.$$
The randomness used in checking the predicate is sampled from a distribution $D_r$. We can simply take $D_r$ to be the uniform distribution, which can be mapped to any distribution we need when computing the predicate.

All $\epsilon$ within the scope of this work is negligible in $\lambda$.
\end{definition}


\begin{remark}
    To further explain the above correctness property, we consider the following (abstract) example.
    Given $\eval  = (\{\SecAlg_i(\sk_i, \pk, x \in \sX)\}_{i \in [\ell]}, \{\PubAlg_j( \pk, x \in \sX \}_{j \in [m]})$ , $F_R$ samples input $x \in \sX$ using the first part of $r$; then it runs algorithm (supposing randomized, using part 2 of string $r$) $\SecAlg_1(\sk_1,\pk,x)$ to give some outcome $y_1$; then runs (supposing deterministic) $\PubAlg_1(\pk_1, y_1)$ to give outcome $y_2$; check if $(x, y_1, y_2) \in R$; output 1 if yes, 0 otherwise. 

    A concrete example is a public key encryption scheme: the predicate is encrypting a message $x$ using $r$ and then decrypting it  to check if one can recover the original message. 
\end{remark}


\paragraph{Game-based Security}
A security property of the cryptographic primitive $P$ is described by a  interactive procedure $G_P$ between a challenger and adversary $\A$. 

$G_P$ will involve $\KeyGen, \SecAlg, \PubAlg$ as its subroutines. $G_P$ outputs a bit 1 if the game is not aborted and a certain condition has been met at the end of the game; else it outputs 0.

The security of a primitive $P$ with respect to $G_P$ says: For all $\lambda \in \N$, there exists some function $\eta = \eta(\lambda)$ such that for all \emph{admissible} $\A$ , there exists a function 
$\negl(\lambda)$:
$$ \Pr[G_P(1^\lambda, \A) = 1  ] \leq \eta + \negl(\lambda)$$

where $\eta$ is the trivial probability for any admissible $\A$ to make $G_P$ output 1 and the probability is over the randomness used in $G_P$. 

\paragraph{2-Stage Game-based Security}

To be compatible with the context of watermarking, we will  view the game $G_P$ as a 2-stage security game. 2-stage game-based security is a central notion that connects a prmitive's "plain security" to its watermarking security.

The $\KeyGen$ procedure and a first part of the interactions between the challenger and $\A$ are taken out of the original game and executed as a first stage $G_P^1$. 
Then, $G_P^2$ denotes the rest of the game and will take in parameters $(\sk,\pk)$ generated in stage 1 as well as some auxiliary input.

\begin{definition}[2-Stage Game-based Security]
\label{def:2-stage_game}
    A security property of the cryptographic primitive $P$ is described by a stage-1 (possibly interactive) key generation procedure a challenger and adversary $\A_1$ followed by a fixed-parameter game $G_P^2$ between a challenger and adversary $\A_2$.

We denote $G_P^1(1^\lambda, \A_1 )$ as the first stage adversary $\A_1$ interacting a challenger which runs the $\KeyGen(1^\lambda), \\ \SecAlg, \PubAlg$ algorithm; together they output a key pair $(\sk,\pk)$ and some auxiliary parameters $\aux$ which will be later used in the game $G_P$. $\A_1$ only gets $\pk, \aux$ but may make arbitrary polynomial number of admissible queries to oracles provided by the challenger during the interaction.

The stage-2 game challenger $G_P^2$ is parametrized by inputs $\sk,\pk,\aux$ generated during stage 1 and will involve $\SecAlg, \PubAlg$ as its subroutines. 
Stage-2 adversary $\A_2$ gets an arbitrary polynomial size state $\st$ from stage 1 $\A_1$.
$G_P^2$ outputs a bit 1 if the game is not aborted and a certain condition has been met at the end of the game; else it outputs 0.


The security of a primitive $P$ with respect to $G_P$ says: there exists some function $\eta = \eta(\lambda)$ such that for all \emph{admissible} $\A = (\A_1, \A_2)$ and any non-negligible $\gamma = \gamma(\lambda)$, there exists a function 
$\negl(\lambda)$  for all $\lambda \in \N$:
$$ \Pr[\A_2(\st) \text{ is $\gamma$-good in } G_P^2(\sk,\pk, \aux, \cdot) : \{(\sk,\pk), \aux, \st\} \gets G_P^1(\A_1, 1^\lambda) ] \leq \negl(\lambda)$$
where $\A_2$ is said to be $\gamma$-good if:
$$\Pr[G_P^2(\sk,\pk, \aux, \A_2) = 1 ] \geq \eta + \gamma$$
where $\eta$ is the trivial probability for any admissible $\A$ to make $G_P$ output 1 and the probability is over the randomness used in $\KeyGen$ and by $G_P$. 
\end{definition}

\begin{remark}
While some readers may find the introduction of parameter $\gamma$ in the above 2-stage game confounding, we would like to
    make a note that the above two definitions are essentially equivalent. 
    
    An alternative way of stating the 2-stage game security would be:
    there exists some function $\eta = \eta(\lambda)$ such that for all \emph{admissible} $\A = (\A_1, \A_2)$, there exists two negligible  functions 
$\negl_1(\lambda), \negl_2(\lambda)$ such that  for all $\lambda \in \N$:
$$ \Pr\left[\Pr[G_P^2(\sk,\pk, \aux, \A_2(\st)) = 1 ] \leq \eta + \negl_1(\lambda): \{(\sk,\pk), \aux, \st\} \gets G_P^1(\A_1, 1^\lambda) \right] \geq 1- \negl_2(\lambda)$$
\end{remark}

\begin{remark}[Division into a 2-stage Game]
    How we divide a security game into two stages depends on the application and suppose that the game $G_P$ has different stages by its definition, the way we divide stage-1 and 2 is usually not the same as in the original definition of $G_P$. We will see examples later.

    In most settings, stage-1 game $G_P^1$ only involves running the key generation $\KeyGen$ and letting $\A_1$ make some queries. In a few special settings, $\A_1$ needs to commit to some challenge messages which will be put into $\aux$ and be part of the input to stage 2.
    Examples include the challenge attribute/identity in an attribute/identity-based encryption, because the queries made in the first stage need to go through the "admissibility" check that depends on $\A$'s choice of challenge messages.
\end{remark}

\begin{remark}
    In the rest of this work, it is usually clear from the context which stage of $G_P$ we are refering to because $G_P^2$ will take in paramters $\sk,\pk,\aux$ while the entire game $G_P$ takes in only security parameter $1^\lambda$. Plus $G_P^1$ is seldomly used explicitly in our language. Occasionally we will omit the superscript and slightly abuse the notation to denote $G_P^2$ as $G_P(\sk,\pk, \aux, \cdot)$
\end{remark}

For our convenience in later notations, we also give the following simple definition:
\begin{definition}[Stage-2 Game View]
\label{def:game_view}
    The stage-2 $G_P^2$ can also output a view $\view(\sk,\pk, \aux, \A_2)$ that is the transcript of interaction of $\A_2$ with the challenger in $G_P^2$, including the final output. 
\end{definition}


\begin{remark}
We make a few notes on the scope of security games we consider in this work:
\begin{enumerate}
    \item  We mainly focus on game-based (and falsifiable) security notions within the scope of this work. We need the challenger in the security to be efficient for the watermarking construction to be efficient.


    \item All the admissible adversary need to be PPT. We do not consider security models other than polynomially bounded in time/circuit size (such as bounded storage).

\end{enumerate}

\end{remark}

\subsection{Watermarking-Compatible Construction of Cryptographic Primitive }

\label{sec:watermarking_compatible_construction}

In this section, we characterize what type of black-box cryptographic constructions are watermarking compatible.
As we will see in later sections,
watermarking compatible constructions allow us to give a watermarking scheme for the target primitive constructed, black-boxly from the watermarking schemes of the building 
blocks.

\jiahui{add more explanation at the beginning of this section to help the readers undertsnad these two different reductions}

\paragraph{Outline and Intuition}
Let use denote $P$ as the target, or outcome primitive built in a construction. Let $P_i, i \in [k]$ denote each underlying primitive we use to build $P$. 

Intuitively, one expects to watermark all underlying building blocks to ensure watermarking security of the target primitive $P$.
As discussed in our technical overview section "More advanced watermarking constructions with unwatermarked building blocks", many constructions possess building blocks of primitives that do \emph{not} need to be watermarkable to ensure the watermarkability of the final target primitive.


Here, we discuss the matter in more details:
we need the partitioning for primitives in $\{P_i\}_{i \in [k]}$ into $S$ and $\bar S$. These 
primitives in these two sets will play different roles when we construct the watermarking scheme for the target primitive $P$: the primitives in set $S$ need to be watermarkable and those in set $\bar S$ do not need to be watermarkable. 
The reasoning is that the primitives in set $S$ will have their secret keys generated during the Key Generation of the target primitive $P$, but primitives in set $\bar S$ will only have their secret keys generated freshly during every run of $\SecAlg$ or $\PubAlg$ algorithm of the target primitive $P$.  

A simple example is CCA2-secure PKE scheme based on IBE and one-time signatures, which we discussed in the technical overview. During the encryption algorithm, one generates fresh one-time signature keys and the message is encrypted with the one-time signature's verification key as the identity.
Therefore, the one-time signature in the above construction belongs to set $\bar S$ because its keys are only generated during the encryption algorithm.

In the watermarking construction for $P$,  the $\KeyGen, \PubAlg, \SecAlg$ algorithms will follow from the plain construction for $P$ from primitives in set $S$ and $\bar S$. Therefore, we cannot watermark the keys for the primitives in set $\bar S$ because their keys are not generated when we give out the watermarked key (which only contains keys in set $S$) to a user. 
We therefore distinguish these two sets in our presentation for watermarking-compatible construction in this section and watermarkable implementation in \Cref{sec:watermark_composition}. 

Moreover, looking forward: we will observe that we indeed do \emph{not} need watermarking security (i.e. unremovability) for the primitives in set $\bar S$, to achieve watermarking security for the target primitive P.  We only need to rely on their  “plain security” (e.g. unforgeability for a signature scheme, IND-CPA-security for an encryption scheme). 



\paragraph{Notations}
Now we give formalization for the above outline.

Suppose a cryptographic primitive $P$ is constructed black-boxly from primitives $P_1, P_2, \cdots, P_k$ in the following way, where $P_i, P_j, i \neq j$ are allowed to be the same primitive. 

\begin{itemize}
    \item Let $\cS$ be some fixed subset of $[k]$ defined in the construction. $\cS$ specifies two ways of using primitive $P_i$. The major difference is:
for primitives $P_i, i \in \cS$, the algorithm $\KeyGen(1^\lambda)$ will compute $P_i.\KeyGen(1^\lambda)$ and the secret keys $P_i.\sk$ generated will be used in the secret evaluation algorithm $\wsecEval$.

Without loss of generality, we let the first $|\cS|$ number of $P_i$'s be those corresponding to the set $\cS$.

\item For primitives $P_i, i \notin \cS$, the algorithm $\KeyGen(1^\lambda)$ will not compute $P_i.\KeyGen(1^\lambda)$. They  will either (1) have their keys generated freshly  upon every run of $\SecAlg$ or $\PubAlg$ (2) will have a reference string (which can be viewed as a public key) generated during $\KeyGen(1^\lambda)$, but have no secret keys or their secret keys will not be used in the $\SecAlg$ algorithm of the target primitive $P$.

\item For the sake of formality, we make a further division of the set $\bar S$ into $T_S, T_P, T_K$. They perform slightly different functionalities in the construction.

\begin{itemize}
    \item Let $\cT_S \in [k]$ denote the set of indices $i$ where $P_i$'s keys will be generated during the secret evaluation algorithm $\SecAlg$;

    \item Let $\cT_P \in [k]$ denote the set of indices $i$ where $P_i$'s keys will be generated during the public evaluaton algorithm $\PubAlg$;

    \item  Let $\cT_K$ denote the set of indices $i$ where $P_i$ have no secret keys/secret keys are not used in $\SecAlg$ and will have a public reference string (which can be viewed as a public key) generated during $\KeyGen(1^\lambda)$.
    
    \footnote{As we will see, primitives in $\cT_k$ will have their secret keys (called trapdoors $\td$ here) used only in the security proofs. 
    
    Looking forward: the sets $\cT_S, \cT_P$ will play the same role in the watermarking reduction, and the set $\cT_K$ will incur a slightly different argument but will essentially play the same role as the other primitives in set $\bar{S}$. That is, only their plain security is required to achieve the watermarking security of the target primitive $P$.}.
\end{itemize}

\jiahui{add more explanation on this set}

\end{itemize}

\paragraph{Watermarking-Compatible Construction Syntax}

\begin{description}
      \item $\KeyGen(1^\lambda) \to (\sk, \pk)$: 
      \begin{enumerate}
        \item compute $(\sk_i, \pk_i) \gets P_i.\KeyGen(1^\lambda)$ for all $i \in \cS$; compute $(\pk_i, \td_i) \gets P_i.\KeyGen(1^\lambda)$ for all $i \in \cT_K$. 
        \item output $\sk = (\{\sk_{i}\}_{i \in \cS}); \pk = (\{\pk_{i}\}_{i \in \cS \cup \cT_K})$.
    \end{enumerate}
       \item $\SecAlg(\sk, \pk, x \in \sX) \to y \in \sY_s:$ is an algorithm that 
       \begin{enumerate}
         \item   uses $P_i.\SecAlg(\sk_i,\pk_i, \cdot), i \in \cS$ as subroutines.
         \item uses $P_i.\PubAlg(\pk_i, \cdot), i \in \cS \cup \cT_K$ as subroutines.
           \item computes $(\sk_j, \pk_j) \gets P_j.\KeyGen(1^\lambda)$ for some $j \in \cT_S$ (can include these $\pk_j$ generated as part of the output).
         
       \end{enumerate}

    \item $\PubAlg(\pk, x \in \sX) \to y \in \sY_p:$ 
     is an algorithm that  \begin{enumerate}
         \item uses $P_i.\PubAlg(\pk_i, \cdot), i \in \cS \cup \cT_K$ as subroutines.
           \item computes $(\sk_j, \pk_j) \gets P_j.\KeyGen(1^\lambda)$ for all $j \in \cT_P$ (may include these $\{\pk_j\}_j$ generated as part of the output).       
       \end{enumerate}

\end{description}

We say the construction is watermarking compatible if the above construction of $P$ satisfies:
\begin{enumerate}

    \item \textbf{Correctness of Construction  }: the above construction of $P$ satisfies a correctness property defined by predicate $F_R$.
    Moreover, the proof of this correctness property follows from correctness properties of $P_1, \cdots, P_k$ for predicates $F_{R_i}$ respectively.

    \item \textbf{ Security of Construction}: The above construction satisfies security defined with game $G_P$. The security
    can be proved from the security properties of $P_1, \cdots, P_k$ with respect to some game $G_{P_1}, \cdots, G_{P_k}$

    \item \textbf{ Watermarking Compatible Reduction} The above security proof is of the following format, which we call \emph{Watermarking-Compatible Reduction}, shown below.
\end{enumerate}


\jiahui{Use the "leakage " in the stage 1?}

\jiahui{Give more intuition on the following reductions}

Among the above 3 properties,
the first two are natural properties that come with any black-box cryptographic constructions with provable security and correctness. 
We will focus on discussing property 3.

\paragraph{Watermarking Compatible Reductions}

On a high-level, a watermarking-compatible reduction is essentially a natural security reduction, except with the following feature: we view both the the security game $G_P$ for the target security primitive $P$ and the security game $G_{P_i}$ for the input primitive $P_i$ as two-stage games defined in \Cref{def:2-stage_game}.
The supposed adversary  $\A$ for $G_P$ will just be any usual PPT adversary. But we restrict the reduction in stage 2 to be "oblivious" about the queries made by $\A$ in stage 1: that is, it cannot pass on any queries made by $\A$ to the stage 2 reduction but we should nevertheless make the reduction go through successfully.

We divide our watermarking compatible reductions into two types. 
Type 1 reduction captures the reduction from breaking the security of $P$
to breaking the security of a primitive $P_i, i \in \cS$, i.e. the primitives that have their secret keys generated in the $\KeyGen$ algorithm and used in the $\SecAlg$ algorithm of $P$.

Type 2 reduction accordingly captures the reduction from breaking the security of $P$
to breaking the security of a primitive $P_i, i \in \bar{\cS}$.

\jiahui{added the difference discussion: both reductions need to be oblivious about the queries regarding the watermarked set. }
The main difference is in 
how they are used in proving the unremovability security in the composed watermarking scheme. In the actual watermarking unremovability, type 1 reduction is supposed to be oblivious about the queries made by the adversary and will eventually lead to breaking the unremovability of the underlying primitive; type 2 reduction operates similarly but will eventually lead to breaking the plain security of the underlying primitive.

\paragraph{Watermarking-Compatible Reduction Type 1}

First we consider reductions for primitives $P_i$ for $i \in \cS$.

    
\begin{enumerate}

    \item Consider a reduction to $P_i$ for the $j$-th $P_i \in \cS$; let $\A_i = (\A_i^1, \A_i^2)$ be the two-stage adversary for the 2-stage game of primitive $P_i$ defined in \Cref{def:2-stage_game}.

\jiahui{leakage based description?}
    \item $\A_i$ receives public key $\pk_i$ 

    \item $\A_i$ prepares $(\sk_j, \pk_j) \gets \KeyGen_j(1^\lambda) $ for all $j \neq i, j \in \cS$ and $\A_i$ sends all $(\{\pk_j\}_{j \in \cS})$
to $\A_1$.



    \item  $\A_i$ simulates the security game $G_P$ for $P$ with $\A$, while being the adversary in game $G_{P_i}$.

    \begin{itemize}
    \item    In stage 1, $\A_i^1$ provides the oracles needed for game $G_P^1$ and answer $\A^1$'s queries 
    as follows:
    \begin{itemize}
         \item  For any (admissible) queries in the game $G_P^1$, if $\sk_i$ is required to compute the answer for the query, $\A_i$ can use the oracles provided  in game $G_{P_i}^1$. 
       To answer the entire query, $\cA_i$ may finish the rest of the computation using 
       $\{\sk_j\}_{j\neq i}$.

        \item If only $\{\sk_j\}_{j\neq i}$ are needed, $\A_i$ answers the queries by itself.
    \end{itemize}


    
        \item During the interaction of stage-1 game, $\A^1$ and $\A_i^1$ generates an auxiliary information $\aux$ which is public to both of them.  
        
        \item Upon entering stage 2, $\A^1$ generates an arbitrarily polynomial-size state $\st$ and gives it to $\A^2$. 

     \item  $\A_i$ also enters its second stage $\A_i^2$. $\A_i^2$ also receives all of  $(\{\sk_j\}_{j \neq i}, \{\pk\}_{i \in \cS}, \aux)$ from 
     But $\A_i^2$ does \emph{not} obtain any 
     of  $\A^1$'s queries in stage 1. 

        \item $\A_i^2$  simulates the stage-2 game $G_P^2$ for $\A^2$,
      using the above strategy as $A_i^1$ uses. Note that the oracle operations done by $\A_i^2$  and the conditions on admissible queries \emph{cannot} be dependent on of $\A^1$'s queries in stage 1, because $\A_i^2$ cannot see these queries.

        \item $\A_i^2$ also records all of $\cA^2$'s queries. 

           \item In the challenge phase of $G_{P_i}^2$ 
           $\A_i^2$ receives a challenge input $\inp_i$ from the challenger.

         \item In $G_P$'s challenge phase, $\A_i^2$ samples some randomness $r$ and prepares a challenge input $\inp$ for $\A$ using $r$ and $\inp_i$. $\A_i^2$ sends $\inp_i$ to $\A_2$. \footnote{Note that $G_P$'s challenge phase may be before, after or concurrent with $G_{P_i}$'s challenge phase depending on the reduction. Similarly, the challenge $\A^2$ receives may be dependent on $\inp$.}

         \item $\A_i^2$ continues to simulate the oracles needed in  $G_P^2$ game for $\A^2$ if there are further query stages.  
    \end{itemize}

    \item 
    $\A_i$ computes an efficiently computable function $f_i(\out,r, \cQ, \inp_i)$ on input of $\A^2$'s final answer $\out$, the randomness used to prepare $\A^2$'s challenge input, $\A^2$'s queries $\cQ$ and $A_i$'s challenge input $\inp_i$, and gives it to the challenger of primitive $P_i$.

\end{enumerate}

\begin{remark}
    In the actual watermarking unremovability reduction, $\A_i^1$ will receive a watermarked secret key $\tilde{\sk_i}$ and simulate the queries required using the oracles in game $G_{P_i}^1$ with $\tilde{\sk_i}$ instead.

    More generically, we can also model the queries made by $\A^1$ in stage 1, related to the keys in the set $\cS$, $\{\sk_i\}_{i \in \cS}$ as some arbitrary polynomial size (admissible) leakage on these keys. In the plain security reduction, such leakage correspond to $\A^1$'s adaptive queries made to the oracles provided. In the watermarking unremovability game, it corresponds to admissible marking queries where we give out marked secret keys.
\end{remark}

\paragraph{Watermarking-Compatible Reduction Type 2}

Watermarking-compatible reduction type 2 operates essentially the same as the above reduction for $P_i$  but for some $i \notin \cS$.

The main difference is that in the actual watermarking unremovability security game, in order to simulate oracle queries for $G_{P}^1$, the stage 1 reduction $\A_i^1$ will actually need oracles in the game $G_{P_i}^1$. Naturally, this is because the primitives in set $\bar{\cS}$ are not watermarked in the watermarking composition and the corresponding reduction will not get a watermarked key from the challenger, but only the oracles provided in the plain security game of $G_{P_i}$. 
Even though this also means that the reduction $A_i^1$
gets to observe $\A^1$'s queries (using the oracles of the game  $G_{P_i}$) in the actual watermarking security game, and can make use of them in stage 2,  $A_i^1$ does not need to make any of such queries to the challenger since the keys of any primitive in set $\bar{\cS}$ will only be sampled online by $A_i^1$ itself. Meanwhile, $\A_i^2$ should still not inherit any queries related to the "watermarked" primitives in set $\cS$ since in the real watermarking unremovability game, $\A_i^2$ is not supposed to see them.


\begin{enumerate}
    
    \item Consider doing reduction to $P_i$ for some $P_i, i \notin \cS$; let $\A_i = (\A_i^1, \A_i^2)$ be the two-stage adversary for primitive $P_i$, where the stages are defined by \Cref{def:2-stage_game}.

    \item $\A_i$ receives public key $\pk_i$ from the challenger.

    \item $\A_i$ prepares $(\sk_\ell, \pk_\ell) \gets P_\ell.\KeyGen(1^\lambda) $ for all $\ell \in \cS$ 
$\A_i$ sends all $(\{\pk_\ell)$
to $\A_1$. 


    \item  $\A_i$ simulates the security game $G_P$ for $P$ with $\A$: 
    \begin{itemize}
        \item In stage-1 game $G_P^1$, for any (admissible) queries  if $\sk_\ell, \ell \in \cS$ is required to compute the answer for the query, $\A_i$ can answer the query since it possesses the keys  $\{\sk_\ell\}_{\ell \in \cS}$.


        The secret keys of primitives  $\{P_j\}_{j \notin \cS}$  are all sampled freshly upon every run of $\wsecEval(\sk, \pk,\cdot)$ (or sampled freshly in $\wpubEval(\pk, \cdot)$), which $\A_i^1$ can do it on its own.
        


        \item During the interaction of stage-1 game, $\A^1$ and $\A_i^1$ generates an auxiliary information $\aux$ which is public to both of them.

        \item Entering stage 2, $\A^1$ generates an arbitrarily polynomial-size state $\st$ and gives it to $\A^2$.  

    $\A_i$ also enters its second stage $\A_i^2$. $\A_i^2$ also receives all of  $(\{\sk_\ell\}_{\ell \in \cS}, \{\pk_\ell\}_{\ell \in \cS}, \aux)$  
   but  does \emph{not} obtain any of  $\A^1$'s queries involving the use of $\{\sk_\ell\}_{\ell \notin \cS}$. 

        \item $\A_i^2$ continues to provide the oracles needed for game $G_P^2$ and answer $\A^2(\st)$'s queries using the above strategy as $A_i^1$ uses. 

        \jiahui{anything that needs to pay attention here?}
        
        \item $\A_i^2$ also records all  of $\cA^2$'s queries. 

        \item In the challenge phase of $G_{P_i}$, 
        $\A_i^2$ receives a challenge input $\inp_i$ from the challenger.

         \item In $G_P^2$'s challenge phase, $\A_i^2$ samples some randomness $r$ and prepares a challenge input $\inp$ for $\A$ using $r$ and $\inp_i$. $\A_i^2$ sends $\inp_i$ to $\A_2$. 

    \item $\A_i^2$ continues to simulate the oracles needed in  $G_P^2$ game for $\A^2$ if there are further query stages
         
    \end{itemize}

    \item 
    $\A_i$ computes an efficiently computable function $f_i(\out,r, \cQ, \inp_i)$ on input of $\A^2$'s final answer $\out$, its queries $\cQ$ and $A_i$'s challenge input $\inp_i$ and secrret radomness $r$, and gives it to the challenger of primitive $P_i$.

\end{enumerate}

\paragraph{Properties of Watermarking-Compatible Reductions} 
We further give some properties of watermarking-compatible reduction. These properties come with any natural black-box security roof with hybrid argument and reductions. We present them here for the sake of convenience and use them as facts later.

\jiahui{change the following statement fo stage 2 game}
\begin{fact}[Reduction Property 1]
\label{fact:reduction_property1}
   A watermarking-compatible reduction guarantees that: if there exists some $\A$ such that the advantage of $\A^2$ winning the game $G_P$ is non-negligible, i.e. $\Pr[G_P(1^\lambda,A) = 1] \geq \eta+\gamma$ where $\eta$ is the trivial winning probability and $\gamma$ is non-negligible in $\lambda$, then there exists some $i \in [k]$  such that $\A_i$, using the above reduction strategy, wins $G_{P_i}$ with some non-negligible advantage $\gamma_i$.
\end{fact}


\jiahui{add the alternative def statement}
\begin{remark}
The above property follows naturally from any hybrid argument of proof.

    If using the language of the 2-stage security game: 
If there exists some adversary $\A = (\A^1, \A^2)$ such that for some non-negligible $\epsilon$ and some non-negligible $\gamma$, we have:
$$\Pr\left[\Pr[G_P^2(\sk,\pk, \aux, \A_2(\st)) = 1 ] \geq \eta + \gamma: \{(\sk,\pk), \aux, \st\} \gets G_P^1(\A_1, 1^\lambda) \right] \geq \epsilon$$
Then there exists some $i \in [k]$  such that $\A_i = (\A_i^1, \A_i^2)$, using the above reduction strategy, such that for some non-negligible $\gamma_i, \epsilon_i$, we have:
     $$\Pr\left[\Pr[G_{P_i}^2(\sk_i,\pk_i, \aux, \A_i^2(\st)) = 1 ] \geq \eta + \gamma_i: \{(\sk_i,\pk_i), \aux, \st\} \gets G_{P_i}^1(\A_i^1, 1^\lambda) \right] \geq \epsilon_i.$$
\end{remark}


\begin{fact}[Reduction Property 2]
\label{fact:reduction2}
    For all $i \in \cS$, once give the secret key $P_i.\sk$ in the clear, there exists a PPT algorithm $T_i$ that wins the security game $G_P$ with probability 1.

Additionally, such a $T_i$ can be used black-boxly by the watermarking-compatible reduction algorithm $\A_i$ to win the security game $G_{P_i}$ with noticeable probability.
\end{fact}

\jiahui{still need to capture what $\A_i$ does in order to "win" according to $\A$'s behavior. Here is my take:}

\paragraph{Reduction Function}
Recall that in the end of the reduction,  $\A_i$ computes an efficiently computable function $f_i(\out,r, \cQ, \inp_i)$ on input of $\A^2$'s final answer $\out$, the randomness used to prepare $\A^2$'s challenge input, $\A^2$'s queries $\cQ$ and $A_i$'s challenge input $\inp_i$, and gives it to the challenger of primitive $P_i$.

For convenience, we will refer to the function $f_i$ used by the reduction as \emph{reduction function}. 

\begin{remark}   
[Reduction Function]
\label{def:reduction_func}
    We refer to the  above function $f_i$ used by the reduction as \emph{reduction function}. 
\end{remark}



\begin{remark}[Special Primitives in $\cT_K$]
\label{remark:set_t_k_nizk}
The role of the primitives in the set $\cT_K$ may be confusing to the readers at this point. We make some further explanation. As we go into later sections and go into examples (\ifllncs\else\Cref{sec:naor_yung_wm}\fi), its role will become clear.

  More specifically, $\cT_k$ only contains the simulation-based primitives $P_\ell$ where
  \begin{enumerate}
      \item there exists an efficient simulator algorithm such that the output of $P_\ell.\PubAlg(P_\ell.\pk, \cdot)$ is indistinguishable from the output of the simulator $\Sim(P_\ell.\td, \cdot)$. 
  Their secret key (trapdoor $\td$) will only come up in the security proof for $P$.

  \item In the security proof for $P$, the rest of the hybrid arguments and reduction happen in a hybrid world where we have already invoked the above security for $P_\ell$.
  \end{enumerate}


The one example primitive in the set $\cT_K$ we use within this work is a NIZK scheme with a trapdoor. Its "keys" (the common reference string $\crs$, and trapdoor $\td$) are generated in the main key generation algorithm, but the trapdoor is not used in any of $\SecAlg, \PubAlg$, and only in the security proof.
    
\end{remark}

\section{Watermarking Composition Framework}

\subsection{Definition: Watermarkable Implementation of a Cryptographic Primitive }
\label{sec:def_watermarkable_primitive}

In this section, we will define what we call a "watermarkable implementation" of a cryptographic primitive $P$. We distinguish it from the usual naming ("watermarkable P") because of some differences in syntax and definition. But we will sometimes use watermarkable P in our work for convenience. Please note that all "watermarkable P" used in the technical part of this work refers to a watermarkable implementation of a cryptographic primitive $P$ defined below.

On a high-level, a watermarkable implementation of $P$ differs from most existing watermarking definition in two aspects:
\begin{enumerate}
    \item Unremovability security: A pirate circuit produced by the adversary in the unremovability security game is considered to function successfully as long as it can break the security game $G_P$.

    \item Extraction algorithm syntax: the extraction key used to extract a watermark is able to simulate the game $G_P$ for the underlying "plain security" of $P$.
\end{enumerate}

In more detail, a watermarkable implementation of a cryptographic primitive has the following syntax and properties.

\paragraph{Watermarkable Primitive Syntax} A watermarkable primitive $\WP$ for a cryptographic primitive $P$  with a security game $G_P$ consists of the following algorithms:

\begin{description}
    \item $\wmsetup(1^\lambda ) \to (\sk, \pk, \mk, \xk)$:  on security parameter, outputs a secret key $\sk$, public key $\pk$, marking key $\mk$, extraction key $\xk$.

    %

    \item $\wsecEval(\sk, \pk, x)$:  takes in the secret information $\sk$, public information $\pk$ and some input $x$ from input space $\sX$, and output a value $y \in \sY_s$.

    \item $\wpubEval(\pk, x \in \sX)$: takes in the public information $\pk$ and some input $x$ from input space $\sX$, and output a value $y \in \sY_p$.

    \item $\watermark(\mk, \sk, \tau \in \cM_{\watermark}) \to \sk_\tau$: takes in marking key $\mk$, a secret key $\sk$ and a message $\tau \in \cM_{\watermark}$; output a marked key $\sk_\tau$\footnote{In defintions in the watermarking literature, the output of $\watermark$ is the program $\wsecEval(\sk_\tau, \pk, \cdot)$. Since  conventionally by Kerkhoff's principle, the evaluation algorithm itself is public, giving out only the marked secret key is equivalent to giving out the program, we give out the key for convenience that come up later. There are watermarking schemes where the evaluation algorithm may be different when running on a marked key; in that case we can consider $\wsecEval$ to have two modes, one for unmarked keys one for marked keys.}.

    \item $\extract(\xk, \pk, \aux, C) \to \tau \in \cM_\watermark / \vec{\tau} \in \cM_\watermark^{k}/\bot$: on input extraction key $\sk$, public key $\pk$, auxiliary information $\aux$ and circuit $C$, outputs a mark message $\tau \in \cM_\tau$ or a tuple of marked messages $\vec{\tau} \in \cM^k_{\tau}$ where $k$ is a constant/small polynomial parameter related to the concrete watermarking construction. 
\end{description}

\begin{remark}
The $\extract$ algorithm is allowed to output a tuple of marks when working on adversarial programs, when the watermarkable implementation is one that comes from composing underlying watermarking implementations. More details to be discussed later in \Cref{sec:watermark_composition}.
\end{remark}

\begin{remark}
The extraction algorithm takes in the public key $\pk$ (which is not a limitation because it's public) and an auxiliary input $\aux$. As we will see in the concrete examples, depending on the security game of the primitives we watermark, we may need $\extract$ to take in some $\aux$ or may not.

    When it is clear from the context that there is no public key or auxiliary information for $\extract$ to take, we will omit them from the input parameters for notation cleanness.
\end{remark}

\begin{remark}
    There can be several different (constant number of) $\wsecEval$ and $\wpubEval$. 
    As aforementioned, we view all secret algorithms as one algorithm $\SecAlg(\sk,\pk, \cdot)$ that will take in an index $i \in [\ell]$ to decide which mode to use, similarly for public algorithms.
    We thus also use $\sk$ to denote $ (\{\sk_i\}_{i \in [\ell]}$.
\end{remark}
Let us denote $\eval = (\wsecEval, \wpubEval)$.

\paragraph{Correctness}
The construction should satisfy "unmarked" correctness for unmarked keys:
a watermarkable implementation of $P$ is said to be correct with respect to predicate $F_R$ if there exists a negligible function $\negl(\cdot)$ such that for all $\lambda \in \N, x \in \sX$:
$$ \Pr[F_R(\eval, \pk,\sk, x, r) = 1 : (\sk, \pk) \gets \wmsetup(1^\lambda), r \gets D_r] \geq 1-\negl(\lambda).$$
where the probability is taken over randomness used in $\wmsetup$ and  $D_r$. Recall that randomness $r$ is used in checking whether the predicate $F_R(\eval, \pk, \sk, x, r)$ is satisfied \Cref{def:predicate} and $D_r$ can be simply taken to be the uniform distribution.


\paragraph{Functionality Preserving vs. Exact Functionality Preserving}
We present both definitions of functionality-preserving for the sake of comprehensiveness because they have both appeared in watermarking literatures. 
Before \cite{goyal2019watermarking}, the watermarking literature mainly used exact functionality preserving only.
While it is not particularly vital to the contributions in this work, we would like to mention that exact functionality preserving is a stronger notion. For some of the underlying building blocks we use,  only constructions under the relaxed functionality-preserving definition are known, for example the watermarkable signature scheme in \cite{goyal2019watermarking}.

\paragraph{ Functionality Preserving}
The functionality-preserving property says: 
if for a predicate $F_R$, the underlying primitive $P$ satisfies he correctness in \Cref{def:correctness_p}, then there exists a negligible function $\negl(\lambda)$ ssuch that for all $\lambda \in \N, x \in \sX$:
$$ \Pr \left[F_R(\eval, \pk, \tildesk, x, r) =1 :  \begin{array}{cc} (\sk,\pk, \xk, \mk)  \gets \wmsetup(1^\lambda) \\  \widetilde{\sk} \gets \watermark(\mk, \sk, \tau), r \gets D_r   \end{array} \right] \geq 1 -\negl(\lambda).$$


\jiahui{added more explanation on why we have two definitions}

\paragraph{Exact Functionality Preserving}
The exact functionality preserving is a stronger property that: there exists a negligible function $\negl(\lambda)$ such that for all $\lambda \in \N, x \in \sX, \tau \in \cM_\tau$:
$$ \Pr \left[\wsecEval(\tildesk,\pk, x) = \SecAlg(\sk,\pk, x)  :  \begin{array}{cc} (\sk,\pk, \xk, \mk)  \gets \wmsetup(1^\lambda)  \\  \widetilde{\sk} \gets \watermark(\mk, \sk, \tau)   \end{array} \right] \geq 1 - \negl(\lambda).$$

\paragraph{Correctness of Extraction}
There exists a publicly known efficient algorithm $T$ and a negligible function $\negl(\lambda)$ such that for all $\lambda \in \N$, for all $\tau \in \cM_\tau$:
$$ \Pr\left[\extract(\xk, T(\sk_\tau, \cdot)) = \tau \,:\, \begin{array}{cc} (\sk,\pk, \xk, \mk) \gets \wmsetup(1^\lambda)  \\ \widetilde{\sk} \gets \watermark(\mk, \sk, \tau)  
\end{array} \right] \geq 1 - \negl(\lambda)$$

\paragraph{Security with Security Game $G_P$}
The watermarkable primitive $\WP$ satisfies the same security defined by the security game $G_P$ for $P$, except that the subroutines used in $G_P$, $\KeyGen,\SecAlg,\PubAlg$ are replaced with $\wmsetup,\wsecEval,\wpubEval$ respctively. The $\extract, \watermark$ algorithms and keys are ignored in the context of these games.


\paragraph{Watermarking Security Compatible with A Security Game}

Now we present the unremovability security of the watermarking implementation for primitive P.

Informally, the security guarantees that: any PPT adversary $\A$, when given the watermarked key(s), generates a pirate circuit $C^*$; if $C^*$ can win the security game $G_P$ with some non-negligible advantage, then we must be able to extract a (previously queried) watermark from $C^*$.

\begin{definition}[$\gamma$-Unremovability with Game $G_P$]
\label{def:unremovability_with_security_game}
   We say a watermarkable implementation of $P$ is $\gamma$-unremovable if:
   
   For every stateful $\gamma$-unremovable admissible PPT adversary $\A$, there exists a
negligible function $\negl(\cdot)$ such that for all $\lambda \in \N$, the following holds:
$$ \Pr\left[\extract(\xk, \pk,\aux, C^*) \notin \cQ: \begin{array}{cc}
    (\sk,\pk,\mk,\xk) \gets \wmsetup(1^\lambda)   \\
    (\aux, C^*) \gets \A^{\watermark(\mk,\sk,\cdot)}(1^\lambda, \pk)\\      
\end{array}  \right] \leq \negl(\lambda) $$
$\cQ$ denotes the set of f marks queried by $\A$ to the marking oracle $\watermark(\mk, \sk, \cdot)$, 
$G_P^1$ is the stage 1 of a security game $G_P$, as defined in \Cref{def:2-stage_game}. 

We call a PPT adversary $\A$ as $\gamma$-unremovable admissible if $\A$'s output $C^*$ is an admissible adversary in the \emph{stage-2} security game $G_P^2$ and is $\gamma$-good, where $\gamma$-good is defined as:
$$ \Pr[G_P^2(\sk,\pk, \aux, C^*) = 1] \geq \eta + \gamma$$ 
Here, $\eta$ is the trivial success probability for any admissible adversary in $G_P$.

The randomness over testing whether $C^*$ is $\gamma$-good is the randomness used 
answering any oracle queries from $C^*$ and preparing the challenge for $C^*$, in the stage-2 security game $G_P^2$.

To match our syntax, we denote $\extract(\xk, \pk,\aux, C^*) \notin \cQ$ to mean that:
\begin{enumerate}
    \item If $\extract(\xk, \pk,\aux, C^*)$ outputs a single mark $\tau \in \cM_\tau / \tau = \bot$, it is considered to be not in the query set $\cQ$ if and only if $\tau \notin \cQ$.

    \item If $\extract(\xk, \pk,\aux, C^*)$ outputs a tuple of marks $\vec{\tau} = (\tau_1, \cdots, \tau_k)$, where each $\tau_i \in \cM_\tau / \tau_i = \bot$, it is considered not in the query set $\cQ$ if and only if for all $i$,  $\tau_i \notin \cQ$.
\end{enumerate}
\end{definition}

\ifllncs
\else
\begin{remark}
The $\extract$ algorithm will only output a tuple of marks when used on adversarial circuits $C$. Because of composition, the adversary might be able to acquire a composed watermarked circuit where each "component" has a different watermark, but this does not formulate an attack as long as we can still extract any valid (previously queried) watermarks from the circuit.
We will discuss this aspect in more details in \Cref{sec:watermark_composition}.
\end{remark}

\begin{remark}[Discussions on the $\gamma$-Unremovability Definition]

We give the above definition with a parameter $\gamma$ and it helps illustrate how this definition corresponds to the 2-stage security game in \Cref{def:2-stage_game}: $\A$ corresponds to the "stage-1" adversary and the circuit $C^*$ corresponds to stage-2 adversary. 

If we want a watermarking scheme's $\gamma$ to be any non-negligible function (the strongest notion here) in terms of the security parameter, we
can simply require it to satisfy $\gamma$-unremovability for any non-negligible $\gamma$ (defined below as "strong unremovability").

In our paper, all our underlying watermarking building blocks used for composition satisfy $\gamma$-unremovability for any non-negligible gamma. So they are \emph{not} sensitive to $\gamma$. Our statements about the watermarking compiler also only consider this case (see \Cref{lem:unremovability_composition}).

Besides, we also inherit this $\gamma$-removability definition partially from \cite{goyal2019watermarking}.  In some other watermarking literatures, it may be also meaningful to consider a weaker security where gamma is not necessarily negligible. 
    
\end{remark}

\fi

\begin{definition}[Strong Unremovability]
\label{def: strong_unremovability}
    We say a watermarking implementation of $P$ satisfies strong unremovability if it satisfies $\gamma$-unremovability for any non-negligible $\gamma$.
\end{definition}

\ifllncs\else
\begin{remark}[Statefulness of the Adversarial Program]
Since the adversarial circuit in our security game is assumed to be interactive and stateful, one may ask how stateful it can be.

    This issue with stateful pirates is almost always present in the traitor tracing and watermarking literature. For this reason, the vast majority of works operate in a stateless program model. This requires an assumption that the pirate program can be reset by the tracer to its initial conditions, which makes sense if, say, the program is given as software. In our work, we allow the pirate program to be stateful, in the sense that we allow it to play an interactive security game. However, we always assume between runs of the game that the program is reset to its initial conditions. This notion of statefulness generalizes the existing watermarking/traitor tracing definitions from non-interactive games to interactive games, while also avoiding the issue of self-destruction. Like in the existing models, our model makes sense if, say, the program is given as software.

\end{remark}
\fi

We leave more remark on the unremovability definition \ifllncs
\else at the end of this section 
\Cref{remark:define_stages_in_game}. \fi




\jiahui{A "simulated view" property is not reuqired. Only the syntax is required?
}


\paragraph{Security Game Simulation Property of the Extraction Key}

We additionally require the extraction key $\extract$ to be able to simulate the (stage-2) security game $G_P^2$ for the primitive $P$ to be watermarked.

Consider the security game $G_P$ for the underlying primitive $P$: $G_P = (G_P^1, G_P^2)$ is an interactive game between a challenger and admissible $\A = (\A_1, \A_2)$. Recall that in the stage-2 \Cref{def:2-stage_game} $G_P^2$ can output a view $\view(\sk, \pk, \aux, \A_2)$ (\Cref{def:game_view}).

\begin{definition}[Extraction key simulation property]
\label{def:extraction_key_simulation}
    The extraction key simulation property of the extraction key says that: 
given $(\sk, \pk,  \aux)\gets G_P^1(\A_1, 1^\lambda)$ where we use $\wmsetup(1^\lambda)$ to obtain $(\sk, \pk, \xk)$ in $G_P^1$, there exists a PPT simulator $\Sim(\xk, \pk, \aux, \A_2)$ ( where $
\Sim$ interacts with $\A_2$ black-boxly)
that outputs a simulated view $\view_\Sim(\xk, \pk, \aux, \A_2)$,  such that for any $\lambda \in \N$, for any admissible $\A$ in $G_P$, the distributions $\view_{G_P^2} = \{\view
(\sk,\pk,\aux, \A_2)\}$ and $\view_{\Sim} = \{\view_\Sim
( \xk, \pk, \aux, \A_2)\}$ are perfectly/statistically/computationally indistinguishable. 
\end{definition}

\paragraph{Extraction Syntax: Simulation of Security Game $G_P$ inside $\extract$}
Following the above property, 
we require that the $\extract$ algorithm in a watermarkable implementation of a primitive $P$ follows a specific format: it simulates the stage-2 security game $G_P^2$ for any input circuit, while trying to extract a mark.

\begin{mdframed}[frametitle = $\extract$ Algorithm ]
\label{frame:extract_syntax}
On input $(\xk, \pk, \aux)$ and a program $C$:

\begin{itemize}
    \item 
     Run an algorithm $E$ where $E$ uses the following subroutine  (for possibly $\mathsf{poly}(\lambda)$) many times):

\begin{itemize}
    \item  Use $(\xk, \pk, \aux)$ to simulate the stage-2 game $G_P^2(\sk,\pk,\aux, C)$ for primitive $P$, running $C$ black-boxly by treating $C$ as the stage-2 adversary $\A_2$ defined in \Cref{def:2-stage_game}.

 \item If $C$ is non-interactive (i.e. $C$ does not make any queries or respond to interaction in the challenge phase of $G_P^2$), then $E$ samples challenge inputs on its own and runs $C$ on them.  
\end{itemize}

\item $E$ can take any of $C$'s outputs, including $C$'s queries made during the above simulated game, as inputs to compute the extraction algorithm's final output.

\item Output a watermark $\tau \in \cM_\tau$ or $\bot$ 

\end{itemize}

\end{mdframed}

\paragraph{Examples of Extraction Algorithm that Simulates a Security Game}

It is not hard to create an extraction algorithm with the above syntax and simulation capability. In fact, some existing works have already built watermarking schemes that have $\extract$ algorithm with such a format.
\begin{itemize}
    \item To enable the extraction key to simulate the security game, a naive solution in the
 private extraction case is to simply let $\xk$ contain the secret key $\sk$. 

 One example is the extraction algorithm in a watermarkable signature scheme (see \ifllncs\else \Cref{sec:wm_signatures} \fi). To answer the signature queries 
 for the pirate program, the extraction key is the signing key.
 
 \item In some other scenarios, we don't need the secret key to simulate the security game and one can even have public extraction.
 
 For example, a CPA secure PKE scheme.
 Another example is when we can sample from the input-output space of the evaluation oracles without having the actual key: in the weak PRF setting (with non-adaptive queries), we can simply answer its queries by sampling. In \cite{goyal2021beyond}, one can use indistinguishability obfuscation to build such a sampler, so that we have watermarkable weak PRF with public extraction.  
\end{itemize}

\jiahui{added a comment that statefulness and interaction are not required!}

\ifllncs
\else
\begin{remark}[Extraction of watermarks from honestly generated circuit]
\label{remark:extract_honest}
Note that even though our extraction algorithm simulates a possibly interactive game for any input program, it does not require or assume the input program to be interactive or even stateful.
As we have described, when the input circuit is not interactive, then the extraction algorithm simply samples challenge inputs and run $C$ on them to get $C$'s outputs.

For example, we can also correctly extract from an honestly watermarked circuit which is neither stateful nor interactive.

\end{remark}

\begin{remark}[Computational Simulation]
\label{remark:computational_sim}
    In most of our applications, the extraction key can simulate the security game's view with perfect indistinguishability. 
     But whether the simulation quality is perfect/statistical/computational does not affect our watermarking composition theorem and applications, as long as the difference is negligible.
     
     We will encounter computational simulation only when there are primitives in the set $\cT_K$: the extraction algorithm will simulate a "hybrid version" of the security game, after invoking the security of the primitive in  $\cT_K$, instead of the original security game. See more discussions \Cref{remark: comp_sim_2}.
     
     The one example provided in our paper is the use of NIZK proof with trapdoors in \Cref{sec:naor_yung_wm}. 
\end{remark}

\paragraph{Meaningfulness}
The meaningfulness property says that an honestly generated key should not have watermarks.

A watermarking implementation of primitive $P$ scheme satisfies the meaningfulness property if there exists a negligible function $\negl(\cdot)$ such that for all $\lambda \in \N$:
$$ \Pr\left[\extract(\xk, \sk) \neq \bot: \begin{array}{cc}
    (\sk,\pk,\mk,\xk) \gets \wmsetup(1^\lambda)   \\      
\end{array} \right] \leq \negl(\lambda)$$

\paragraph{Private Extraction vs. Public Extraction}
A watermakable implementation of 
a cryptographic primitive is a scheme with public-extraction if the extraction key $\xk$ can be made public while the above properties still hold. Otherwise it has private extraction.

\paragraph{Single-key vs. Collusion Resistant}
A watermakable implementation of 
a cryptographic primitive is collusion resistant if in the unremovability security game, $\A$ is allowed to query the marking oracle $\watermark(\mk,\sk, \cdot)$ for arbitrarily many times and $q$-bounded collusion resistance if it is allowed to query for a-priori bounded polynomial $q$-times. the scheme is single-key secure if $q = 1$.

\begin{remark}[Defining stages in the unremovability game]
\label{remark:define_stages_in_game}

As discussed briefly in \Cref{def:2-stage_game}:
  an example of subtlety when defining the staged unremovability game is: in the original security game for $P$, $\A$ needs to commit to some challenge messages that will be used by the challenger to prepare the final challenge for $\A$. The question is whether we make $\A$ commit to such challenge messages in stage 1, or let the stage-2 adversary, i.e. the circuit $C$ output by $\A$ choose their challenge messages later in stage 2.

  In summary, if admissibility of queries made in stage 1 will be dependent on the challenge, then we must make $\A$ commit to its own part on the challenge (as some auxiliary information $\aux$) before it produces $C$ , and we run $C$ with
these prefixed $\aux$ in stage-2 game $G_P^2(\pk,\xk, \aux, C)$.

  
    Let's take a "flexible" example, where we can let the stage-2 adversary $\A_2$ do such commitment, instead of forcing it to commit in stage 1. Accordingly, we can let the watermarking adversary's output program $C$ selects its own commitment when we run $C$: for example, in CPA/CCA encryption games, we can allow $\A$ to submit its challenge messages $(m_0, m_1)$ either together/before it outputs $C$ or output a $C$ that will choose its own $(m_0, m_1)$. This flexibility comes from the fact that  whether the queries made by $\A$ and $C$ are valid does \emph{not} depend on the choice of $(m_0, m_1)$. 
    
    A non-flexible example would be the challenge attribute/identity in ABE/IBE. We must let $\A$ choose its challenge attribute/identity before/at the same time with outputting $C$. $C$ is simulated in a second stage security game with this pre-selected attribute/identity as its input. 
    Otherwise $\A$ can easily cheat because the extraction algorithm(stage-2 game) does not get to check if the previous key generation queries made by $\A$ are valid with respect to challenge attribute/identity.
\end{remark}

\fi

\subsection{Watermarking Composition: Target Primitive from Input Primitives}
\label{sec:watermark_composition}
In this section, we show that if primtivive $P$ is built from $P_1,\cdots,P_k$ where the security can be shown via a watermarking-compatible reduction,  we can construct a watermarkable implementation of $P$, named $\WP$ to satisfy definition in \ref{sec:def_watermarkable_primitive}, from existing watermarkable implementations of $P_1,\cdots,P_k$  called $\WP_1, \WP_2, \cdots, \WP_k$, satisfying definition in \ref{sec:def_watermarkable_primitive}.
We will still refer to $P$ as the target primitive and $P_1,\cdots,P_k$ as input primitives.

\paragraph{Outline and Intuition}
On a high level, the watermarking scheme of the target primitive $P$ simply follows the construction of plain $P$ from the plain underlying building blocks in terms of evaluation algorithms $\SecAlg,\PubAlg$.
Correctness and functionality-preserving compose in a relatively natural way.

To mark a key of $P$, the marking algorithm concatenates the marked keys of all underlying $P_i$ that need to be marked. To extract a mark, we attempt to run the extraction algorithm of all underlying $P_i$ on the input circuit. If any valid mark is extracted, then the circuit is considered marked. 

In particular, the extraction algorithm will treat the circuit as an adversary in the stage-2 security game of $P$ and turns it black-boxly into a reduction for the security game of each underlying $P_i$, one by one and then run the underlying extraction algorithm of $P_i$ on it.

In order to remove a mark from a $P$'s key, the adversary $\A$ must remove all marks from each $P_i$'s key. Meanwhile, the pirate program still needs to win the security game of $P$, then we must be able to use to break the unremovability of at least one underlying $P_i$. In more generic scenarios, the pirate program made by $\A$ may not break any unremovability security of watermarkable building blocks, but get around the task of removing marks by breaking the security of some unwatermarked building blocks. By similar means, we can use the pirate program to build our reduction to the security 
of these unwatermarked building blocks. By the properties of the watermarking-compatible reduction (\Cref{sec:watermarking_compatible_construction}) that $P$'s construction satisfies, the above analysis is exhaustive.

\paragraph{Construction of $\WP$}
Similar to the description of construction in \Cref{sec:watermarking_compatible_construction}, we recall the following notations:

\begin{itemize}
    \item Let $\cS \subset [k]$ be a fixed set used in $P$'s construction from $P_1, \cdots, P_k$, where the primitives $P_i$ with $i \in \cS$'s keys will be generated in the $\KeyGen$ algorithm of $P$.
Without loss of generality, we let the first $|\cS|$ number of $P_i$'s be those corresponding to the set $\cS$.

\item Let $\cT_S \in [k]$ denote the set of indices $i$ where $P_i$'s keys will be generated during $\SecAlg$; Let $\cT_P \in [k]$ denote the set of indices $i$ where $P_i$'s keys will be generated during $\PubAlg$. 
Let $\cT_K$ denote the set of indices $i$ where $P_i$'s secret key (we call trapdoor $\td$) will be generated during $\KeyGen(1^\lambda)$ but will not be used in $P$'s algorithms. 

\item Without loss of generality, we sometimes assume a numbering on the primitives so that the first $|\cS|$ primitives are in the set $\cS$. 

\end{itemize}

\begin{description}
    \item $\wmsetup(1^\lambda) \to (\sk, \pk, \xk, \mk):$ 
    \begin{enumerate}
        \item compute $(\sk_i, \pk_i, \xk_i, \mk_i) \gets \WP_i.\wmsetup(1^\lambda)$ for all $i \in \cS$.
        \item compute $(\pk_\ell, \td_\ell) \gets P_\ell.\KeyGen(1^\lambda)$ for all $\ell \in \cT_K$;
        
        \item output $\sk = (\sk_{j_1}, \cdots, \sk_{|\cS|}); \pk = (\pk_1, \cdots, \pk_{|\cS|)},  \{\pk_\ell\}_{\ell \in \cT_k})); \\ \xk = (\xk_1, \cdots, \xk_{|\cS|}, \{\td_\ell\}_{\ell \in \cT_k}); \mk = (\mk_1, \cdots, \mk_{|\cS|})$
    \end{enumerate}

    \item $\wsecEval(\sk, \pk, x )$:
    \begin{enumerate}
        \item parse input $\sk = (\sk_{1}, \cdots, \sk_{|\cS|}); \pk = (\pk_1, \cdots, \pk_{|\cS|},  \{\pk_\ell\}_{\ell \in \cT_k}))$.

        \item $\wsecEval(\sk, \pk,  )$ is the same algorithm as $P.\SecAlg(\sk,\pk, \cdot)$ in the construction of $P$ from $P_1, \cdots, P_k$, except that $P_i.\SecAlg(\sk_i, \cdot)$ is replaced with $\WP_i.\wsecEval(\sk_i, \cdot)$ for $i \in \cS$.
        Overall, $\wsecEval(\sk, \pk,  )$ is an algorithm that:
        \begin{enumerate}
         \item   uses $\WP_i.\wsecEval(\sk_i,\pk_i, \cdot), i \in \cS$ as subroutines.
         \item uses $\WP_i.\wpubEval(\pk_i, \cdot), i \in \cS \cup \cT_k$ as subroutines.
           \item computes $(\sk_j, \pk_j) \gets P_j.\KeyGen(1^\lambda)$ for some $j \in \cT_S$ ( and may include these $\pk_j$ generated as part of the output). 
           
       \end{enumerate}
    \end{enumerate}

    \item $\wpubEval(\pk, x)$
        \begin{enumerate}
        \item parse input $\pk = (\pk_1, \cdots, \pk_{|\cS|}, \{\pk_\ell\}_{\ell \in \cT_k})$.

        \item $\wpubEval(\pk,  )$ is the same algorithm as $P.\PubAlg(\pk, \cdot)$ in the construction of $P$ from $P_1, \cdots, P_k$, except that $P_i.\PubAlg(\pk_i, \cdot)$ is replaced with $\WP_i.\wpubEval(\pk_i, \cdot)$ for $i \in \cS \cup \cT_k$.
        Overall, $\wpubEval(\sk, \pk, \cdot )$ is an algorithm that:
        \begin{enumerate}
         \item   uses $\WP_i.\wpubEval(\pk_i, \cdot), i \in \cS$ as subroutines.
         \item uses $\WP_i.\wpubEval(\pk_i, \cdot), i \in \cS$ as subroutines.
           \item computes $(\sk_j, \pk_j) \gets P_j.\KeyGen(1^\lambda)$ for some $j \in \cT_P$ ( and may include these $\pk_j$ generated as part of the output). 
           
       \end{enumerate}
    \end{enumerate}


    \item $\watermark(\mk, 
    \sk, \tau \in \cM_\watermark) \to \sk_\tau$ 
    \begin{enumerate}
        \item parse $\mk = (\mk_1, \cdots, \mk_{|\cS|}); \sk = (\sk_1, \cdots, \sk_{|\cS|}); \pk = (\pk_1, \cdots, \pk_{|\cS|},  \{\pk_\ell\}_{\ell \in \cT_k})$.

        \item Compute $\sk_{i,\tau} \gets \WP_i.\watermark(\mk_i, \sk_i, \tau)$ for all $i \in \cS$. 

        \item output $\tilde{\sk} = (\sk_{1,\tau}, \cdots, \sk_{|\cS|,\tau})$

    \end{enumerate}

    \item $\extract(\xk, \pk, \aux, C) \to \tau \in \cM_\watermark / \vec{\tau} \in \cM_\watermark^{|\cS|}/\bot:$ 
    \begin{enumerate}
        \item parse $\xk = (\xk_1, \cdots, \xk_{|\cS|}, \{\td_\ell\}_{\ell \in \cT_K}); \aux = (\aux_1, \cdots, \aux_{|\cS|}); \pk = (\pk_1, \cdots, \pk_{|\cS|},  \{\pk_\ell\}_{\ell \in \cT_k})$.

        \item Initialize an empty set $\vec{\tau}$.

        \item For each $i \in \cS$:
        \begin{enumerate}
        \item prepare the following circuit $C_i$ using black-box access to $C$.

        \begin{enumerate}
            \item $C_i$ uses $(\{\xk_j\}_{j \in \cS, j\neq i}, \{\td_\ell\}_{\ell \in \cT_K})$  and external queries to simulate security game $G_P^2(\sk, \pk, \aux, C)$ for $C$ \footnote{$C_i$  can also have $\xk_i$ if the watermarkable implementation  $\WP_i$ has a public extraction key. In this case, $C_i$ also does not need to make external queries.}:
            \begin{itemize}
                \item  For any (admissible) queries in the stage-2 game $G_P^2(\sk, \pk, \aux, C)$, if the oracles provided in security game $G_{P_i}(\sk_i, \pk_i, \aux_i, \cdot)$ is required to compute the answer for the query, $C_i$ can will make a query to an external challenger.
                To answer the entire query, $C_i$ may finish the rest of the computation using $(\{\xk_j\}_{j\neq i}, \{\td_\ell\}_{\ell \in \cT_k})$.

        \item If only $\{\xk_j\}_{j\neq i}$ are needed to answer a query, $C_i$ answers it by itself.
            \end{itemize}

            \item $C_i$ records queries from $C$ into a set $\cQ_C$.

            
            \item $C_i$ receives its challenge input $\inp_i$ from the interaction with an external challenger, samples a random string $r$, and prepares a challenge input $\inp$ for $C$ using $\inp_i$ and $r$. 

            \item If there is a query phase after the challenge phase, $C_i$ continues to simulate the oracles required using $\{\xk_j\}_{j\neq i}$  and external queries to a challenger.

            \item After $C$ makes its final output $\out$, $C_i$ computes the function $f_i(\out, r, \cQ_c, \inp)$  where $f_i$ is the reduction function (\Cref{def:reduction_func}) used in watermarking-compatible reduction of $P$ to $P_i$ . Output the result of this computation.
        \end{enumerate}
            
            
                \item compute $\tau_i/\bot \gets \WP_i.\extract(\xk_i, C_i)$.

                \item 
                add $\tau_i$(or $\bot$) to the tuple $\vec{\tau}$, and go to step 2 with $i := i+1$.

        \end{enumerate}

        \item Output
        \begin{itemize}
            \item   $\vec{\tau} = (\tau_1, \cdots, \tau_{|\cS|})$ if  $\exists i,j$ where $\tau_i \neq \tau_j$;

            \item else if $\tau_i = \tau_j = \tau$(or $\bot$) for all $i,j$, output $\tau$ (or $\bot$ resp.).

        \end{itemize}
      
    \end{enumerate}
\end{description}

\ifllncs
\else
\begin{remark}
We also give a more systematic description of Extract is using the language of watermarking-compatible reduction:

The $\extract$ algorithm:
\begin{enumerate}
    \item On input circuit $C$ and $\xk, \pk, \aux$;  for $i \in \cS$:

\begin{enumerate}

    \item Create program $C_i$ ($C_i$ has black-box access to $C$). $C_i$ treats $C$ as the stage-2 adversary $\A^2$ in the stage-2 game $G_P^2(\sk,\pk,\aux, \A_2)$. 
    
    \item $C$ acts as a stage-2 reduction $\A^2_i$ in the watermarking-compatible reduction from primitive P to $P_i$, with the difference that $C$ uses $\{\xk_j\}_{j\neq i}$ (instead of $\{\sk_j\}_{j\neq i}$) and external queries to challenger to simulate the game.
    

       \item compute $\tau/\bot \gets \WP_i.\extract(\xk_i, C_i$); add it to the extracted mark tuple. 
    \end{enumerate}
      \item Output the tuple of marks extracted.
\end{enumerate}
\end{remark}

\begin{remark}
As discussed previously, for the convenience of discussions in some later parts, our $\watermark$ algorithm outputs a marked key instead of circuits since the $\wsecEval$ algorithm itself is public. 

\end{remark}

\begin{remark}[Outputting a tuple of marks]
The $\extract$ algorithm will only output a tuple of marks (with $|\cS|$ number of them at most) when used on adversarial circuits $C$. For honestly generated circuit $C$, the extraction will always output only one marked message.

First this is clearly not a limitation because $|\cS|$ is the number of input primitives and usually a small constant/polynomial.
Second, the need for the extraction algorithm to output a tuple of marks is essential for the composability of watermarking schemes, otherwise we cannot capture the following trivial "attack": the adversary receives an honestly marked output primitive key, which consists of $k$ marked secret keys of the input primitives.
It can simply destroy all the keys except one and this leftover marked key suffices for breaking the security of $P$. Such an "attack" should not be considered as a successful attack because we can still extract a valid watermark from the entire program.
Therefore, we need the extraction to output a tuple of marks so that we can check if the adversary has succeeded at removing/replacing marks from each underlying functionality.


\end{remark}

\begin{remark}
\label{remark: comp_sim_2}
    As previously discussed in \Cref{remark:computational_sim}, when we have a primitive in the set $\cT_K$, 
    the $\extract$ algorithm needs to use the trapdoor(s) $\td$ of the primitives in set $\cT_K$ to simulate a hybrid version of the security game $G_P$. This simulation is supposedly computationally indistinguishable from the original game, by the security of the primitive(s) in $\cT_K$.

  In more detail, $\cT_k$ only contains primitives $P_\ell$ where
  \begin{enumerate}
      \item There exists an efficient simulator algorithm such that the output of $P_\ell.\PubAlg(P_\ell.\pk, \cdot)$ is indistinguishable from the output of the simulator $\Sim(P_\ell.\td, \cdot)$. 
  Their secret key (trapdoor $\td$) will only come up in the security proof for $P$.

  \item In the security proof for $P$, the rest of the hybrid arguments and reduction happen in a hybrid game where we have already invoked the above security for $P_\ell$.
  \end{enumerate}
Due to the above reasons, the $\extract$ algorithm has to simulate hybrid version of game $G_P$ where the security of $P_\ell \in \cT_K$ is already invoked, in order for the reduction (and thus extraction) for the rest of the primitives to go through.

Essentially, we simply need the plain security of primitives in the set $\cT_K$ to realize the watermarking security of $P$, like other primitives in $\Bar{\cS}$.

     Such examples are rare. The only example in this work is a NIZK proof system with a trapdoor used to simulate proofs (\Cref{sec:naor_yung_wm}).
\end{remark}

We now state our main theorem:

\begin{theorem}[Watermarking Composition]
    \label{thm:watermark_compiler}
    Suppose the input primitives in the set $\cS$, $\{P_i\}_{i \in \cS}$ have watermarkable implementations  $\{\WP_i\}_{i \in \cS}$ which satisfy the definitions in  \ref{sec:def_watermarkable_primitive}, the primitives outside the set  $\cS$ have secure constructions (defined by their corresponding correctness and security in \Cref{sec:general_crypto_primitive}),
    and the construction of $P$ is watermarking-compatible by the definitions in \ref{sec:watermarking_compatible_construction}, then the output watermarkable implementation $\WP$ for primitive $P$ will satisfy the properties in Definition \ref{sec:def_watermarkable_primitive}.
\end{theorem}

To prove \Cref{thm:watermark_compiler}, we show that all the properties from \Cref{sec:def_watermarkable_primitive} hold.

 \paragraph{Functionality-Preserving}
\begin{lemma}[Functionality-Preserving and Exact Functionality-Preserving]
    Suppose the above construction of $P$ from $P_1, \cdots, P_k$ satisfies correctness of construction, and constructions $\WP_1,\cdots, \WP_{|\cS|}$ satisfy the functionality-preserving properties in Definition \ref{sec:def_watermarkable_primitive}, then $\WP$ satisfies correctness defined by predicate $F_R$.

    If all $\WP_1,\cdots, \WP_{|\cS|}$ satisfy the exact functionality-preserving properties, then $\WP$ satisfies exact functionality-preserving.
\end{lemma}

\begin{proof}
 Since $\WP_1,\cdots, \WP_k$ satisfy the functionality-preserving properties in Definition \ref{sec:def_watermarkable_primitive}, the correctness properties of $P_1, \cdots, P_k$ are preserved by their watermarkable implementations. Since $P$'s correctness follows from the correctness properties of $P_1, \cdots, P_k$ and moreover, $\WP$'s evaluation algorithms are constructed the same way as they are in $P$ (even though the $\wmsetup(1^\lambda)$ algorithm differs from $\KeyGen(1^\lambda)$, we can ignore the watermarking-related components $\xk,\mk$ in this context). The constructions for primitives $P_i, i \notin \cS$ also do not affect functionality preserving since we use their plain, unwatermarked keys during evaluations.
 $\WP$'s correctness property thereby follows.

 It is easy to see that if all $\WP_i$ satisfies exact functionality preserving, then $\WP$ satisfies exact functionality preserving. 
\end{proof}

\paragraph{Correctness of Extraction}
Our correctness property is a computational correctness that relies on the unremovability of the input watermarkable implementations. Since our $\extract$ algorithm interacts with any input circuit black-boxly by simulating the security game $G_P$, the only way to extract a watermark is through the circuit's input-output behavior in this game. 
We use the guarantee that "if the circuit can win the game, then we can extract a watermark" to show correctness on a circuit embedded with an honestly watermarked key (which can win the game with probability 1 naturally).
Also note that we do not need our honestly watermarked circuit to be stateful/interactive as in \Cref{remark:extract_honest}.

\begin{lemma}[Correctness of Extraction]
    Suppose the above construction of $P$ from $P_1, \cdots, P_k$ has a watermarking-compatible reduction, and all watermarkable implementation of input primitives in set $\cS$, $\WP_1, \cdots, \WP_{\vert \cS \vert}$ satisfy strong unremovability 
     and the extraction syntax, then the watermarkable implementation of target primitive $\WP$ satisfies correctness of extraction. 
\end{lemma}
\begin{proof}
   Given an honestly marked key $\sk_\tau = (\sk_{1,\tau}, \cdots, \sk_{|\cS|,\tau}) \gets \WP.\watermark(\mk, \sk)$:
we can observe that there exists a public algorithm $T$ that once given $\sk_\tau$, can win the  game $G_P(\sk,\pk, \cdot)$ with probability $(1-\negl(\lambda))$. 
We can sample a random $i \gets \cS$ and design a circuit $T$ to work as follows: $T$ on any challenge input, uses a strategy associated with the key $\sk_i$ to compute the challenge given in game $G_{P}^2$.

Take the concrete example of the watermarkable implementation $\WP$ for a CCA2 encryption scheme, used in our technical overview section (\Cref{sec:tech_overview}). The watermarked key is a concatenation of a watermarked decryption key/weak PRF key $\sk_{1,\tau}$ and a watermarked MAC signing key $\sk_{2, \tau}$. Then the algorithm $T$ can choose the weak PRF key as its functionality: on any challenge ciphertext, 
$T$ simply decrypts the ciphertext using the watermarked key $\sk_{1,\tau}$ and outputs the correct answer.

  We then look into what happens when we compute $\extract(\xk, C)$: $C$ is treated as a stage-2 adversary in the watermarking-compatible reduction from the security of $P$ to the security of $P_i$ and $C_i$ with black box access to C is a stage-2 reduction algorithm. 
   $C$, for each $i \in \cS$. 
   We input $T(\sk_\tau, \cdot)$ into the extraction algorithm.
By the properties of  watermarking-compatible reduction, such a $T(\sk_\tau, \cdot)$  will be used to build $C_i$ that is $\gamma_i$-good for some noticeable $\gamma_i$ for each $i \in \cS$. When it comes to the index $i$ that we choose to be our functionality in $T$, then the created circuit $C_i$ using $T$ will naturally help win the security game $G_{P_i}^2$.

If we view the honest user as an adversary that makes one marking query on some message $\tau$, then by the unremovability security of $\WP_i$, we must have $\Pr[\WP_i.\extract(\xk_i, \pk_i, \aux, C_i = \tau] \geq  1-\negl(\lambda)$ for all $i$. Therefore $\Pr[\WP.\extract(\xk, \pk, \aux, (\sk_\tau, \cdot)) = \tau] \geq  1-\negl(\lambda)$.
Since extracting one valid watermark suffices, we can conclude the correctness of extraction.

Going back to the CCA2 encryption example: 
the circuit $T$ can answer any challenge by decrypting using the PRF key; therefore, when it is used by the circuit $C_1$ created in $\extract$ algorithm, its output helps $C_1$ distinguish between pseudorandom input and real random input. By the unremovability security of the watermarkable weak PRF, we must be able to extract the queried watermark $\tau$ from $C_1$.


    

\end{proof}

\paragraph{Security with $G_P$} We show that the (plain) security of $P$ with game $G_P$ holds.

\begin{lemma}[Security with $G_P$]
     Suppose the above construction of $P$ from $P_1, \cdots, P_k$ has a provable black-box security, and watermarkable implementations $\WP_1,\cdots, \WP_{|\cS|}$ as well as the (unwatermarked) constructions for $\{P_i\}_{i \notin \cS}$ satisfy the security with the security defined by game $G_{P_i}, i\in [k]$ respectively, then $\WP$ satisfies security defined by game $G_P$.
\end{lemma}

\begin{proof}
    Since $\WP_1,\cdots, \WP_k$ satisfy the security property with game $G_{P_i}$ in Definition \ref{sec:def_watermarkable_primitive}, the security properties of $P_1, \cdots, P_k$ are preserved by their watermarkable implementations. $P$'s security with respect to $G_P$ follows from the security properties of $P_1, \cdots, P_k$ with respect to $\{G_{P_i}\}_{i \in [k]}$. Moreover, $\WP$'s evaluation algorithms are constructed the same way as they are in $P$ and even though the $\wmsetup(1^\lambda)$ algorithm differs from $\KeyGen(1^\lambda)$, we can ignore the watermarking-related components $\xk,\mk$ in this context. Therefore, $\WP$'s security property follows from the security of $\WP_1, \cdots, \WP_k$ as $P$'s security property follows from the security of $P_1, \cdots, P_k$.
\end{proof}

\paragraph{Unremovability}
We now show that the construction satisfies unremovability if all building blocks satisfy unremovability (or plain security, for unwatermarked building blocks) and the properties on extraction key/algorithm defined in \cref{sec:def_watermarkable_primitive}.

\jiahui{1. change stage 1 to use watermarked key to answer queries
2. explain more about reductions for primitives in set $T_K$}

\begin{lemma}[Unremovability]
\label{lem:unremovability_composition}
    Suppose the above construction of $P$ from $P_1, \cdots, P_k$ has a watermarking-compatible reduction,   the watermarkable implementations  $\WP_1, \cdots, \WP_{|\cS|}$ satisfy strong unremovability (\Cref{def: strong_unremovability}), the constructions for $\{P_{j}\}_{j \notin \cS}$ satisfy  security defined by their game $\{G_{P_j}\}_{j \notin \cS}$ respectively,
    and the extraction algorithms of $\WP_1, \cdots, \WP_{|\cS|}$ satisfy the extraction syntax in \Cref{frame:extract_syntax},  then the watermarkable implementation $\WP$ satisfies strong unremovability. 
\end{lemma}
\begin{proof}
    If there exists an adversary that breaks $\gamma$-unremovability with game $G_P$ as defined in \Cref{def:unremovability_with_security_game}, then we  We have $\Pr[\extract(\xk, C) \notin \cQ \wedge C \text{ is $\gamma$-good}] \geq \epsilon$ for some non-negligible $\epsilon$
    where $C$ is $\gamma$-good in the sense that $\Pr[G_P^2(\sk, \pk, \aux, C) = 1] \geq \eta+\gamma$. Here the overall winning probability of $\A$ is 
    taken over the randomness used in $\KeyGen$ and the unremovability game; the probability for $G_P^2(\sk, \pk, \aux, C) = 1$ is taken over the randomness used in the stage-2 game $G_P^2$.

   By the design of our $\extract$ algorithm, it must be the case that, for all $i \in \cS$, $\WP_i.\Pr[\extract(\xk_i, C_i) \notin Q] \geq \epsilon$, where $C_i$ is the circuit made black-boxly from $C$ in $\WP.\extract$ algorithm for index $i$.

    Now we divide our analysis into cases.  First, we consider that there is no input primitive in the set $\cT_k$, i.e. no trapdoor $\td$ used in the extraction procedure and the extraction key can perfectly simulate the security game for $G_P$.
    
    For any  $\A$  producing a circuit $C$ such that $\Pr[\extract(\xk, \pk, \aux, C) \notin \cQ \wedge C \text{ is $\gamma$-good}] \geq \epsilon$ (where the probability is taken over the randomness in $\KeyGen$ and randomness used throughout the game), one of the following cases must hold:

   %
    \begin{itemize}
        \item \textbf{Case 1:
       For some non-negligible probability $\epsilon_i$, during the execution of $\extract(\xk, \pk, \aux, C)$, there exists some $\istar \in \cS$ such that $C_{\istar}$ is a $\gamma_\istar$-good program in $G_{P_i}^2(\sk_i, \pk_i, \aux_i, C_i)$ for some non-negligible $\gamma_\istar$.}




     Then, there exists a reduction to break the $\gamma_\istar$-unremovability of $\WP_{i^*}$ for some non-negligible $\gamma_\istar$, using $\A$.
    Let $\A_\istar$ be the reduction and adversary in $\gamma_\istar$-unremovability game of $\WP_\istar$.
    \begin{itemize}
        \item $\A_\istar$ receives $\pk_\istar$ from the challenger and samples $\{\sk_i, \pk_i, \xk_i, \mk_i\}_{i \in \cS, i \neq \istar}$ from $\WP_\istar.\wmsetup$. It gives 
        $\{ \pk_i\}_{i \in \cS, i \neq \istar}$ to $\A$.
        
        \item When $\A$ makes a marking query on some symbol $\tau$, $\A_\istar$ responds as follows:
        \begin{itemize}
            \item compute $\sk_{i,\tau} \gets \watermark(\mk_i, \sk_i, \tau)$ for $i \in \cS, i \neq \istar$.

            \item query the $\WP_{\istar}$ challenger on oracle $\watermark(\mk_\istar, \sk_\istar, \tau)$ to obtain $\sk_{\istar, \tau}$.

            \item Let $\sk_\tau  = \{\sk_{i, \tau}\}_{i \in \cS}$; send  $\sk_\tau$ to $\A$.
        \end{itemize}


        \item Next, $\A$ outputs circuit polynomial-size $C$.

        \item $\A_\istar$ creates the following circuit $C_\istar'$ using $C$  as a black-box:
        \begin{itemize}
        \item $C_\istar'$ is hardcoded with $\{\sk_i\}_{i \in \cS, i \neq \istar}, \{\pk_i\}_{i \in \cS}$.
        
        
            \item  When $C$ makes a query to oracles in game $G_P$, 
        $C_\istar'$ responds as follows:
        if $\sk_\istar$ is required to answer the query, $C_\istar$ queries the external challenger on the oracles provided in game $G_{P_\istar}$; otherwise, $C_\istar'$ answers the query using $\{\sk_i\}_{i \in \cS, i \neq \istar}$.

        \item $C_\istar'$ records all queries from $C$, denoted as $\cQ_C$.
      
        \item In the challenge phase, $C_\istar'$ receives challenge input $\inp_i$ from the interaction with an external challenger, samples a random string $r$, and prepares a challenge input $\inp$ for $C$ using $\inp_i$ and $r$. 
        
        \item When $C$ outputs its final answer $\out$, $C_\istar$ computes the function $f_\istar(\out, \cQ_C, \inp)$ according to the reduction function $f_\istar$ in the watermarking-compatible reduction from $P$ to $P_\istar$.
        \end{itemize}
\item In short, $C$ is treated as a stage adversary $\A_2$ in security game $G_P$; circuit $C_\istar$ is a stage-2 reduction algorithm in the watermarking-compatible reduction from target primitive $P$ to input primitive $P_\istar$, with black box access to $C$. 
  \end{itemize}

  By the design of the $\WP.\extract$ algorithm, it using extraction keys $\{\xk_i\}_{i \in \cS}$, will treat the input circuit as  an adversary in stage-2 game $G_P^2$ and simulate $G_P^2$ for it.
 
Given any adversarial circuit $C$ that is a $\gamma$-good adversary in the security game $G_P$, following our extraction algorithm syntax requirement, the extraction algorithm running at step $i = \istar$ will create a circuit $C_{\istar}$ that has exactly the same functionality as the circuit $C_\istar'$ created in the above reduction \footnote{Note that the external queries made by $C_\istar$ will be answered by $\WP_\istar.\extract(\xk_\istar, C_\istar)$ during its execution, where by our requirement on the extraction syntax, $\WP_\istar.\extract(\xk_\istar, C_\istar)$ will simulate the game $G_{P_\istar}$ for $C_\istar$, answering its queries. }. Therefore, if $C_\istar$ created by $\extract$ is a $\gamma_\istar$-good in breaking $G_{P_\istar}$, then so is $C_\istar'$ created by the above procedure; since we have $\WP_\istar.\Pr[\extract(\xk_\istar, C_\istar) \notin Q] \geq \epsilon$ and the set $Q$ is the same for $\A$ and $\A_\istar$, then we also have $\WP_i.\Pr[\extract(\xk_i, C_\istar') \notin Q] \geq \epsilon$.

\vspace{0.15in}

\item \textbf{Case 2: With overwhelming probability, during the execution of $\extract(\xk, \pk, \aux, C)$, there exists \emph{no} $i \in \cS$ such that $C_i$ is a $\gamma_i$-good adversary in $G_{P_i}^2(\sk_i, \pk_i, \aux_i, C_i)$ for some non-negligible $\gamma_i$.}

In this case, it must be that there exists some $\istar \notin \cS$ such that we can use the unremovability adversary $\A$ for $\WP$ to break the \emph{security} of $P_\istar$, i.e. winning game $G_{P_i}, i \notin \cS$ with some non-negligible advantage.

By the properties 
of watermarking-compatible reduction, if (with non-negligible probability), the circuit $C$ produced by $\A$  is a $\gamma$-good adversary in $G_P$, then there must exist some $i \in [k]$ such that using the watermarking-compatible reduction, one can create a $\gamma_i$-good adversary for some non-negligible $\gamma_i$. By the above analysis, $C_i$ supposedly performs exactly the reduction for $P_i$ and if there is no good $C_i$ for $i \in \cS$, then $C$ can only be used to build a reduction to break the security of some $P_i, i \notin \cS$.

The reduction is as follows.
Let $\A_\istar$ be the reduction and adversary in the security game $G_{P_\istar}$.
    \begin{itemize}
        \item $\A_\istar$ receives $\pk_\istar$ from the challenger and samples $\{\sk_i, \pk_i, \xk_i, \mk_i\}_{i \in \cS}$ from $\WP_\istar.\wmsetup$. It gives 
        $\{ \pk_i\}_{i \in \cS}$ to $\A$.
        
        \item When $\A$ makes a marking query on some symbol $\tau$, $\A_\istar$ answers the query on its own because it has all the $\{\sk_i,  \mk_i\}_{i \in \cS}$.

        \item Besides answering marking queries, $\A_\istar$ will treat $\A$ as a stage-1 adversary in the watermarking-compatible reduction Type 2 (\Cref{sec:watermarking_compatible_construction}):
        when $\A$ makes a query to oracles in game $G_P$, 
        $\A_\istar$ responds as follows:
        if $\sk_\istar$ is required to answer the query, $\A_\istar$ queries the challenger in $G_{P_\istar}$ security game; otherwise, $\A_\istar$ answers the query using $\{\sk_i\}_{i \in \cS}$. 

        For the secret keys of $\{\WP_i\}_{i \notin \cS, i \neq \istar}$, they are all sampled freshly upon every query of $\wsecEval(\sk, \pk,\cdot)$ (or sampled freshly in $\wpubEval(\pk, \cdot)$ which $\A$ can do it on its own).  

        \item Next, $\A$ outputs circuit polynomial-size $C$.

        \item $\A_\istar$ creates the following circuit $C^*$ using $C$  as a black-box:
        \begin{itemize}
        \item $C^*$ is hardcoded with $\{\sk_i, \pk_i\}_{i \in \cS}$.
        
        
            \item  When $C$ makes a query to oracles in game $G_P$, 
        $C^*$ responds as follows:
        if $\sk_\istar$ is required to answer the query, $C^*$ queries the external challenger on the oracles provided in game $G_{P_\istar}$; otherwise, $C^*$ answers the query using $\{\sk_i\}_{i \in \cS}$.

        \item $C^*$ records all queries from $C$, denoted as $\cQ_C$.
      
        \item In the challenge phase, $C^*$ receives challenge input $\inp_i$ from the interaction with an external challenger, samples a random string $r$, and prepares a challenge input $\inp$ for $C$ using $\inp_i$ and $r$. 
        
        \item When $C$ outputs its final answer $\out$, $C^*$ computes the function $f_\istar(\out, \cQ_C, \inp)$ according to the reduction function $f_\istar$ in the watermarking-compatible reduction from $P$ to $P_\istar$.
        \end{itemize}
\item In short, $C$ is treated as a stage-2 adversary $\A_2$ in security game $G_P$; circuit $C^*$ is a stage-2 reduction algorithm in the watermarking-compatible reduction Type 2 from target primitive $P$ to input primitive $P_\istar$, with black box access to $C$. By the property of the reduction, $C^*$ should win the security game $G_{P_\istar}$ with non-negligible advantage.
  \end{itemize}


Finally, we discuss the rare case when the extraction key simulates the security game for $C$ statistically/computationally close to a real game.

In the setting where $(\{\xk_i\}_{i \in \cS}, \{\td_\ell\}_{\ell \in \cT_K})$ can only simulate $G_P^2$ for $C$ statistically close by some negligible distance, then if $C$ is a $\gamma$-good circuit in a real $G_P^2(\sk,\pk, \aux, \cdot)$ game, then 
$C$ is a $(\gamma-\negl(\lambda))$-good
circuit in the game simulated by the $\extract$ algorithm and thus by the reduction.  Afterwards, we can go back to use $C$ in the above Case 1 and Case 2.

Coming to the computational setting, we will see that they will also be transformed into Case 1 or Case 2.

Note that within the scope of this work, this case corresponds precisely to having primitives in the set $\cT_k$ and the computational property comes from primitives in the set $\cT_k$. Recall that $\cT_k$ only contains the primitives $P_\ell$ where there exists a simulator algorithm such that the output of $P_\ell.\PubAlg(P_\ell.\pk, \cdot)$ is indistinguishable from the output of the simulator $\Sim(P_\ell.\td, \cdot)$ and in the security proof for $P$, the rest of the hybrid arguments and reduction happen in a hybrid world where we have already invoked the security for $P_\ell$. (see discussions in \Cref{remark:set_t_k_nizk}, \Cref{def:extraction_key_simulation} and \Cref{remark:computational_sim}, \Cref{remark: comp_sim_2}).

We provide one such example where 
we need to use a NIZK proof and provide simulated proofs to the adversarial circuits in the extraction algorithm. See the example in \Cref{sec:naor_yung_wm} for details.


First consider the case where we only have one primitive $P_\ell$ in set $\cT_K$ in our construction (this is the case in our example \Cref{sec:naor_yung_wm}).
By our premise, using $(\{\xk_i\}_{i \in \cS}, \td_\ell\})$ where $\td_\ell \gets P_\ell.\KeyGen(1^\lambda)$, one can simulate $G_P^2$ for $C$ computationally close by some negligible amount. If with some non-negligible probability $\epsilon$, we have 
 $C$ is a $\gamma$-good circuit in a real $G_P^2(\sk,\pk, \aux, \cdot)$ game. Then with $(\epsilon-\negl(\lambda))$ probability,
$C$ is a $(\gamma-\negl(\lambda))$-good
circuit in the game simulated by the $\extract$ algorithm and the reductions, otherwise there exists a PPT distinguisher that can use $C$ to distinguish between the real game and simulated one to break the security of a primitive in set $\cT_k$: a reduction $\cB$ can samples all $\sk,\pk,\xk, \mk$ on its own and receive the public parameters from the challenger of some $G_{P_\ell}$. The rest of simulation is the same as in Case 2 of the above analysis. At the end, $\cB$ tests if the circuit output by $\A$ is a $\gamma$-good one, if yes, output guess "real view; else output guess "simulated view".

If there are more than one$P_\ell \in \cT_K$, we can then use a hybrid to apply the above to each of them, following the hybrid order in the security proof for the plain security of $P$.

Therefore, by the security of $P_\ell$,  we have that $C$ is a $(\gamma-\negl(\lambda))$-good
circuit in the game simulated by the $\extract$ algorithm.
Since the drop in $C$'s advantage is negligible, we can go back to Case 1 and Case 2. 





  \jiahui{Be careful with probability distribution discussed}

 \end{itemize}
  

\end{proof}

\fi

\ifllncs
\else

\ifllncs
\section{Additional Properties of Watermarking Composition}
\label{sec:append_watermarking_property}

\else
\fi

\paragraph{Additional Properties}
Finally, we discuss some additional properties realized by the watermarking composition.

\paragraph{Private vs. Public Extraction} Firstly, we have a relatively straightforward observation that if all underlying constructions satisfy public extraction, then the target construction also does.

\begin{lemma}[Public Extraction]
  \label{lem:public_extract_compiler}
   Suppose a watermakable implementation $\WP$ is constructed from watermarkable implementations $\WP_1, \cdots, \WP_{|\cS|}$,  if $\WP_1, \cdots, \WP_{|\cS|}$ all have public extraction and there exists no primitive in the set $\cT_K$, then $\WP$ has public extraction.
\end{lemma}

\begin{proof}
    It is easy to observe that if all $\WP_i$ have public extraction, then all $\xk_i$ can be made public and the extraction algorithm of $\WP$ can be made public.
    During the extraction procedure, since the keys in the primitives in the set $\cT_S, \cT_P$ are sampled online, even a public procedure can sample them. 
    But if we have a primitive in the set $\cT_K$, then we need a trapdoor during extraction, so we have to rule out this case.
\end{proof}

\paragraph{Meaningfulness}
The meaningfulness of $\WP$ follows  from the meaningfulness of each $\WP_i$ 
construction naturally.

\paragraph{Collusion Resistance}
Finally, we show that collusion resistance also composes. Informally, suppose a watermakable implementation $\WP$ is constructed from watermarkable implementations $\WP_1, \cdots, \WP_{|\cS|}$, if all of $\WP_1, \cdots, \WP_{|\cS|}$ are collusion resistant, then $\WP$ is collusion resistant.
If at least one of them is $q$-bounded collusion resistant, then $\WP$ is $q$-collusion resistant  $q$ (by the smallest $q_i, i\in |\cS|$).

\begin{lemma}[Collusion Resistance]

  \label{lem:wm_collusion_resistant_compiler}

  Suppose the above construction of $P$ from $P_1, \cdots, P_k$ has a watermarking-compatible reduction,   the watermarkable implementations  $\WP_1, \cdots, \WP_{|\cS|}$ all satisfy strong unremovability (\Cref{def: strong_unremovability}) with $q$-collusion-resistance, the constructions for $\{P_{j}\}_{j \notin \cS}$ satisfy  security defined by their game $\{G_{P_j}\}_{j \notin \cS}$ respectively,
    and the extraction algorithms of $\WP_1, \cdots, \WP_{|\cS|}$ satisfy the extraction syntax in \Cref{frame:extract_syntax},  then he watermarkable implementation $\WP$ satisfies strong unremovability with $q$-collusion resistance.

\end{lemma}
\begin{proof}
    Since each marking query of $\WP$ requires querying each underlying marking oracle $\WP_i.\watermark$ once on the same marking message, the number of queries allowed for the reduction from $\WP$'s unremovability is bounded by the input primitive $\WP_i$'s query bound.
    If all underlying schemes are $q$-collusion resistant for some $q$, then all reductions can query $q$ times and the unremovability proof for \Cref{lem:unremovability_composition} stays the same.
    $\WP$ is also $q$-collusion resistant. 
    

\end{proof}

\paragraph{Extraction Simulation Property and Extraction Syntax}
Having this additional property in the composed  watermarkable implementation scheme allows us to further compose watermarking schemes from input schemes that come from composition themselves.

It is easy to see that the ability of the extraction algorithm to simulate the security game also composes.

We can observe that if all watermarkable implementations of  input  primitives satisfy the extraction syntax  (\Cref{frame:extract_syntax},and thus key simulation property \Cref{def:extraction_key_simulation}) , then the combined keys $(\{\xk_i\}_{i \in \cS}, \{\pk_i\}_{i \in \cS})$ can be used to simulate the game $G_P$ because all oracles used in $G_P$ can all be built from oracles used in $\{G_{P_i}\}_{i \in \cS}$ according to the watermarking-compatible reduction.

Since all the $\WP_i.\extract(\xk_i, \pk_i, \aux_i, \cdot)$ algorithm perfectly(statistically/computationally, resp.) simulates the security game for $G_{P_i}$ stage 2 (as a subroutine of extraction), when each $C_i$ created from input program $C$ is executed inside the procedure  $\WP_i.\extract(\xk_i, \pk_i, \aux_i, C_i)$: $C_i$ will simulate the game $G_P$ stage 2 for $C$ using external oracle queries when interacting with the algorithm $\WP_i.\extract(\xk_i, \pk_i, \aux_i, C_i)$ and the other extraction keys $\{\xk_i\}_{i \in \cS}$. The entire $\extract$ algorithm consists of subroutines where one simulates $G_P$ perfectly(statistically/computationally, resp.) stage 2 for $C$ when running $C$.

When we have input primitives in $\cT_k$, as we have discussed in several remarks: in the extract trapdoor(s) $\td$ of primitives in $\cT_k$ to simulate part of the security game, which results in a computationally close simulation from the original game. One concrete example of such primitive in $\cT_k$ is NIZK, which we will use in \Cref{sec:naor_yung_wm} and \Cref{sec:ddn91_wm}.

\fi

\ifllncs
\else
\section{Two Simple Examples: Watermarkable CPA and CCA-secure SKE}

In this section, we give two relatively simple examples through the following two steps: \begin{itemize}
    \item First we give their "plain" construction and show that their reduction is watermarking-compatible.
    \item Then give their watermarking implementation of the target primitive in the watermarking-compatible reduction language.
\end{itemize}
For the remaining examples we show in this work, we will integrate the above two discussions for brevity.

\subsection{Watermarkable CPA-secure SKE from Watermarkable weak PRF}
\label{sec:wm_cpa_ske}


In this section, we show a simple example of obtaining a watermarkable implementation of a CPA secure secret-key encryption scheme from a watermarkable implementation of weak PRF.
The security of watermarkable weak PRF satisfying our requirements in \Cref{sec:def_watermarkable_primitive} can be built from LWE (with private extraction) or iO (with public extraction) \cite{goyal2021beyond}. The following example was also discussed in \cite{goyal2021beyond}.

\subsection{Definition: Watermarkable Implementation of PRF }
\label{sec:def_wm_prf}
We give the functionality-preserving unremovability security definition for PRF and CPA-secure SKE. The other definitions (pseudorandomness, CPA-security, correctness, functionality preserving) follow naturally from the correctness and security defintions of PRF and SKE respectively.

A watermarkable implementation of a PRF scheme consists of the following algorithms
$(\wmsetup, \eval, \watermark, \extract)$ and satisfies the following properties:

\begin{description}
    \item $\wmsetup(1^\lambda) \to (\sk, \mk, \xk)$: on input security parameter, outputs a secret key $\sk$, marking key $\mk$, extraction key $\xk$.

      \item $ \eval(\sk, x \in \{0,1\}^\ell) \to y \in \{0,1\}^\ell$: a deterministic evaluation algorithm that on input $\sk, x$, outputs $y$.

      \item $\watermark(\mk, \sk,\tau)$: on marking key $\mk$, secret key $\sk$, a message $\tau$, output a marked key $\sk_\tau$.

      \item $\extract(\xk, C)$: on extraction key $\sk$ and program $C$, output a mark $\tau \in \cM_\tau$ or $\bot$.
\end{description}

\paragraph{Functionality-Preserving} A watermarkable PRF is functionality preserving if  then there exists a negligible function $\negl(\lambda)$ ssuch that for all $\lambda \in \N, , x \in \sX, \tau \in \cM_\tau$:
$$ \Pr \left[\eval(\tildesk, x) = \eval(\sk, x)  :  \begin{array}{cc} (\sk,\pk, \xk, \mk)  \gets \wmsetup(1^\lambda)  \\  \widetilde{\sk} \gets \watermark(\mk, \sk, \tau)   \end{array} \right] \geq 1 - \negl(\lambda).$$

\paragraph{Weak Pseudorandomness}
A watermarkable PRF satisfies weak pseudorandomness if for all PPT $\A$, there exists a negligible function $\negl(\lambda)$ such that for all $\lambda \in \N$:
$$ \Pr \left[ \A^{\cO(\sk, \cdot)}(x, y_b) = b:  \begin{array}{cc} x \gets \{0,1\}^\ell, y_0 = \eval(\sk, x), y_1  \gets \{0,1\}^\ell, b \gets \{0,1\}\end{array} 
\right] \leq \frac{1}{2} + \negl(\lambda).$$
where the oracle $\cO(\sk, \cdot)$ samples a uniformly random $x \gets \{0,1\}^\ell$ upon every query and outputs $(x, \eval(\sk, \cdot))$. 

We can also consider the equivalent notion where $\A$ is given an oracle $\cO(\sk,\cdot)$ which either computes $\prf.\eval(\sk,\cdot)$ on random inputs or computes a real random function that samples random input together with a random output. There will be no challenge input and $\A$ outputs a bit indicating which oracle it is given.
Similarly for the security below.

\paragraph{Variant of Weak Pseudorandomness Game}
To be better compatible with our unremovability definition, we consider the following variant of the above weak PRF game, described in a 2-stage fasion:
in stage 1, $\A^1$ is allowed to query \emph{adaptively} $\cO(\sk, \cdot)$ on any input of its own choice.
Entering stage 2, it can only query $\cO(\sk, \cdot)$ in a way that each input is sampled at random.  Then $\A^2$ is fed with challenge input
$(x, y_b)$ where $x \gets \{0,1\}^\ell, y_0 = \eval(\sk, x), y_1  \gets \{0,1\}^\ell, b \gets \{0,1\} $ and needs to output the correct $b$. 

\paragraph{$\gamma$-Unremovability of weak PRF}
The $\gamma$-Unremovability for a watermarkable implementation of a weak PRF scheme
 says,  for all PPT admissible stateful adversary $\A  $,  there exists a negligible function $\negl(\lambda)$ such that for all $\lambda\in \N$:
$$ \Pr \left[ \extract(\xk, C) \notin \cQ \wedge C \text{ is $\gamma$-good }:  \begin{array}{cc} (\sk, \xk, \mk)  \gets \wmsetup(1^\lambda)  \\ 
C \gets \A^{ 
\watermark(\mk, \sk, \cdot)}(1^\lambda) \\
\end{array} \right] \leq \negl(\lambda).$$
$\cO(\sk, \cdot)$ samples a uniformly random $x \gets \{0,1\}^\ell$ upon every query and outputs $(x, \eval(\sk, \cdot))$. 

 $\cQ$ is the set of marks queried by $\A$ and $C$ is a PPT admissible, stateful $\gamma$-good adversary in the security game $G_{\prf}(\sk, \cdot)$ if:
$$ \Pr \left[ C^{\cO(\sk, \cdot)}(x, y_b) = b:  \begin{array}{cc} x \gets \{0,1\}^\ell, y_0 = \eval(\sk, x), y_1  \gets \{0,1\}^\ell, b \gets \{0,1\}\end{array} 
\right] \geq \frac{1}{2} + \gamma.$$

\begin{remark}
In the above setting, the oracle $\cO( \cdot)$ will sample a random $r$ and output $r, \cO(r)$. 

 We can consider an even slightly stronger  notion of weak PRF: 
    in the "fully adaptive-query" setting, $\A, C$ are allowed to make any query to the real $\eval$ oracle; only the challenge input $x$ is sampled uniformly at random. 
    Then $C$ is asked o distinguish between $y_0 = \eval(\sk, x), y_1  \gets \{0,1\}^\ell$.

Note that 
\cite{goyal2021beyond} consider the first notion so that they can realize public extraction only by letting the extraction key sample input-output pairs. This notion also suffices to prove watermarkable CPA security. Both cann be constructed in \cite{goyal2021beyond} .

\end{remark}


\subsection{Definition: Watermarkable Implementation of CPA-secure SKE}
\label{sec:def_wm_cpa_ske}

A watermarkable implementation of a SKE scheme consists of the following algorithms:

\begin{description}
    \item $\wmsetup(1^\lambda) \to (\sk, \mk, \xk)$: on input security parameter, outputs a secret key $\sk$, marking key $\mk$, extraction key $\xk$.

      \item $ \Enc(\sk, m) \to \ct$: on input secret key $\sk$ and message $m$, outputs a ciphertext $\ct$.

        \item $\Dec(\sk, \ct) \to m$: on input secret key $\sk$ and ciphertext $\ct$ outputs a message $m$.
        
      \item $\watermark(\mk, \sk,\tau)$: on marking key $\mk$, secret key $\sk$, a message $\tau$, output a marked key $\sk_\tau$.
      \item $\extract(\xk, C)$: on extraction key $\sk$ and program $C$, output a mark $\tau \in \cM_\tau$ or $\bot$.
\end{description}

The correctness property is the same as the correctness property of a plain SKE scheme and we combine it with the functionality preserving property it due to it being standard.

\paragraph{Correctness and Functionality Preserving}
A watermarkable CPA-secure SKE is functionality preserving if  then there exists a negligible function $\negl(\lambda)$ ssuch that for all $\lambda \in \N, , m \in \cM, \tau \in \cM_\tau$:
$$ \Pr \left[\Dec(\tildesk, \ct') = m \wedge \Dec(\sk, \ct) = m  :  \begin{array}{cc} (\sk,\pk, \xk, \mk)  \gets \wmsetup(1^\lambda)  \\  \widetilde{\sk} \gets \watermark(\mk, \sk, \tau) \\
\ct \gets \Enc(\sk, m), \ct' \gets \Enc(\sk_\tau, m)  \end{array} \right] \geq 1 - \negl(\lambda).$$

\paragraph{CPA-security} 
A watermarkable CPA secure SKE scheme satisfies CPA security if  for all PPT admissible stateful adversary $\A = (\A_1, \A_2)$,  there exists a negligible function $\negl(\lambda)$ such that for all $\lambda\in \N$:
$$ \Pr \left[ \A_2^{\Enc(\sk, \cdot)}(\ct_b) = b:  \begin{array}{cc}
(m_0, m_1) \gets \A_1^{\Enc(\sk, \cdot)} \\
\ct_b \gets \Enc(\sk, m_b), b \gets \{0,1\}
\end{array} \right] \leq \frac{1}{2} + \negl(\lambda).$$

\paragraph{$\gamma$-Unremovability of CPA-secure SKE}
The $\gamma$-Unremovability for a watermarkable CPA secure SKE scheme
 says,  for all PPT admissible stateful adversary $\A  $,  there exists a negligible function $\negl(\lambda)$ such that for all $\lambda\in \N$:
$$ \Pr \left[ \extract(\xk, C) \notin \cQ \wedge
C \text{ is } \gamma\text{-good}:  \begin{array}{cc} (\sk, \xk, \mk)  \gets \wmsetup(1^\lambda)  \\ 
C \gets \A^{ \watermark(\mk, \sk, \cdot)}(1^\lambda) 
\\
\end{array} \right] \leq \negl(\lambda).$$
where $\cQ$ is the set of marks queried by $\A$ and $C = (C_1, C_2)$ is said to be a (PPT admissible, stateful) $\gamma$-good program in the security game $G_{CCA}^2(\sk, \cdot)$, more specifically:
$$ \Pr \left[ C_2^{\Enc(\sk, \cdot)}(\ct_b) = b:  \begin{array}{cc}
(m_0, m_1) \gets C_1^{\Enc(\sk, \cdot)} \\
\ct_b \gets \Enc(\sk, m_b), b \gets \{0,1\}
\end{array} \right] \geq \frac{1}{2} + \gamma.$$

\begin{remark}
    The above definition is equivalent to letting $A$ output $(m_0, m_1)$ and making $(m_0, m_1)$ as the auxiliary input $\aux$ for $G_{\cca}$ and $\extract$. 
    Because we can view any distribution over $(m_0, m_1)$ used by $C$ as a convex combination of different message pairs $\{(m_0, m_1)_i\})_i$ and its corresponding strategy such that the overall winning probability is $1/2+\gamma$. Thus, we can let $\A$ pick the message-pair and corresponding strategy with the largest winning probability and hardcode them into $C$ instead.
    The other way of implication is easy to see.

    This also applies to the CCA security game. 
\end{remark}

\subsection{Watermarking-Compatible Reduction}
We show that the textbook construction of CPA encryption from weak PRF is watermarking-compatible.
Given a PRF scheme that satisfies weak PRF security $\prf = (\prf.\KeyGen, \prf.\eval)$, where input and output lengths are $\ell$ and key length is $n$.

The syntax of $\prf$ and 
CPA-security is standard and we refer them to the watermarkable definitions \Cref{sec:def_wm_cpa_ske}, \Cref{sec:def_wm_cpa_ske}, by replace $\wmsetup$ with $\KeyGen$ and removing $\extract, \watermark$ algorithms.

\paragraph{Plain Construction of CPA-secure Encryption from weak PRF}
\begin{itemize}
    \item $\KeyGen(1^\lambda): \sk \gets \prf.\KeyGen(\lambda)$

    \item  $\Enc(\sk, m \in \{0,1\}^\ell) \to \ct:$ samples $r \gets \{0,1\}^\ell$; output $\ct \gets (r, \prf.\eval(\sk, r) \oplus m)$.

    \item $\Dec(\sk, \ct)$: parse $\ct := (r, \ct')$; compute $m  := \prf.\eval(\sk, r) \oplus \ct'$.
\end{itemize}

\paragraph{Watermarking-Compatible Reduction}
We briefly recall the textbook reduction for the above construction, in the language of watermarking-compatible.

\begin{claim}
    \label{claim:wprf_cpa_reduction_compatible}
    The CPA-secure SKE to weak PRF pseudorandomness reduction is a watermarking compatible reduction.
\end{claim}

$\prf$ pseudorandomness adversary $\A_\prf$ simulates CPA game for adversary $\A$ as follows:
\begin{itemize}
    \item There is no public key and public evaluation algorithm $\PubAlg$ in the scheme, so $\A_\prf$ and $\A$ will not receive information regarding public key $\pk$.
    \item First consider stage 1: For any encryption query from $\A^1$, $\A^1_\prf$ will simulate the response as follows: it will query the oracle $\cO(\sk, \cdot)$ in the security game $G_\prf$, 
    where $\cO(\cdot)$
    outputs $\prf.\eval(\sk, x)$ on any query $x$.
    
    \item After entering stage 2, for any encryption query from $\A^2$, $\A^2_\prf$ will still simulate the response as the above. Recall that $\A^2_\prf$ will not get to see any queries made in stage 1.

    \item Note that $\A^2_\prf$'s access to $\cO(\cdot)$ is changed to non-adaptive. $\cO(\cdot)$
   performs the following:
    upon every query, sample a fresh $r$, output $\prf.\eval(\sk, r)$.
    
    \item In the challenge phase, $\A^2_\prf$ will use the challenge input from the weak $\prf$ challenger: 
    $r^* \gets \{0,1\}^\ell, y^* \in \{0,1\}^\ell$ where $y^*$ is either $\prf.\eval(\sk, r^*)$ or uniformly random. 

    \item $\A^2$ sends in challenge messages $(m_0, m_1); $  $\A^2_\prf$ flips a coin $b \gets \{0,1\}$ and sends $\ct = (r^*, y^* \oplus m_b)$ to $\A^2$.

    \item If $\A^2$ guesses the correct $b$, then $\A^2_\prf$ output 0, for $y^*$ is $\prf.\eval(\sk, r^*)$", else $\A^2_\prf$ output 1, for "$y^*$ is uniformly random".
    
\end{itemize}
In the above reduction, $\A^2_\prf$'s reduction function $f(\out, b, \cQ, \inp)$ is: ignore $\cQ, \inp$, check if $\A^2$'s output $\out$ is equal to the randomness $b$ used in preparing the challenge input, if yes output $0$ else $1$.

We denote the security game for weak $\prf$ as $G_{\prf}$ and for CPA encryption as $G_{CPA}$; we also denote the advantage of $\A$ in $G_{CPA}$ as $\adv_{CPA}$ and advantage of $\A_\prf$ as $\adv_\prf$.
We can observe that $\adv_\prf \geq \adv_{CPA}/2-\frac{q}{2^{\ell}}$ where $q$ is the number of queries made by $\A$.

\subsubsection{Security of Watermarkable Implementation for CPA-secure SKE}
\label{sec:wm_cpa_from_prf_construction}

\begin{theorem}
 \label{thm:wm_wprf_to_cpa}
 Assuming that LWE is secure, then there exists secure single-key watermarkable implementation of CPA-secure SKE with private tracing.
 
Assuming that there exists indistinguishability obfuscation, then there exists secure collusion-resistant watermarkable implementation of CPA-secure SKE with public tracing.
   
\end{theorem}

\begin{theorem}
  \label{thm:goyal21_wprf} [\cite{goyal2021beyond,PKC:MaiWu22}] 
   Assuming that LWE is secure, then there exists secure collusion-resistant watermarkable implementation of weak PRF with private extraction.
   
   Assuming that there exists indistinguishability obfuscation, then there exists secure collusion-resistant watermarkable implementation of weak PRF with public extraction.

\end{theorem}

Note that both schemes above  satisfy the  properties we require in \Cref{sec:def_watermarkable_primitive}, including: extraction key of the scheme can be used to simulate the weak PRF pseudorandomness security game \emph{perfectly}; the extraction procedure satisfies the extraction syntax \footnote{For conveniece, in the rest of this work, when we cite an existing watermarking scheme as a "watermarking implementation of a primitive", then the scheme satisfies all properties defined in \Cref{sec:def_watermarkable_primitive}.}. 
\begin{remark}
    In \cite{goyal2021beyond,PKC:MaiWu22}, the watermarkable implementation of weak PRF is called traceable PRF. Their exact extraction procedure does not satisfy our requirement that the adversarial circuit $C$ can make $\prf.\eval(\sk, \cdot)$ queries during the extraction/tracing procedure. But it can be easily modified to satisfy this property: their tracing key has the capability to sample from the input-output space of the PRF, and thus can answer weak PRF queries.
\end{remark}

\begin{lemma}
    \label{lem:wm_wprf_to_cpa}
     Assuming that there exists collusion-resistant (resp. single-key) secure watermarkable implementation of weak PRF with public(resp. private) extraction, then there exists secure collusion-resistant (resp. single-key) watermarkable implementation of CPA-secure SKE with public(resp. private)extraction.
\end{lemma}

To prove \Cref{lem:wm_wprf_to_cpa}, we present the construction as below.

\paragraph{Watermarking Implementation of CPA-secure Encryption}

Given a watermarkable implementation of weak PRF $(\wprf.\wmsetup,\wprf.\eval, \wprf.\watermark,\wprf.\extract)$, we can construct a watermarkable implementation of CPA-secure SKE as follows:

\begin{description}
    \item $\wmsetup(1^\lambda ) \to (\sk, \xk,\mk)$: compute $(\sk, \xk,\mk) \gets \wprf.\wmsetup(1^\lambda )$

    %

    \item $\Enc(\sk, m \in \{0,1\}^\ell)$:   
samples $r \gets \{0,1\}^\ell$; output $\ct \gets (r, \wprf.\eval(\sk, r) \oplus m)$.

    \item $\Dec(\sk, \ct)$:
    parse $\ct := (r, \ct')$; compute $m  := \wprf.\eval(\sk, r) \oplus \ct'$.

    \item $\watermark(\mk, \sk, \tau \in \cM_{\watermark})$:
    \begin{enumerate}
        \item parse $\mk = \wprf.\mk$
        \item  output $\sk_\tau \gets \wprf.\watermark(\mk, \sk, \tau)$;
        
        
    \end{enumerate}

    \item $\extract(\xk, C):$
    \begin{enumerate}
        \item On input circuit $C$ and $\xk = \wprf.\xk$;

        \item Treat $C$ as the stage-2 adversary $\A_2$ in the watermarking-compatible reduction from CPA-security to weak PRF security to create circuit $C_\prf$ ($C_\prf$ has black-box access to $C$) as a stage-2 reduction $\A_\prf^2$ algorithm from CPA-security to the weak PRF security.
        \item 
        compute $\tau/\bot \gets \wprf.\extract(\xk, C_\prf$); if the mark extracted is not $\bot$, abort and output the mark.
    \end{enumerate}
\end{description}

\begin{proof}
    We describe the reduction in the watermarking-compatible reduction language.

    Suppose there exists a PPT adversary $\A_\watermark$ for the $\gamma$-unremovability security of the above watermarkable CPA encryption, we can consider the adversary $\A_\watermark$ to be the stage-1 adversary $\A^1$ in the CPA-security game: $\A_\watermark$ gets to query the marking oracle $\watermark(\mk, \sk, \cdot)$.
    The reduction $\cB$ which is an adversary in the strong unremovability security game  of the watermarkable implementation of weak PRF, can simulate both the answers to the marking queries by querying the PRF's marking oracle.

    Then $\A_\watermark$ outputs a circuit $C$; during the $\extract(\xk, \cdot)$ algorithm.
    C will be treated as a stage-2 adversary in the CPA-security game; with black-box access to $C$, circuit $C_\prf$ will act like a stage-2 reduction from CPA-security to weak PRF security. Since $\wprf.\extract(\wprf.\sk, \cdot)$ algorithm will perfectly simulate the weak PRF security game on its input circuit,  $C_\prf$ will be able to answer $C$'s queries and use $C$'s output to finish the reduction.

    The reduction $B$ will create a circuit $C'$ with black-box use of $C$ the same way as the above $C_\prf^C$ and output $C'$ in the $\gamma/2$-unremovability security game  of the watermarkable implementation of weak PRF.

    By our assumption, $C$ should satisfy that $\Pr[G_{CPA}(\sk, C) = 1] \geq \frac{1}{2} + \gamma$ and $\Pr[\extract(\xk, C) \notin \cQ] \geq \epsilon$ for some non-negligible $\epsilon$.
    By the design of our $\extract$ algorithm, it must be the case that $\Pr[\wprf.\extract(\wprf.\xk, C_\prf) \notin \cQ] \geq \epsilon$. Since the circuit $C'$ created by $\cB$ is exactly the same as $C_\prf$ and the set of marking queries $\cQ$ made by $\cB$ to the marking oracle is the set as $\A_\watermark$'s, we must have that $\Pr[\wprf.\extract(\wprf.\xk, C') \notin \cQ] \geq \epsilon$. 

    Meanwhile, by the property of the reduction, we have that $$\Pr[G_{\prf}(\wprf.\sk, C') = 1] \geq \frac{1}{2} + \gamma/3$$. Therefore, $\cB$ breaks $\gamma/3$-unremovability for watermarkable implementation of weak PRF for some non-negligible $\gamma$.

\end{proof}

\subsection{Watermarkable CCA-secure Encryption from Watermarkable MAC and PRF}

In this section, we show an example of obtaining a watermarkable implementation of a CCA secure secret-key encryption scheme from a watermarkable implementation of weak PRF and a watermarkable implementation of MAC(or signature).

The security of watermarkable implementation of signatures (and therefore MAC) satisfying our requirements in \Cref{sec:def_watermarkable_primitive} can be built from a modification of the watermarkable signature scheme in \cite{goyal2019watermarking}.

\subsubsection{Watermaking-Compatible Reduction of CCA Security to MAC and weak PRF}
\label{sec:plain_reduction_cca_wprf}

We first give the "plain" construction from weak PRF and MAC to CCA-secure SKE and show that the reduction is watermarking-compatible.

Given a PRF scheme that satisfies weak PRF security $\prf = (\prf.\KeyGen, \prf.\eval)$, where input and output lengths are $\ell$ and key length is $n$ and a MAC scheme $\mac = (\mac.\KeyGen,\mac.\sign, \mac.\verify)$. 

\paragraph{Construction}
\begin{itemize}
    \item $\KeyGen(1^\lambda): \prf.\sk \gets \prf.\KeyGen(\lambda), \mac.\sk \gets \mac.\KeyGen(\lambda)$; output $\sk = (\prf.\sk, \mac.\sk)$.

    \item  $\Enc(\sk, m \in \{0,1\}^\ell) \to \ct:$ 
    \begin{itemize}
    	\item parse $\sk = (\prf.\sk, \mac.\sk)$;
    	\item   samples $r \gets \{0,1\}^\ell$; 
    	\item  compute $\ct_0 \gets (r, \prf.\eval(\sk, r) \oplus m)$;
    	\item compute $\sigma = \mac.\sign(\mac.\sk, \ct_0)$
    	\item output $\ct = (\ct_0, \sigma)$,
\end{itemize}

    \item $\Dec(\sk, \ct)$: 
     \begin{itemize}
     	\item parse $\ct = (\ct_0, \sigma)$ and $\sk = (\prf.\sk, \mac.\sk)$;
     	\item If $\mac.\verify(\mac.\sk, \ct_0, \sigma)  = 1$ continue, else abort and output $\bot$.
     	\item parse $\ct_0 := (r, \ct')$; compute $m  := \prf.\eval(\prf.\sk, r) \oplus \ct'$. 
    	\item output $m$.
     \end{itemize}

\end{itemize}

\begin{claim}
\label{claim:cca_reduction_compatible}
    The CCA-secure SKE to weak PRF pseudorandomnes's reduction and the CCA-secure SKE to EUF-CMA-security of MACs reduction are both watermarking compatible reductions.
\end{claim}

\paragraph{Watermarking-Compatible Reduction }
We describe the security reduction of the above construction in the language of watermarking-compatible.

First consider a reduction for $\mac$ security,
$\A_\mac$ simulates CPA game for adversary $\A$ as follows:
\begin{itemize}
    \item There is no public key and public evaluation algorithm $\PubAlg$ in the scheme, so $\A_\prf$ and $\A$ will not receive information regarding public key $\pk$.

    \item $\A_\mac$ samples its own $\prf$ secret key $\prf.\sk$.
    \item For stage 1: For any encryption query from $\A^1$, $\A^1_\mac$ will simulate the response as follows: it will compute the $\ct_0$ part of ciphertext and then query the signing oracle $\mac.\sign(\mac.\sk, \cdot)$ in the security game $G_\mac$.
    For any decryption query, $\A_\mac$ will query the verification oracle $\mac.\verify(\mac.\sk, \cdot)$.
    \item After entering stage 2, for any encryption or decryption query from $\A^2$, $\A^2_\mac$ will still simulate the response as the above. Recall that $\A^2_\mac$ will not get to see any queries made in stage 1 and $\A^2_\mac$ will record all queries from $\A_2$.
    
    \item In the challenge phase, $\A^2$ sends in challenge messages $(m_0, m_1); $  $\A^2_\mac$ flips a coin $b \gets \{0,1\}$ and sends the encryption of $m_b$ to $\A^2$.

    \item $A^2_\mac$ continues to simulate the encryption and decryption oracles for $\A^2$.
    
    \item In the end, $\A^2_\mac$ will look up  $\A^2$'s queries: find a decryption query  $\ct = (\ct_0, \sigma)$ such that it has never been the output of an encryption query and the decryption oracle did not output $\bot$ on this query. $\A^2_\mac$ output $(\ct_0, \sigma)$ as its forgery.
    
\end{itemize}
If such a decryption query 
exists in the end, then $\A_\mac$ break the EUF-CMA-security of $\mac$
(Otherwise, we will do a reduction to weak PRF).

In the above reduction, $\A^2_\mac$'s reduction function $f(\out, r, \cQ, \inp)$ is: ignore $\out,\inp, r$, check if $\A^2$'s queries $\cQ$ contain a decryption query where the input of this query is not from the output of an encryption query; if found, output this query.

Since such a ciphertext pair was never output by the encryption oracle, then it means that the reduction $\A_\mac$ never queried the corresponding message-signature out of the decryption oracle. Thus $\A_\mac$ can use it to break EUF-CMA-security of MAC.


If no such a decryption query 
exists, we can build a reduction to break the weak pseudorandomness of PRF. The reduction to weak PRF security will be similar to \Cref{sec:wm_cpa_ske}, except that the reduction will sample its own MAC keys to simulate the oracles. Since we have discussed this example in sufficient detail in the technical overview, we omit repeating the details here. 

\subsubsection{Definition: Watermarkable Implementation of CCA-secure SKE}
\label{sec:wm_cca_ske_def}
We give the CCA2 security and  unremovability security definition for CCA2-secure SKE. The other definitions (correctness, functionality preserving) follow naturally from the correctness and previously discussed CPA secure SKE.
\paragraph{CCA2-security} 
A watermarkable CPA secure SKE scheme satisfies CPA security if  for all PPT admissible stateful adversary $\A = (\A_1, \A_2)$,  there exists a negligible function $\negl(\lambda)$ such that for all $\lambda\in \N$:
$$ \Pr \left[ \A_2^{\Enc(\sk, \cdot), \Dec(\sk, \cdot)}(\ct_b) = b:  \begin{array}{cc}
(m_0, m_1) \gets \A_1^{\Enc(\sk, \cdot),  \Dec(\sk, \cdot)} \\
\ct_b \gets \Enc(\sk, m_b), b \gets \{0,1\}
\end{array} \right] \leq \frac{1}{2} + \negl(\lambda).$$
where $\A_2$ can only make decryption queries on $\ct \neq \ct_b$.

\paragraph{$\gamma$-Unremovability of CCA2-secure SKE}
The $\gamma$-Unremovability for a watermarkable CCA2 secure SKE scheme
 says,  for all PPT admissible stateful adversary $\A = (\A_1, \A_2) $,  there exists a negligible function $\negl(\lambda)$ such that for all $\lambda\in \N$:
$$ \Pr \left[ \extract(\xk, C) \notin \cQ \wedge
C \text{ is } \gamma\text{-good} :  \begin{array}{cc} (\sk, \xk, \mk)  \gets \wmsetup(1^\lambda)  \\ 
C \gets \A^{
\watermark(\mk, \sk, \cdot)}(1^\lambda) \\
\end{array} \right] \leq \negl(\lambda).$$
where $\cQ$ is the set of marks queried by $\A$ and $C = (C_1, C_2)$ is a PPT admissible, stateful $\gamma$-good adversary in the security game $G_{CCA}(\sk, \cdot)$, more specifically:
$$ \Pr \left[ C_2^{\Enc(\sk, \cdot), \Dec(\sk, \cdot)}(\ct_b) = b:  \begin{array}{cc}
(m_0, m_1) \gets C_1^{\Enc(\sk, \cdot), \Dec(\sk, \cdot)} \\
\ct_b \gets \Enc(\sk, m_b), b \gets \{0,1\}
\end{array} \right] \geq \frac{1}{2} + \gamma.$$
$C$ is admissible if it only make queries to the oracle on $\ct \neq \ct_b$.

We refer the watermarkable MAC (and signatures) definition to \Cref{sec:def_wm_signatures}.

\subsubsection{Security of Watermarkable Implementation for CCA-secure SKE}

\begin{theorem}
 \label{thm:wm_cca}
 Assuming LWE, then there exists secure bounded collusion-resistant watermarkable implementation of CCA-secure SKE with private extraction.
 
   Assuming that there exists indistinguishability obfuscation, then there exists secure bounded collusion-resistant watermarkable implementation of CCA-secure SKE with private tracing.
   
\end{theorem}

Tha above theorem is based on \Cref{thm:goyal21_wprf} and the following:

\begin{theorem}
  \label{thm:goyal19signature} [Revision of \cite{goyal2019watermarking}] 
   Assuming that there exists secure OWFs, then there exists secure  watermarkable implementation of digital signatures (and MAC) with private extraction and bounded collusion resistance.
   
\end{theorem}

\begin{remark}
    The construction for the watermarkable digital signature scheme in the original \cite{goyal2019watermarking} does not exactly satisfy our  requirements in \Cref{sec:def_watermarkable_primitive}. But we will show that a minor modification of the scheme will suffice.
\end{remark}

\begin{lemma}
    \label{lem:wm_mac_cca}
     Assuming that there exists collusion-resistant (resp. single-key) secure watermarkable implementation of weak PRF and MAC, then there exists secure collusion-resistant (resp. single-key) watermarkable implementation of CCA-secure SKE with  private extraction.
\end{lemma}

\begin{remark}
    \cite{goyal2021beyond,PKC:MaiWu22}can give us a watermarkable weak PRF scheme that  suffices for our use, though their definition of weak PRF allows only non-adaptive queries of random $r$ and its evalaution. We need the weak PRF to allow adaptive queries to be used in CCA construction. But their prvate-key watermarkable PRF scheme already satisfies this stronger definition, by having the PRF evaluation key included in the tracing key.
\end{remark}

To prove \Cref{lem:wm_mac_cca}, we present the construction as below.

\paragraph{Watermarking Implementation of CCA-secure Encryption}

Given a watermarkable implementation of weak PRF $(\wprf.\wmsetup,\wprf.\eval, \wprf.\watermark,\wprf.\extract)$ and one of MAC $(\wmac.\wmsetup, \\\wmac.\sign, \wmac.\verify, \wmac.\watermark, \wmac.\extract)$, we can construct a watermarkable implementation of CCA-secure SKE as follows:

\begin{description}
    \item $\wmsetup(1^\lambda ) \to (\sk, \xk,\mk)$: 
    \begin{enumerate}
        \item compute $(\wprf.\sk, \\ \wprf.\xk,\wprf.\mk) \gets \wprf.\wmsetup(1^\lambda )$; compute $(\wmac.\sk, \wmac.\xk,\mac.\mk) \gets \wmac.\wmsetup(1^\lambda )$

        \item output $\sk = (\wprf.\sk,\wmac.\sk); \mk = (\wprf.\mk,\wmac.\mk); \xk = (\wprf.\xk,\wmac.\xk)$
    \end{enumerate}
    %

    \item $\Enc(\sk, m \in \{0,1\}^\ell)$:   
  \begin{enumerate}
    	\item parse $\sk = (\wprf.\sk, \wmac.\sk)$;
    	\item   samples $r \gets \{0,1\}^\ell$; 
    	\item  compute $\ct_0 \gets (r, \wprf.\eval(\wprf.\sk, r) \oplus m)$;
    	\item compute $\sigma = \wmac.\sign(\wmac.\sk, \ct_0)$
    	\item output $\ct = (\ct_0, \sigma)$,
\end{enumerate}

    \item $\Dec(\sk, \ct)$:

     \begin{enumerate}
     	\item parse $\ct = (\ct_0, \sigma)$ and $\sk = (\wprf.\sk, \wmac.\sk)$;;
     	\item If $\wmac.\verify(\wmac.\sk, \ct_0,\sigma)  = 1$ continue, else abort and output $\bot$.
     	\item parse $\ct_0 := (r, \ct')$; compute $m := \wprf.\eval(\wprf.\sk, r) \oplus \ct'$. 
    \item output $m$.
      \end{enumerate}

    \item $\watermark(\mk, \sk, \tau \in \cM_{\watermark})$:
    \begin{enumerate}
        \item parse $\mk = (\wprf.\mk, \wmac.\mk)$
        \item  compute $\sk_{\prf,\tau} \gets \wprf.\watermark(\wprf.\mk, \wprf.\sk, \tau); \\
        \sk_{\mac,\tau} \gets \wmac.\watermark(\wmac.\mk, \wmac.\sk, \tau); $ 
        
        
        
        \item output 
        $(\sk_{\prf,\tau}, \sk_{\mac,\tau})$
    \end{enumerate}

    \item $\extract(\xk, C):$
    \begin{enumerate}
        \item On input circuit $C$ and $\xk = (\wprf.\xk, \wmac.\xk) $;

         \item Initialize empty tuple $\vec{\tau}$.
         
         \item Let $P_1 := \prf$ and $P_2 := \mac$; For $i = 1,2$:
        \begin{enumerate}
           \item Treat $C$ as the stage-2 adversary $\A_2$ in the watermarking-compatible reduction from CCA-security to $P_i$ security; create a circuit $C_{P_i}^C$ (i.e. $C_{P_i}$ has black-box access to $C$) as a stage-2 reduction algorithm $\A_{i}^2$ from CCA-security game to $P_i$'s security game.
        \item compute $\tau/\bot \gets \mathsf{wP_i}.\extract(\mathsf{wP_i}.\xk, C_{P_i}^C$); if the mark extracted is not $\bot$, add it to $\vec{\tau}$.
    \end{enumerate}
    \item output $\vec{\tau}$
        \end{enumerate}
        
\end{description}

Now we prove the unremovability security of the watermarkable implementation of the CCA secure encryption scheme. 
\begin{proof}
    We describe the reduction in the watermarking-compatible reduction language.

    Suppose there exists a PPT adversary $\A_\watermark$ for the $\gamma$-unremovability security of the above watermarkable CCA encryption, we can consider the adversary $\A_\watermark$ to be the stage-1 adversary $\A^1$ in the CCA-security game: $\A_\watermark$ gets to query the marking oracle $\watermark(\mk, \sk, \cdot)$, which is a leakage on the secret key $\{\wprf.\sk, \wmac.\sk\}$.
    
    Let us denote $P_1 = \prf; P_2 = \wmac$. The reduction $\cB_i$ which is an adversary in the $\gamma/3$-unremovability security game of the watermarkable implementation of primitive $P_i$, can simulate both the answers to the marking queries 
    by querying the challenger.

    Then $\A_\watermark$ outputs a circuit $C$; during the $\extract(\xk, \cdot)$ algorithm.
    C will be treated as a stage-2 adversary in the CCA-security game; with black-box access to $C$.
    circuit $C_{P_i}$ will act like a stage-2 reduction from CCA-security to $P_i$'s security. Since $\mathsf{wP_i}.\extract(\mathsf{wP_i}.\xk, \cdot)$ algorithm will perfectly simulate the security game $G_{P_i}$ on its input circuit,  $C_{P_i}$ will be able to answer $C$'s queries and use $C$'s output to finish the reduction.

    Since we are assuming that $\Pr[\extract(\xk, C) \notin \cQ] \geq \epsilon$, for some non-negligible $\epsilon$. By the design of our $\extract$ algorithm, it must be that for both $i \in \{1,2\}$, $\Pr[\mathsf{wP_i}.\extract(\mathsf{wP_i}.\xk, C_i^C) \notin \cQ] \geq \epsilon$. 
    
    We now consider two cases.
      For any 
    $\A^{\Enc(\sk, \cdot), \Dec(\sk, \cdot), \watermark(\mk,
    \sk, \cdot)}(\pk)$ producing a $\gamma$-good $C$ such that $\Pr[\extract(\xk, \pk, C) \notin \cQ] \geq \epsilon$, one of the following cases must hold:

    \begin{itemize}
        \item \textbf{Case 1}: 
  In case 1, with some non-negligible probability, $C_1$ is $\gamma_1$-good in the game $G_{\prf}(\wprf.\sk, \cdot )$ for some non-negligible $\gamma_1$. 
 In other words, during the execution of $\extract(\xk, C)$, $C$ makes decryption queries on ciphertexts that all come from the encryption queries.

    The reduction $B_1$ will create a circuit $C_1'$ with black-box use of $C$ the same way as the circuit  $C_{P_1}^C$ created in the extraction algorithm and output $C_1'$ in the $\gamma_1$-unremovability security game  of the watermarkable implementation of weak PRF.

    
    Since the circuit $C'$ created by $\cB_1$ is exactly the same as $C_\prf^C$ and the set of marking queries $\cQ$ made by $\cB$ to the marking oracle is the set as $\A_\watermark$'s, we must have that $\Pr[\wprf.\extract(\wprf.\xk, C_1') \notin \cQ] \geq \epsilon$. 

    Meanwhile, by the property of the reduction, we have that $\Pr[G_{\prf}(\wprf.\sk, C_1') = 1] \geq \frac{1}{2} + \gamma_1$ for some non-negligible $\gamma_1$. Therefore, $\cB$ breaks $\gamma_1$-unremovability for watermarkable implementation of weak PRF.

    \item \textbf{Case 2}: 
      In case 2, with some non-negligible probability, $C_2$ is $\gamma_2$-good in the game $G_{\mac}(\wmac.\sk, \cdot )$ for some non-negligible $\gamma_2$. 

     In other words, during the execution of $\extract(\xk, C)$, $C$ makes decryption queries on ciphertexts that not only come from the encryption queries.

    The reduction $B_2$ will create a circuit $C_2'$ with black-box use of $C$ the same way as the circuit  $C_{\mac}^C$ created in the extraction algorithm and output $C_2'$ in the $\gamma_2$-unremovability security game  of the watermarkable implementation of weak PRF.

    
    Since the circuit $C'$ created by $\cB_2$ is exactly the same as $C_\mac$ and the set of marking queries $\cQ$ made by $\cB_2$ to the marking oracle is the set as $\A_\watermark$'s, we must have that $\Pr[\wmac.\extract(\wmac.\xk, C_2') \\ \notin \cQ] \geq \epsilon$. 

    Meanwhile, by the property of the reduction, we have that $\Pr[G_{\mac}(\wmac.\sk, C_2') = 1] \geq \frac{1}{2} + \gamma_2$ for some non-negligible $\gamma_2$. Therefore, $\cB_2$ breaks $\gamma_2$-unremovability for watermarkable implementation of MAC.
      \end{itemize}
    By the property of the watermarking-compatible reduction (shown in \Cref{sec:plain_reduction_cca_wprf}), one of the above two cases must hold.
\end{proof}

\section{Watermarkable CCA-secure PKE from Watermarkable Identity-based Encryption and Strong One-Time Signature}
\label{sec:wm_cca2_pke_ibe}
In this section, we give an example of how to achieve watermarkable implementation of CCA-secure PKE, where type-2 reduction is used. The watermarkable CCA-secure PKE scheme is based on the construction from selectively secure IBE and strong one-time signatures (\cite{boneh2007chosen})

\subsection{Definition: Watermarkable CCA2-secure PKE}
\label{sec:wm_cca2_pke_def}

A watermarkable implementation of a SKE scheme consists of the following algorithms:

\begin{description}
    \item $\wmsetup(1^\lambda) \to (\sk, \pk, \mk, \xk)$: on input security parameter, outputs a secret key $\sk$, public key $\pk$, marking key $\mk$, extraction key $\xk$.

      \item $ \Enc(\pk, m) \to \ct$: on input public key $\pk$ and message $m$, outputs a ciphertext $\ct$.

        \item $\Dec(\sk, \ct) \to m$: on input secret key $\sk$ and ciphertext $\ct$ outputs a message $m$.
        
      \item $\watermark(\mk, \sk,\tau)$: on marking key $\mk$, secret key $\sk$, a message $\tau$, output a marked key $\sk_\tau$.
      \item $\extract(\xk, \pk, C)$: on extraction key $\sk$, public key $\pk$ and program $C$, output a mark $\tau \in \cM_\tau$ or $\bot$.
\end{description}

The correctness property is the same as the correctness property of a plain PKE scheme and we combine it with the functionality preserving property it due to it being standard.

\paragraph{Correctness and Functionality Preserving}
A watermarkable PKE is functionality preserving if  then there exists a negligible function $\negl(\lambda)$ ssuch that for all $\lambda \in \N, , m \in \cM, \tau \in \cM_\tau$:
$$ \Pr \left[\Dec(\tildesk, \ct) = m \wedge \Dec(\sk, \ct) = m  :  \begin{array}{cc} (\sk,\pk, \xk, \mk)  \gets \wmsetup(1^\lambda)  \\  \widetilde{\sk} \gets \watermark(\mk, \sk, \tau) \\
\ct \gets \Enc(\pk, m) \end{array} \right] \geq 1 - \negl(\lambda).$$

\paragraph{CCA2-security} 
A watermarkable PKE scheme satisfies CCA2 security if  for all PPT admissible stateful adversary $\A = (\A_1, \A_2)$,  there exists a negligible function $\negl(\lambda)$ such that for all $\lambda\in \N$:
$$ \Pr \left[ \A_2^{\Dec(\sk, \cdot)}(\st, \ct_b) = b:  \begin{array}{cc}
(m_0, m_1, \st) \gets \A_1^{\Dec(\sk, \cdot)}(\pk) \\
\ct_b \gets \Enc(\pk, m_b), b \gets \{0,1\}
\end{array} \right] \leq \frac{1}{2} + \negl(\lambda).$$
$\A$ is admissible if it only make queries to the oracle on $\ct \neq \ct_b$.

\paragraph{$\gamma$-Unremovability of CCA2-secure PKE}
The $\gamma$-Unremovability for a watermarkable CCA2 secure PKE scheme
 says,  for all PPT admissible stateful adversary $\A = (\A_1, \A_2) $,  there exists a negligible function $\negl(\lambda)$ such that for all $\lambda\in \N$:
$$ \Pr \left[ \extract(\xk, \pk, C) \notin \cQ \wedge
C \text{ is } \gamma\text{-good}:  \begin{array}{cc} (\sk, \pk, \xk, \mk)  \gets \wmsetup(1^\lambda)  \\ 
C \gets \A^{ 
\watermark(\mk, \sk, \cdot)}(1^\lambda, \pk) \\
\end{array} \right] \leq \negl(\lambda).$$
where $\cQ$ is the set of marks queried by $\A$ and $C = (C_1, C_2)$ is a PPT admissible, stateful $\gamma$-good adversary in the security game $G_{CCA}(\sk, \pk, \cdot)$, more specifically:
$$ \Pr \left[ C_2^{\Dec(\sk, \cdot)}(\ct_b) = b:  \begin{array}{cc}
(m_0, m_1) \gets C_1^{\Dec(\sk, \cdot)}(\pk, 1^\lambda) \\
\ct_b \gets \Enc(\pk, m_b), b \gets \{0,1\}
\end{array} \right] \geq \frac{1}{2} + \gamma.$$
$C$ is admissible if it only make queries to the oracle on $\ct \neq \ct_b$.

\subsubsection{Preliminaries}
\paragraph{Definition: Watermarkable Implementation of Identity-Based Encryption}
\label{sec:wm_ibe_def}

For clarity, 
we first give the definition for watermarkable implementation of a selectively secure identity-based encryption $\wibe$
where one can mark the master secret key, though it is easy to see that they match the definition in \Cref{sec:def_watermarkable_primitive}. It consists of the following algorithms.

\begin{description}
    \item $ \wmsetup(1^\lambda) \to (\msk, \mpk, \mk, \xk)$: 
     on input security parameter, outputs a master secret key $\msk$, master public key $\mpk$, marking key $\mk$, extraction key $\xk$.
    
    \item $\KeyGen(\msk, \id) \to (\sk_\id)$: on input mster secret key, and $\id \in \mathcal{ID}$, output a key $\sk_\id$.

    \item $\Enc(\mpk, m, \id) \to \ct_\id$: on input public key $\pk$, message $m$ and $\id \in \mathcal{ID}$, outputs a ciphertext $\ct$.

    \item $\Dec(\sk_\id, \ct) \to m: $ on input secret key $\sk_\id$ and ciphertext $\ct$ outputs a message $m$.

    \item $\watermark(\mk, \msk, \tau) \to \msk_\tau$: on marking key $\mk$, secret key $\msk$, a message $\tau$, output a marked key $\msk_\tau$.

    \item $\extract(\xk, \mpk, \aux, C) \to \tau/\bot: $ on extraction key $\sk$, public key $\mpk$ auxiliary input $\aux$ and program $C$, output a mark $\tau \in \cM_\tau$ or $\bot$.
\end{description}

Note that for our applications, we only consider marking a master secret key. But a scheme where one marks an ID-embedded key can also be applicable in certain compositions. Both types of watermarkable IBE can be consructed from modifications of the watermarkable ABE scheme in \cite{goyal2019watermarking}.
    
\paragraph{Correctness}
The functionality-preserving property says, for all $m \in \cM, \id \in ID$, there exists a negligible function $\negl(\lambda)$ such that:
$$ \Pr \left[\begin{array}{cc}
(\Dec(\sk_\id, \ct) = m )\wedge \\
(\Dec(\msk, \ct) = m)
\end{array}:  \begin{array}{cc} (\msk,\mpk, \xk, \mk)  \gets \wmsetup(1^\lambda)  \\
\sk_\id \gets \KeyGen(\msk, \id) \\
\ct \gets \Enc(\mpk, \id, m)
\end{array} \right] \geq 1 - \negl(\lambda).$$

\paragraph{Functionality-Preserving}
The functionality-preserving property says, for all $m \in \cM, \id \in ID$ and all $\tau \in \cM_\tau$, there exists a negligible function $\negl(\lambda)$ such that:
$$ \Pr \left[\begin{array}{cc}
(\Dec(\sk_{\id, \tau}, \ct) = m )\wedge \\
(\Dec(\msk_\tau, \ct) = m)
\end{array}:  \begin{array}{cc} (\msk,\mpk, \xk, \mk)  \gets \wmsetup(1^\lambda)  \\ \msk_\tau \gets \watermark(\mk, \sk, \tau) \\
\sk_{\id, \tau} \gets \KeyGen(\msk_\tau, \id) \\
\ct \gets \Enc(\mpk, \id, m)
\end{array} \right] \geq 1 - \negl(\lambda).$$

\paragraph{Selective-ID Security}
The selective-ID CPA security
property says, for all $m \in \cM, \id \in ID$ and all $\tau \in \cM_\tau$, and for all PPT admissible stateful adversary $\A = (\A_1, \A_2, \A_3)$,  there exists a negligible function $\negl(\lambda)$ such that:
$$ \Pr \left[ \A_3^{\KeyGen(\msk, \cdot)}(\ct_b) = b :  \begin{array}{cc} \id^* \gets \cA_1(1^\lambda) \\ (\msk,\mpk, \xk, \mk)  \gets \wmsetup(1^\lambda)  \\ 
(m_0, m_1) \gets \A_2^{\KeyGen(\msk, \cdot)}(\mpk) \\
\ct \gets \Enc(\mpk, \id^*, m_b), b \gets \{0,1\}
\end{array} \right] \leq \frac{1}{2} + \negl(\lambda).$$
where $\A = (\A_1, \A_2, \A_3)$ is admissible if $\A_2, \A_3$ only makes queries to the oracle $\KeyGen(\msk, \cdot)$ on $\id \neq \id^*$.

Additionally, we denote the above security game after $\id^*,msk, \mpk$ are fixed as $G_\ibe(\msk,\mpk, \id^*, \cdot)$.

Note that the stages in the above defintion do not exactly match the stage in the watermarking-compatible reduction, as shown below.

\paragraph{$\gamma$-Unremovability}
The $\gamma$-Unremovability for a watermarkable selective-ID secure IBE scheme
 says, for all $m \in \cM, \id \in ID$ and all $\tau \in \cM_\tau$, and for all PPT admissible stateful adversary $\A = (\A_1, \A_2) $,  there exists a negligible function $\negl(\lambda)$ such that:
$$ \Pr \left[\begin{array}{cc} \extract(\xk, \mpk, \aux= \id^*, C) \notin \cQ  \\
\wedge C \text{ is } \gamma\text{-good}
\end{array} :  \begin{array}{cc} \id^* \gets \cA_1(1^\lambda) \\ (\msk,\mpk, \xk, \mk)  \gets \wmsetup(1^\lambda)  \\ 
C \gets \A_2^{
\watermark(\mk, \msk, \cdot)}(\mpk) \\
\end{array} \right] \leq \negl(\lambda).$$
where $\cQ$ is the set of marks queried by $\A_2$ and $C = (C_1, C_2)$ is a PPT admissible, stateful $\gamma$-good adversary in the security game $G_\ibe(\msk, \mpk, \id^*, \cdot)$, more specifically:
$$ \Pr \left[ C_2^{\KeyGen(\msk, \cdot)}(\ct_b) = b:  \begin{array}{cc}
(m_0, m_1) \gets C_1^{\KeyGen(\msk, \cdot)}(\mpk) \\
\ct_b \gets \Enc(\mpk, \id^*, m_b), b \gets \{0,1\}
\end{array} \right] \geq \frac{1}{2} + \gamma.$$
$C$ is admissible if  and only if 
$C$ only make queries to the oracle $\KeyGen(\msk, \cdot)$ on $\id \neq \id^*$.

\subsection{Watermarkable CCA-secure PKE Construction}

In this section, we prove 
\Cref{thm:wm_cca_pke}with an alternative approach, based on \cite{boneh2007chosen}.



\begin{theorem}
\label{thm:ibe_from_lwe}
        Assuming the hardness of LWE, there exists secure watermarkable implementation of identity-based encryption with selective-ID security, with private extraction and collusion-resistant security.
\end{theorem}

We prove this above theorem in \Cref{sec:abe_construction} by making a modification of \cite{goyal2019watermarking}.

\begin{lemma}
\label{lem:cca_pke_from_ibe_ots}
      Assuming there exists secure watermarkable implementation of (private-extraction, collusion resistant) identity-based encryption with selective-ID security and strongly unforgeable one-time signatures, then there exists secure (private-extraction, collusion resistant) watermarkable implementation of CCA-secure PKE
\end{lemma}

Now we describe the construction for watermarkable implementation of CCA-secure PKE from a watermarkable implementation of selective-ID secure IBE scheme $\wibe = (\wmsetup, \KeyGen, \\ \Enc, \Dec, \watermark, \extract)$ and a one-time strongly unforgeable signature scheme $\OTS = (\KeyGen, \sign, \verify)$ (without watermarking).

\begin{description}
      \item $ \wmsetup(1^\lambda) \to (\msk, \mpk, \mk, \xk)$: 
      compute $(\wibe.\msk, \wibe.\mpk, \wibe.\mk, \wibe.\xk) \gets \wibe.\wmsetup(1^\lambda)$; output $\sk = \wibe.\msk; \pk = \wibe.\mpk; \mk = \wibe.\mk; \xk = (\wibe.\xk, \wibe.\mpk)$.

    \item $\Enc(\pk, m) \to \ct$:
    \begin{enumerate}
        \item compute $(\OTS.\vk, \OTS.\sk) \gets \OTS.\KeyGen(1^\lambda)$;

        \item compute $\ct_\vk \gets \wibe.\Enc(\pk, \id = \vk, m)$; 

        \item compute $\sigma \gets \OTS.\sign(\OTS.\sk, \ct_\vk)$;

        \item output $\ct = (\ct_\vk, \sig, \vk)$.
    \end{enumerate}

    \item $\Dec(\sk, \ct) \to m: $
    \begin{enumerate}
        \item parse $\ct := (\ct_\vk, \sig, \vk), 
 \sk = \wibe.\msk$;

        \item if $\OTS.\verify(\vk, (\ct_\vk, \sig)) = 1$, continue; else abort and output $\bot$.

        \item compute $\sk_\vk \gets \wibe.\KeyGen(\wibe.\msk, \id = \vk)$;

        \item Output $m \gets \wibe.\Dec(\sk_\vk, \ct_\vk)$.
    \end{enumerate}

    \item $\watermark(\mk, \sk, \tau) \to \sk_\tau$: 
   compute $\sk_\tau \gets \wibe.\watermark(\mk, \sk = \wibe.\msk, \tau)$.

    \item $\extract(\xk, \aux = \bot, C) \to \tau/\bot: $ 
     \begin{enumerate}
     
        \item parse $\xk := \wibe.\xk,  \wibe.\mpk$. Samples $(\OTS.\vk^*, \OTS.\sk^*) \gets \OTS.\KeyGen(1^\lambda)$.
        
        \item Create a circuit $C_\ibe$ that has black-box access to $C$.

        \begin{itemize}
        \item $C_\ibe$ is hardcoded with challenge $\id^* = \OTS.\vk^*$ and $\OTS.\sk^*$, $\wibe.\mpk$ 
      \item $C_\ibe$ simulates the CCA PKE security game for $C$ by simulating $\Dec(\sk, \cdot)$ oracles as follows: 
            \begin{enumerate}
                \item  when $C$ makes a query $\ct = (\ct_\vk, \sig, \vk)$; first check if $\OTS.\verify(\vk, (\ct_\vk, \sig)) = 1$, if yes continue; else output $\bot$. 

                \item make an external query to a $\wibe.\KeyGen(\msk, \cdot)$ oracle for $\sk_\vk \gets \wibe.\KeyGen(\msk, \vk)$. Output $m \gets \wibe.\Dec(\sk_\vk, \ct_\vk)$.

            \end{enumerate}
             \item In the challenge phase, $C$ outputs $(m_0, m_1)$; then $C_\ibe$ submits $(m_0, m_1)$ to the external challenger and receives challenger ciphertext $\ct_{b, \vk^*} = \Enc(\mpk, \OTS.\vk^*, m_b), b \gets \{0,1\}$;

\item feed $C$ with $\ct_b^* = (\ct_{b, \vk^*}, \sig^* = \OTS.\sign(\OTS.\sk^*, \ct_{b, \vk^*}), \OTS.\vk^*)$.

 \item Continue to simulate $\Dec(\sk, \cdot)$ oracle
          for $C$ and only answer queries $\ct \neq \ct_b^*$; finally $C$ outputs a bit  $b'$ and $C_\ibe$ outputs the same.
        \end{itemize}
        \item Output  $\tau/\bot \gets \wibe.\extract(\wibe.\xk, \aux = \id^*, C_\ibe)$; 
    \end{enumerate}
\end{description}

\begin{proof}
    We devide our analysis into cases.

    Suppose there exists adversary $\A$ that breaks the $\gamma$-unremovability of watermarkable CCA-secure PKE for some non-negligible $\gamma$, i.e. $\A^{\Dec(\sk, \cdot),\watermark(\mk, \sk,\cdot)}(\pk)$ produces some program $C$ such that $\Pr[G_{CCA}(\sk, \pk, C) = 1] \geq \frac{1}{2}+\gamma$ (for $G_{CCA}$ see \Cref{sec:wm_cca_ske_def}) and $\Pr[\extract(\xk, C) \in\cQ] \geq \epsilon$ for some non-negligible $\epsilon$.
   

    For any $(\sk, \pk, \xk, \mk)$ generated by $\wmsetup$ and any $\A^{ 
    \watermark(\mk,
    \sk, \cdot)}(\pk)$ producing a $\gamma$-good $C$ such that $\Pr[\extract(\xk, \pk, C) \notin\cQ] \geq \epsilon$, one of the following cases must hold:

    \begin{itemize}
        \item \textbf{Case 1: during the execution of $\extract(\xk, C)$, the program $C$ produced by the adversary $\A$ will only make queries that have the format $\ct = (\ct_\vk, \vk, \sig)$, where $\vk \neq \vk^*$, $\vk*$ is the $\OTS$'s verification key part in the challenge ciphertext $\ct_b^* = (\ct_{b, \vk^*}, \sig^*, \OTS.\vk^*)$}.

        In this case, we show that we can use $\A$ to break the $\gamma$-unremovability of the watermarkable IBE.

        The reduction $\cB_\wibe$ interacts with $\A$ as follows:
        \begin{itemize}

            \item $\cB$ samples $(\vk^*, \sk^*) \gets \OTS.\KeyGen(1^\lambda)$ and submits $\vk^*$ as the challenge $\id^*$ to the challenger;
            \item $\cB_\wibe$ receives $\wibe.\mpk$ from the challenger and gives it as the public key to $\A$. For $\A$'s marking queries, $\cB$ queries the marking oracle $\wibe.\watermark(\wibe.\mk, \wibe.\msk, \cdot)$;


        \end{itemize}

        Next, $\A$ outputs a program $C$.
        $\cB$ creates a circuit $C_\wibe$ with black-box access to $C$. $C$ is treated as a stage-2 adversary in the reduction from CCA-secure PKE to selective-ID secure IBE:
        \begin{itemize}
            \item $C_\wibe$ is hardcoded with the challenge $\id^* = \vk^*$ and its corresponding signing key $\sk^*$ and $\wibe.\mpk$. Note that $C_\wibe$ does not receive any queries $\cB$ has answered/queried when answering $\A$s queries in the previous stage.

            \item $C_\wibe$ simulates answers to $C$s decryption queries:
            when $C$ makes a query $\ct = (\ct_\vk, \sig, \vk)$; first check if $\OTS.\verify(\vk, (\ct_\vk, \sig)) = 1$, if yes continue; else output $\bot$. Make an external query to the $\wibe.\KeyGen(\msk, \cdot)$ oracle for $\sk_\vk \gets \wibe.\KeyGen(\msk, \vk)$. Output $m \gets \wibe.\Dec(\sk_\vk, \ct_\vk)$. Since all the $\vk$ queried are not equal to $\vk^*$, these queries arre all valid.

            \item  In the challenge phase, $C$ outputs $(m_0, m_1)$; then $C_\wibe$ submits $(m_0, m_1)$ to the external challenger and receives challenger ciphertext $\ct_{b, \vk^*} = \Enc(\mpk, \vk^*, m_b), b \gets \{0,1\}$;

\item feed $C$ with $\ct_b^* = (\ct_{b, \vk^*}, \sig^* = \OTS.\sign(\sk^*, \ct_{b, \vk^*}), \vk^*)$.

 \item Continue to simulate $\Dec(\sk, \cdot)$ oracle
          for $C$ and only answer queries $\ct \neq \ct_b^*$; finally $C$ outputs a bit  $b'$ and $C_\wibe$ outputs the same.
        \end{itemize}

In the execution of the extraction $\extract(\xk, C)$, $C$ is used in the program $C_\ibe$ as a black box and by our design, it must hold that $\Pr[\wibe.\extract(\xk, \mpk, C_\ibe) \in \cQ \wedge C \text{ is $\gamma$-good}] \geq \epsilon$ for some non-negligible $\epsilon$. 
Note that $C_\ibe$ created by $\extract(\xk, C)$ differs from $C_\wibe$ created by $\cB$ in that: $\extract(\xk, C)$ samples a fresh $(\vk^*, \sk^*) \gets \OTS.\KeyGen(1^\lambda)$;  $C_\wibe$ is hardcoded with a  $(\vk^*, \sk^*)$ sampled by $\cB$ previously. However,  the challenge $\id$'s in both settings are sampled freshly at random, independent of the $\wibe.\KeyGen(\wibe.\msk,  \cdot)$(i.e. $\Dec(\sk, \cdot)$) queries made in the stage of $\cB$ interacting with $\A$ before $\A$ outputs $C$ (i.e. stage-1 in the watermarking-compatible reduction). Therefore, the output of $C_\wibe, C_\ibe$ are the same unless in the real CCA PKE security game, $\A$ has made a query on some ciphertext $\ct = (\ct_{\vk^*} = \wibe.\Enc(\mpk, \vk^*, m_b), \vk^*, \sig)$
for the same $\vk^*$ that is later sampled by $\extract(\xk, C)$  and for one of the challenge messages $m_b \in \{m_0. m_1\}$. But this happens with negligible probability since $\vk^*$ is sampled at random.


Since we have $\Pr[\wibe.\extract(\xk, \aux = \vk^*, C_\ibe) \in \cQ] \geq \epsilon$ and given our assumption in case 1, we also have $\Pr[G_\ibe(\wibe.\mpk, \wibe.\msk, \id^* = \vk^*, C_\ibe) = 1] \geq 1/2+ \gamma$, we should deduce that the program $C_\wibe$ produced by the reduction $\cB_\wibe$ from $C$ performs the same with all but negligible difference: with probability $\epsilon - \negl(\lambda)$, we still have $\wibe.\extract(\xk, \aux = \id^* C_\wibe) \in \cQ]$ and  $\Pr[G_\ibe(\wibe.\mpk, \wibe.\msk, \aux = \id^*,  C_\wibe) = 1] \geq 1/2+ \gamma$.

\vspace{\baselineskip}

        \item \textbf{Case 2:
        during the execution of $\extract(\xk, C)$, the program $C$ produced by the adversary $\A$ will make at least one query that have the format $\ct = (\ct_\vk, \sig, \vk)$, where $\vk =  \OTS.\vk^*$, $\OTS.\vk^*$ is the $\OTS$'s verification key part in the challenge ciphertext $\ct^*$; further more, the query $\ct = (\ct_\vk, \sig, \vk)$ satisfies $\OTS.\verify(\OTS.\vk^*, \ct_\vk, \sig) =  1$ and $\ct^* \neq  \ct$}

        In this case, we show a Type-2 watermarking compatible reduction  to the strongly unforgeability security of the one-time signature scheme.

        In stage-1 of the watermarking-compatible reduction, $\cB_\OTS^1$ operates as follows:
        \begin{itemize}
            \item $\cB_\OTS^1$ receives $\vk^*$ from the challenger. It samples $\wibe.\xk, \msk,\mpk,\mk$ on its own. $\cB_\OTS^1$ can thus answer all the marking queries from $\A$.


        \end{itemize}

        After $\A$ outputs program $C$, since we are given that  $\Pr[G_{CCA}(\sk,\pk, C) = 1] \geq 1/2+\gamma$ and $C$ makes decryption queries on valid ciphertexts with the $\vk^*$.
        $\cB_\OTS$ can do the following.  

        Entering stage-2 of the reduction, $\cB_\OTS^2$ receives all stateful information from $\cB_\OTS^1$ except for the queries made by $\A$.

        $\cB_\OTS^2$ runs $C$ and simulates the CCA-security game, $\Dec(\sk, \cdot)$ oracles for $C$ as follows:
        \begin{itemize}
            \item   For decryption queries o the format $\ct = (\ct_0, \vk, \sig)$,  first checks if $\OTS.\verify(\vk, (\ct_\vk, \sig)) = 1$, if yes continue; else output $\bot$. 

            Then decrypt $\ct$ directly using $\sk_\vk \gets \KeyGen(\msk, \vk)$ and $m \gets \Dec(\sk_\vk, \ct_\vk)$.
        \end{itemize}

        In the challenge phase, $C$ submits 2 messages $(m_0, m_1)$ and $\cB_\OTS^2$ outputs challenge ciphertext $\ct^* = (\ct_{\vk^*, b} = \wibe.\Enc(\mpk,\vk^*, m_b \gets (m_0, m_1)), \vk^*, \sig^*)$ by making one query to the signing oracle $\OTS.\sign(\sk^*, \cdot)$ on message $\ct_{\vk^*, b}$.

        Then the query phase continues and $C$ is only allowed to query on $\ct\neq \ct^*$. Since in the end, there will be one query $\ct$ such that $\ct = (\ct', \vk^*, \sig')$ where $\ct \neq \ct^*$ but $\OTS.\verify(\vk^*, \ct', \sig') = 1$. Thus $\cB_\OTS^2$ can output $(\ct', \sig')$ as a forgery.
    \end{itemize}
\end{proof}

\fi

\ifllncs
\else
\bibliography{main}
\bibliographystyle{alpha}
\fi

\appendix

\ifllncs
\else
\fi 

\ifllncs
\else
\fi

\ifllncs
\else

\section{Watermarkable CCA-secure PKE from Watermarkable CPA-secure PKE and Statistically Simulation-Sound NIZK}
\label{sec:naor_yung_wm}

In this section we present a watermarkable CCA2-secure PKE, based on a modification of the \cite{naor1990public} scheme.


\subsection{Preliminaries: Statistically Simulation Sound $\nizk$ Proof for $\NP$}
\label{def:sss-nizk}
A statistically simulation sound $\nizk$ proof for language $L \in \NP$ 
with relation $R_L$ consists of the following efficient algorithms.

\begin{itemize}
    \item $\nizk.\setup(1^\lambda) \to (\crs, \td)$: On input the security parameter
$1^\lambda$, the setup returns a common reference string $\crs$ and a trapdoor $\td$.

\item $\nizk.\prove(\crs, w, x) \to \pi$: On input a common  reference string $\crs$, a classical witness $w$, and a
statement $x$, the proving algorithm returns a proof $\pi$.

\item $\nizk.\ver(\crs, \pi, x) \to 0/1$: On input a common reference string $\crs$, a proof $\pi$, and a statement $x$, the verification
algorithm returns a bit in $\{0, 1\}$.
\end{itemize}

A SSS-$\nizk$ for $\NP$ scheme should satisfy the following properties:

\paragraph{Correctness}
A NIZK proof $(\nizk.\setup, \nizk.\prove, \nizk.\ver)$ is correct if there exists
a negligible function $\negl(\cdot)$ such that for all $\lambda \in \N$, all $x \in L$, and all $w \in \mathcal{R}_L(x)$ it holds that
$$\Pr[\nizk.\ver(\crs, \nizk.\prove(\crs, w, x), x) = 1 ] = 1 - \negl(\lambda)$$
where $\crs \gets \nizk.\setup(1^\lambda)$.

\paragraph{Statistical Soundness}
A $\nizk$ proof $(\nizk.\setup, \nizk.\prove, \nizk.\ver)$ is computationally sound if there exist a negligible function $\negl(\cdot)$ such that for all unbounded adversaries $\A$ and all $x \notin L$, it holds that:
$$\Pr[\nizk.\ver(\crs, \pi \gets \A(\crs, x), x) = 1] = \negl(\lambda)$$

where $\crs \gets \nizk.\setup(1^\lambda)$.

\paragraph{Computational Zero Knowledge}
A $\nizk$ proof $(\nizk.\setup, \nizk.\prove, \nizk.\ver)$ is computationally zero-knowledge if there exists a PPT simulator $\Sim$ such that for all non-uniform PPT adversaries, all statements $x \in L$ and all witnesses $w \in \mathcal{R}_L(x)$, it holds that
$\Sim(1^\lambda, \crs, \td,
, x) \approx_c  \nizk.\prove(\crs, w, x)$
where $\crs \gets \nizk.\setup(1^\lambda)$.

\paragraph{Statistically Simulation Soundness}
A $\nizk$ proof satisfies statistically simulation soundness if
it is infeasible to convince an honest verifier of a false statement even when the adversary itself is provided
with a simulated proof. 

Formally, for all statements $x$ and all (even unbounded) adversaries $\A =(\Sim_1, \Sim_2)$ , there exists a negligible function $\negl(\cdot)$:
$$ \Pr\left[(\crs, \td) \gets \Sim_1(1^\lambda,x), \pi \gets \Sim_2(x,\crs, \td) : \substack{\exists(x',\pi') \wedge x' \neq x \wedge x' \in L_{no}\\ \text{ and }\\ V (\crs, x', \pi') = 1} \right] \leq \negl(\cdot)$$
where $\td$ is a trapdoor generated by the simulator along with the simulated $\crs$.

\paragraph{Instantiation}
A SSS-$\nizk$ proof system can be realized from a standard $\nizk$ proof system \cite{GGHRSW13}, which can be built on the hardness of LWE \cite{peikert2019noninteractive, waters2024nizk}.

\subsection{Construction and Security}

\begin{theorem}
    \label{thm:wm_cca_pke}
    Assuming the hardness of LWE, there exists secure watermarkable implementation of CCA-secure PKE with private extraction and collusion resistance.
\end{theorem}

\begin{lemma}
\label{thm:wm_cca_pke_nizk_cpa}
    Assuming the security of watermarkable implementation of a CPA-secure PKE $\wpke = (\wmsetup,\Enc, \Dec, \watermark, \extract)$ and  statistically simulation sound NIZK scheme $\nizk = (\setup, \prove, \verify)$ (without watermarking) , there exists a secure watermarkable implementation  of CCA(2)-secure PKE.
\end{lemma}

The two building blocks: watermarkable implementation of a CPA-secure PKE $\wpke = (\wmsetup,\Enc, \Dec, \\ \watermark, \extract)$ and a SSS-NIZK scheme $\nizk = (\setup, \prove, \verify)$ can both be obtained from LWE.

\paragraph{Construction} Given a watermarkable implementation of a CPA-secure PKE $\wpke = (\wmsetup,\Enc, \\ \Dec, \watermark, \extract)$ and a SSS-NIZK scheme $\nizk = (\setup, \prove, \verify)$, the construction of watermarkable implementation  of CCA-secure PKE is as follows.

\begin{description}
    \item $\KeyGen(1^\lambda, 1^n):$
     Compute $(\wpke.\pk_1, \wpke.\sk_1, \wpke.\xk_1, \wpke.\mk_1) \gets \wpke.\wmsetup(\lambda)$; $(\wpke.\pk_2, \wpke.\sk_2, \wpke.\xk_2, \wpke.\mk_2) \gets \wpke.\wmsetup(\lambda)$;

     Compute $(\crs, \td) \gets \nizk.\setup(1^\lambda)$;
    Output $\sk = (\wpke.\sk_{i})_{i \{1,2\}}; \pk = (\{\wpke.\pk_{i}\}_{i\in \{1,2\}}, \crs); \xk = (\{\wpke.\xk_{i}\}_{i\in \{1,2\}}, \td); \mk = \{\wpke.\mk_{i}\}_{i\in \{1,2\}}$ 
    
    \item $\Enc(\pk, m)$: 
    \begin{itemize}
        \item parse $\pk := \pk_1, \pk_2, \crs$; 
        
        \item compute 
        $\ct_1 \gets \wpke.\Enc(\pk_{1}, m); \ct_2 \gets \wpke.\Enc(\pk_2, m)$;

        \item compute $\pi \gets \nizk.\prove(\crs, (\ct_1, \ct_2), (r_1, r_2, m))$ for the following statement:

        $\exists$ witness $(r_1, r_2 , m)$ such that $\ct_i = \wpke.\Enc(\pk_{i}, m; r_i)$ for all $i \in [2]$, where $r_i$ is the randomness used in encryption.

        \item Output $\ct = (\{\ct_i\}_{i \in [2]},  \pi)$.
    \end{itemize}

    \item $\Dec(\sk, \ct)$: 
    \begin{itemize}
        \item parse $\ct = (\ct_1, \ct_2, \pi); \sk = (\sk_1,\sk_2)$;
        \item if $\nizk.\verify(\crs, \pi, (\ct_1, \ct_2)) = 1$, continue; else abort and output $\bot$.

        \item compute $m \gets \wpke.\Dec(\sk_{1}, \ct_1)$;
        output $m'$. 
    \end{itemize}

    \item $\watermark(\mk, \sk, \tau)$: parse $\mk = (\mk_1, \mk_2); \sk = (\sk_1, \sk_2)$; output $(\sk_{\tau,1} \gets \wpke.\watermark(\sk_1, \tau); \sk_{\tau,2} \gets \wpke.\watermark(\sk_2, \tau)$;
    
    \item $\extract(\xk, \pk, \aux, C)$: \begin{itemize}
        \item parse $\xk := (\xk_1, \xk_2, \nizk.\td); \pk := (\pk_1, \pk_2, \crs)$. 
        Initialize an empty tuple $\vec{\tau}$;
        
        For $i = 1, 2$:
        Create the following circuit $C_i$ with black-box access to $C$:
        \begin{itemize}
            \item $C_i$ is hardcoded with $(\pk_1, \pk_2, \xk_{j \neq i},  \td, \crs)$ and simulates the $\ccapke$ game for $C$ as follows: 
            
            \item For $C$'s decryption queries $\ct = (\ct_1, \ct_2, \pi)$: 
            \begin{itemize}
                \item First check if $\nizk.\verify(\crs, \pi) = 1$, if 0 output $\bot$; else continue;

                \item By the extraction key simulation property, since $\xk_{j \neq i}$ can be used to simulate the oracles used in $G_{CPA}(\pk_j, \sk_j, \cdot)$, $C_i$ can simulate the oracle $\wpke.\Dec(\sk_j, \cdot)$ and thus decrypt $\ct_j$ in the ciphertext $\ct$ to output $m$.
            \end{itemize}

            \item $C$ submits challenge messages $(m_0, m_1)$; $C_i$ submits $(m_0, m_1)$ to the external challenger; $C_i$ receives challenge ciphertext $\ct_i^* = \wpke.\Enc(\pk_i, m_b), b \gets \{0,1\}$ from the challenger.
            
            $C_i$ prepares the following ciphertext:
            compute $\ct_j \gets \wpke.\Enc(\pk_j, m_{b_j}), b_j \gets \{0,1\}$;
            
            compute $\widehat{\pi} \gets \Sim(\td, \crs, (\ct_1^*, \ct^*_2))$  where $\Sim$ is the simulator algorithm for $\nizk$.
            
            Then it sends 
            $\ct^{*} = (\ct_1^*, \ct^*_2, \widehat{\pi})$ to $C$.
            
            \item $C$ continues to simulate the decryption oracle as above to decrypt only valid ciphertexts $\ct \neq \ct^*$.
            
            \item In the end, $C_i$ outputs the same as
            $C$ outputs.
            \item Add $\tau/\bot \gets \wpke.\extract(\wpke.\xk_i, \wpke.\pk_i, \aux = \bot, C_i)$ to $\vec{\tau}$.
        \end{itemize}
        \item Output $\vec{\tau}$.
    \end{itemize}
    
\end{description}

\begin{proof} 
We make a few claims for the proof of unremovability.

First, we consider a hybrid game where stage-1 (before the unremovability adversary $\A$ produces program $C$) is the same, but the stage-2 game for $\ccapke$ where the program $C$ plays in is different. We call this stage-2 game $G_{\cca, H_1}(\sk, \pk, \cdot)$, defined as follows: 

 \begin{itemize}
            \item On input  $(\sk, \pk) = (\pk_1, \pk_2, \sk_1, \sk_2,  \td, \crs)$; the challenger plays the $\ccapke$ game for $C$ as follows: 
            
            \item For $C$'s decryption queries $\ct = (\ct_1, \ct_2, \pi)$: 
            \begin{itemize}
                \item First check if $\nizk.\verify(\crs, \pi) = 1$, if 0 output $\bot$; else continue; output $m' \gets \wpke.\Dec(\sk_1, \ct_1)$.
            \end{itemize}

            \item $C$ submits challenge messages $(m_0, m_1)$ 
            
            The challenger prepares the following ciphertext:
            compute $\ct_i^* \gets \wpke.\Enc(\pk_i, m_b)$ for all $i = 1,2$, $b \gets \{0,1\}$;
            
            \underline{compute $\widehat{\pi} \gets \Sim(\td, \crs, (\ct_1^*, \ct^*_2))$  where $\Sim$ is the simulator algorithm for $\nizk$.}
            
            Then it sends 
            $\ct^{*} = (\ct_1^*, \ct^*_2, \widehat{\pi})$ to $\A$ and $C$ outputs a guess $b'$.

                \item continues to simulate the decryption oracle as above to decrypt only valid ciphertexts $\ct \neq \ct^*$.
        \end{itemize}

    \begin{claim}
    \label{claim:hyb1_nizk}
          Assuming the computational zero-knowledge property of $\nizk$, for any admissible PPT $C$, if $C$ is a $\gamma$-good program in the original game $G_{\ccapke}(\sk, \pk, \cdot)$, then $C$ is a $(\gamma - \negl(\lambda))$-good program in the above 
            $G_{\cca, H_1}(\sk, \pk, \cdot)$ defined.
    \end{claim}

    \begin{proof}
     We refer the definition of $G_{\ccapke}(\sk, \pk, C)$ to \Cref{sec:wm_cca2_pke_def}. 

Let $\view_{{\cca, H_1}}$(see \Cref{def:game_view} for defintion) be the transcript output by the above game.
     
It is to see that the only difference in $G_{{\cca, H_1}}(\sk, \pk, C)$ and $G_{\ccapke}(\sk, \pk, C)$ is in the generation of the proof $\pi$: $G_{\ccapke}(\sk, \pk, C)$ the proof is generated honestly using $\nizk.\prove$ and in $G_{{\cca, H_1}}(\sk, \pk, C)$ it is generated by $\Sim(\td, \cdot)$. By the computational zero knowledge property of $\nizk$,
$\view_{G_{\ccapke}}$ and $\view_{{\cca, H_1}}$ are computationally indistinguishable.
Otherwise we can build a distinguisher, given $\nizk.\crs$, that samples $\wmsetup$ on its own and use $C$ to break the computational zero knowledge property.

Therefore, the output distributions of any admissible PPT $\A$ in $G_{{\cca, H_1}}(\sk, \pk, C)$ and $G_{\ccapke}(\sk, \pk, C)$ must be computationally indistinguishable. Thus, any $\gamma$-good $C$ in $G_{\cca, H_1}(\sk, \pk, \cdot)$ must be $(\gamma-\negl(\lambda))$-good in $G_{\ccapke}(\sk, \pk, \cdot)$
        
    \end{proof}

    Next, we consider a next stage-2 hybrid game $G_{\cca, H_2}(\sk, \pk, \cdot)$ (note that the stage before $\A$ outputs $C$ is still the same:

 \begin{itemize}
            \item On input  $(\sk, \pk) = (\pk_1, \pk_2, \sk_1, \sk_2,  \td, \crs)$;
            
            \item For $C$'s decryption queries $\ct = (\ct_1, \ct_2, \pi)$: 
            \begin{itemize}
                \item First check if $\nizk.\verify(\crs, \pi) = 1$, if 0 output $\bot$; else continue; output $m' \gets \wpke.\Dec(\sk_1, \ct_1)$.
            \end{itemize}

            \item $C$ submits challenge messages $(m_0, m_1)$ 
            
            The challenger prepares the following ciphertext:
            compute \underline{$\ct_1^* \gets \wpke.\Enc(\pk_1, m_{b_1})$},\\ \underline{$b_1 \gets \{0,1\}$; compute $\ct_2^* \gets \wpke.\Enc(\pk_2, m_{b_2})$  $b_2 \gets \{0,1\}$}
            
            compute $\widehat{\pi} \gets \Sim(\td,\crs, (\ct_1^*, \ct^*_2))$  where $\Sim$ is the simulator algorithm for $\nizk$.
            
            Then it sends 
            $\ct^{*} = (\ct_1^*, \ct^*_2, \widehat{\pi})$ to $\A$ and $C$ outputs a guess $b'$.

                \item continues to simulate the decryption oracle as above to decrypt only valid ciphertexts $\ct \neq \ct^*$.
        \end{itemize}

    \begin{claim}
        \label{claim:hyb_2_pke2}

    For any admissible PPT $C$, if $C$ is a $\gamma$-good program in the game  $G_{\cca, H_1}(\sk, \pk, \cdot)$, then $C$ is a $\gamma/2$-good program in the above 
            $G_{\cca, H_2}(\sk, \pk, \cdot)$ defined.

    \end{claim}
    \begin{proof}
      The only difference between  $G_{\cca, H_1}(\sk, \pk, \cdot)$ and $G_{\cca, H_2}(\sk, \pk, \cdot)$ 
 is that the bits $b_i$ used to decide which of $m_0, m_1$ to encrypt for the challenge ciphertext is the same for $G_{\cca, H_1}(\sk, \pk, \cdot)$ and independent for $G_{\cca, H_2}(\sk, \pk, \cdot)$. When $b_1 = b_2$, the two games are the same. Therefore, we have that $C$ is a $\gamma/2$-good program in the above $G_{\cca, H_2}(\sk, \pk, \cdot)$ given that $C$ is a $\gamma$-good program in the above $G_{\cca, H_1}(\sk, \pk, \cdot)$.

    \end{proof}

Next we show that for any $\gamma$-unremovability adversary $\A$ where the output program $C$ is $\gamma$-good in $G_{\cca, H_2}(\sk, \pk, \cdot)$, $\A$ can be used to build a reduction to $\gamma$-unremovability of CPA-security of for $\wpke$.
\begin{claim}
\label{claim:reduction_wpke2}
   For any $\gamma$-unremovability adversary $\A$ where the output program $C$ is $\gamma$-good in $G_{\ccapke}(\sk, \pk, \cdot)$, $\A$ can be used to build a reduction to $(\gamma/2- \negl(\lambda))$-unremovability of CPA-security of for $\wpke$.
\end{claim}
\begin{proof}
    By \Cref{claim:hyb1_nizk} and \Cref{claim:hyb_2_pke2}, we know that if there is a $\gamma$-unremovability adversary $\A$ for \Cref{sec:wm_cca2_pke_def}, then if we put its output circuit $C$ into our $\extract(\xk, \pk, \aux, C)$ procedure, the game simulated by $\extract$, when $i = 2$ is exactly the same as $G_{\cca, H_2}(\sk, \pk, \cdot)$, i.e. the NIZK proof in the challenge ciphertext is generated by simulator and the bits for encryption are independent.

    Thus, for any $C$ produced by a winning $\A$, if $C$ is $\gamma$-good by definition in $G_{\ccapke}(\sk, \pk, \cdot)$, $C$ will be a  $(\gamma-\negl(\lambda))/2$-good program during the execution of $\extract(\xk, \pk, C)$ for $i = 2$. 

    Also, by the statistical soundness and statistical simulation soundness of $\nizk$, no adversary $\A$ and its output program $C$ will make a query on a $\ct$ that contains a false proof $\pi'$, that will pass the verification, even after $C$ has seen the simulated proof $\widehat{\pi}$, except with negligible probability. Therefore we can say that with overwhelming probability, all queries made in the unremovability game (including execution of $\extract$ contain valid proofs.
    
    Now we can create a reduction $\cB_2$ to break the $\gamma/2 - \negl(\lambda)$-unremovability of $CPA$-secure PKE (of $\sk_2$).
    $\cB_2$ simulates the marking oracle for $\A$ by querying the marking oracle for $\wpke.\sk_1$.
    $\cB_2$ also  samples $\wpke.\sk_1$ on its own.
    
    After $\A$ outputs $C$, $\cB_2$ created program $C_2'$ that works exactly as the circuit $C_2$ in our $\extract(\xk,\pk, C)$ algorithm when $i = 2$, except that $C_2'$ can hardcode the secret key $\sk_1$ sampled by $\cB_2$ to simulate the decryption oracle.

    Also, since we have that $\Pr[\extract(\xk,\pk, C) \notin \cQ ] \geq \epsilon$ for some non-negligible $\epsilon$, we must have that $\Pr[\wpke.\extract(\xk_2,\pk_2, C_2') \notin \cQ ] \geq \epsilon$ for some non-negligible $\epsilon$ by the design of our $\extract$ algorithm.

\end{proof}

We next analyze a symmetric case, consider the following game $G_{\cca, H_3}(\sk, \pk, \cdot)$. 

First, before $\A$ outputs program $C$, we switch to answering $\A$'s decryption queries using $\wpke.\Dec(\sk_2, \cdot)$.

Note that the underlined differences are between \emph{ game $G_{\cca, H_1}(\sk, \pk, \cdot)$ and  game $G_{\cca, H_3}(\sk, \pk, \cdot)$}.

 \begin{itemize}
            \item On input  $(\sk, \pk) = (\pk_1, \pk_2, \sk_1, \sk_2,  \td, \crs)$;
            
            \item For $C$'s decryption queries $\ct = (\ct_1, \ct_2, \pi)$: 
            \begin{itemize}
                \item First check if $\nizk.\verify(\crs, \pi) = 1$, if 0 output $\bot$; else continue; \underline{output $m' \gets \wpke.\Dec(\sk_2, \ct_1)$}.
            \end{itemize}

            \item $C$ submits challenge messages $(m_0, m_1)$ 
            
            The challenger prepares the following ciphertext:
            compute \underline{$\ct_1^* \gets \wpke.\Enc(\pk_1, m_{b_1})$}, \\
            \underline{$b_1 \gets \{0,1\}$; compute $\ct_2^* \gets \wpke.\Enc(\pk_2, m_{b_2})$  $b_2 \gets \{0,1\}$}
            
            compute $\widehat{\pi} \gets \Sim(\td, \crs, (\ct_1^*, \ct^*_2))$  where $\Sim$ is the simulator algorithm for $\nizk$.
            
            Then it sends 
            $\ct^{*} = (\ct_1^*, \ct^*_2, \widehat{\pi})$ to $\A$ and $C$ outputs a guess $b'$.

                \item continues to simulate the decryption oracle as above to decrypt only valid ciphertexts $\ct \neq \ct^*$.
        \end{itemize}

Note that the only difference is that we now answer decryption queries of $C$ using $\sk_2$. 
This results in no difference on the adversary's view because by the statistical soundness of $\nizk$, all ciphertexts submitted to the decryption oracle have their $\ct_1, \ct_2$ encrypt the same message or will result in $\bot$ replies.

By the exact (but symmetric) analysis as \Cref{claim:reduction_wpke2}, we can do a reduction to break the $(\gamma/2-\negl(\lambda))$-unremovability of $\wpke$ where the reduction is challenged with CPA-security game under $(\pk_1,\sk_1, \mk_1, \xk_1)$ and samples $(\pk_2,\sk_2, \mk_2, \xk_2)$ on its own. 

\end{proof}


\section{Watermarkable Weak PRP from Watermarkable weak PRF}
\label{sec:wm_wprp_from_wprf}

In this section, we present a watermarkable weak PRP built using the two-round Feistel network from a watermarkable weak PRF.

A weak $\prp$ of input-output space $\{0,1\}^\ell$ consists of algorithms:
\begin{description}
    \item $\KeyGen(1^\lambda) \to \sk$: a randomized algorithm that generates a secret key $\sk$.
    \item $ \eval(\sk, x \in \{0,1\}^\ell) \to y \in \{0,1\}^\ell$: a deterministic evaluation algorithm that on input $\sk, x$, outputs $y$.
    \end{description}

    The other syntax and definitions (such as correcntess, functionality preserving)  are the same as watermarkable weak $\prf$ and we omit them here.

\paragraph{$\gamma$-Unremovability of weak PRP}
The $\gamma$-Unremovability for a watermarkable implementation of a weak PRF scheme
 says,  for all PPT admissible stateful adversary $\A  $,  there exists a negligible function $\negl(\lambda)$ such that for all $\lambda\in \N$:
$$ \Pr \left[ \extract(\xk, C) \notin \cQ \wedge C \text{ is } \gamma\text{-good}:  \begin{array}{cc} (\sk, \xk, \mk)  \gets \wmsetup(1^\lambda)  \\ 
C \gets \A^{
\watermark(\mk, \sk, \cdot)}(1^\lambda) \\
\end{array} \right] \leq \negl(\lambda).$$

 $\cQ$ is the set of marks queried by $\A$ and $C$ is said to be a PPT admissible, stateful $\gamma$-good circuit if: 
$$ \Pr \left[ C^{\eval(\sk,\cdot)}(x,y_b) = b:  \begin{array}{cc} x \gets \{0,1\}^\ell, y_0 = \eval(\sk, x), y_1  \gets \{0,1\}^\ell, b \gets \{0,1\}\end{array} 
\right] \geq \frac{1}{2} + \gamma.$$

$\eval(\cdot)$ samples a uniformly random $x \gets \{0,1\}^\ell$ upon every query and outputs $(x, \eval(\sk, x))$. 

\subsection{Construction and Security}

Given a watermarkable implementation of a weak PRF $\wprf = (\wmsetup, \eval, \watermark, \extract)$ where the input and output space of $\eval$ is $\{0,1\}^{\ell/2}$, we give the construction based on a 2-round Feistel network as below:

\begin{description}
    \item $\wmsetup(1^\lambda):$ compute $(\sk_1, \xk_1, \mk_1)\gets \wprf.\wmsetup(1^\lambda)$; $(\sk_2, \xk_2, \mk_2)\gets \wprf.\wmsetup(1^\lambda)$; output $\sk = (\sk_1, \sk_2); \xk = (\xk_1, \xk_2); \mk = (\mk_1, \mk_2)$.

    \item $\eval(\sk, x \in \{0,1\}^\ell)$: 
    \begin{enumerate}
        \item parse $\sk = (\sk_1, \sk_2)$; parse $x = x_0 \Vert x_1$ where $x_0, x_1 \in \{0,1\}^{\ell/2}$.

        \item for $i = 1, 2$:
        \begin{description}
            
          \item $x_{i+1} = \wprf.\eval(\sk_i, x_{i}) \oplus x_{i-1}$ \end{description}
           
        \item output $x_2 \Vert x_3$
    \end{enumerate}

    \item $\watermark(\mk, \sk, \tau):$ parse $\sk = (\sk_1, \sk_2); \mk = (\mk_1, \mk_2)$;

    Output $(\tildesk_1 \gets \wprf.\watermark(\sk_1, \tau), \tildesk_2 \gets \wprf.\watermark(\sk_2, \tau))$.

    \item $\extract(\xk, C)$:
    \begin{enumerate}
        \item parse $\xk = (\xk_1, \xk_2)$;
        \item For $i = 1$: create the following circuit $C_1$
        \begin{enumerate}
            \item $C_1$ is harcoded with $\xk_2$ has black-box access to $C$. $C_1$ answers $C$'s queries using $\xk_2$ and external queries:
            \begin{enumerate}
                \item Upon $C$'s query: 
                query external oracle $\cO(\sk_1, \cdot)$ which will return $(x_1 \gets \{0,1\}^{\ell/2}, \\ \wprf.\eval(\sk_1, x_1))$;

                \item sample $(x_2 \gets \{0,1\}^{\ell/2}, \wprf.\eval(\sk_2, x_2 ))$ using $\xk_2$. Let $x_0 := \wprf.\eval(\sk_1, x_1) \oplus x_2$ and $x_3 = \wprf.\eval(\sk_2, x_2 ) \oplus x_1$.

                \item output $(x_0\Vert x_1, x_3 \Vert x_2)$.
            \end{enumerate}
            \item In the challenge phase, receive the challenge $(x_1^*, y^*_1)$ from the external challenger; compute challenge input for $C$: $(x_0^*\Vert x_1^*, x_3^* \Vert x_2^*)$ the same way as above  queries.

            \item $C$ outputs a guess $b'$ and $C_1$ outputs the same.
        \end{enumerate}
        \item compute $\tau_1/\bot \gets \wprf.\extract(\xk_1, C_1)$;

     \item For $i = 2$: create the following circuit $C_2$
        \begin{enumerate}
            \item $C_2$ is harcoded with $\xk_1$ has black-box access to $C$. $C_2$ answers $C$'s queries using $\xk_1$ and external queries:
            \begin{enumerate}
                \item Upon $C$'s query: 
                query external oracle $\cO(\sk_2, \cdot)$ which will return $(x_2 \gets \{0,1\}^{\ell/2}, \\ \wprf.\eval(\sk_2, x_2))$;

                \item sample $(x_1 \gets \{0,1\}^{\ell/2}, \wprf.\eval(\sk_1, x_1 ))$ using $\xk_1$. Let $x_0 := \wprf.\eval(\sk_1, x_1) \oplus x_2$ and $x_3 = \wprf.\eval(\sk_2, x_2 ) \oplus x_1$.

                \item output $(x_0\Vert x_1, x_3 \Vert x_2)$.
            \end{enumerate}
            \item In the challenge phase, receive the challenge $(x_2^*, y^*_2)$ from the external challenger; compute challenge input for $C$: $(x_0^*\Vert x_1^*, x_3^* \Vert x_2^*)$ the same way as above  queries.

            \item $C$ outputs a guess $b'$ and $C_2$ outputs the same.

        \end{enumerate}
        \item compute $\tau_2/\bot \gets \wprf.\extract(\xk_2, C_2)$;

\item output $(\tau_1, \tau_2)$ (which can include $\bot$'s).
        \end{enumerate}

\end{description}

\begin{proof}
    The unremovability proof is relatively straightforward given the framework of proof in \Cref{sec:watermark_composition}. We will make it brief.
    Suppose a given $C$ satisfies that $\Pr[\extract(\xk, C) \notin \cQ] \geq \epsilon$, then it must be that $\Pr[\wprf.\extract(\xk_i, C_i) \notin \cQ]$ for both $i = 1,2$.

    Since the extraction keys of the $\wprf$ scheme can be used to simulate the PRF pseudorandomness game perfectly (see \Cref{sec:wm_cpa_from_prf_construction}), 
     the program $C_i$ built in the $\extract$ algorithm is a stage-2 reduction from weak PRP to weak PRF (with key $\sk_i$). 
    Therefore, for each $i = 1,2$, we can build a reduction $\cB_i$ that samples the key $\sk_{j \neq i, j \in [2]}$ and simulates the 
    $\watermark(\mk, \sk, \cdot)$ queries from $\A$ 
    making queries to the unremovability challenger of $\wprf$ with key $\sk_i$. 
    
     After $\A$ outputs program $C$, $\cB_i$ makes a program $C_i(\sk_j)'$ that uses black-box access to $C$: $C_i(\sk_j)'$ behaves the same as $C_i$ built in $\extract$ except using the real secret key $\sk_j$ instead of the extraction key $\xk_j$ to simulate the game for $C$.

     By the reduction property, either $C_1$ or $C_2$ must be a $\gamma_i$-good program for winning the weak pseudorandomness game of weak PRP for some non-negligible $\gamma_i$: if $C_1$ has only negligible advantage, then we can replace $\wprf.\eval(\sk_1, \cdot)$ in the evaluation algorithm with a real random function (or always give the challenge point evaluation $\wprf.\eval(\sk_1, x_1^*)$  with a random value) and $C$ should still be $(\gamma-\negl(\lambda))$-good. Then we use $C$ in a way where $\wprf.\eval(\sk_1, \cdot)$ is replaced with random to build a program $C_2$: if $C_2$ is also not $\gamma_2$-good for any non-negligible $\gamma_2$, then we must be able to replace $\wprf.\eval(\sk_2, \cdot)$ with a random function while $C$ should still be $(\gamma-\negl(\lambda))$-good. But now since all evaluations are random, $C$ should have no advantage. Contradiction.

\end{proof}

\section{Watermarkable Implementation of Attribute-based Encryption}
\label{sec:abe_construction}

In this section, we give a modified version of the watermarkable $\abe$ scheme in \cite{goyal2019watermarking} which will imply a watermarkable implementation of the $\ibe$ scheme in \Cref{sec:wm_ibe_def}.

\subsection{ Watermarkable Attribute-based Encryption}
\label{def:abe_def}
\paragraph{Syntax} A watermarkable implementation of $\abe$ consists of the following algorithms:

\begin{description}
    \item $\wmsetup(1^\lambda) \to (\mk,\xk,\msk, \mpk)$: on the security parameter, outputs a master public/secret key, a marking key, an extraction key. 

    \item $\KeyGen(\msk, f, \mode \in \{\Marked, \Unmarked\}) \to \sk_f$
    on master secret key $\msk$, a policy $f$, and a mode $\mode \in \{\Marked, \Unmarked\}$  outputs a secret key $\sk_f$.

    \item $\Enc(\mpk, x, m) \to \ct$: on master public key $\mpk$, anttribute $x \in \cX$, a message $m \in \cM$, outputs a ciphertext $\ct$.

    \item $\Dec(\sk_f, \ct) \to m'/\bot$: on secret key $\sk_f$ and ciphertext $\ct$, output a message $m \in \cM$ or $\bot$. 

    \item $\watermark(\mk, \sk, \tau, \mode \in \{\MSK, \SKF\})$: 
    on marking key $\mk$ and secret key $\sk$, a message $\tau \in\cM_\tau$,
and a mode $\mode \in \{\MSK, \SKF\}$, output a marked key $\sk_\tau$.

    \item $\extract(\xk, \mpk, \aux, C)$: on input an extraction key $\xk$, master public key $\mpk$, and program $C$, output a mark $\tau \in \cM_\tau/\bot$. 
\end{description}

\begin{remark}
    The definition in \cite{goyal2019watermarking} only considers running $\KeyGen$ on an unmarked master secret key and marking the policy-embedded key $\sk_f$. 
    
    For the use of our watermarkable $\ibe$ scheme, we consider marking the master secret key and allowing $\KeyGen$ to first "mark" a master secret key.
    
    We additionally allow computing the functional key $\sk_{f,\tau}$ from a marked master secret key $\msk_\tau$.
    When running on $\Marked$ mode, the $\KeyGen$ algorithm also takes in a symbol $\tau$.

    Note that not giving out $\KeyGen(\msk,\cdot)$ oracle to $\A$, but only $\watermark(\msk, \cdot)$ suffices for our use of watermarkable IBE in \Cref{sec:wm_cca2_pke_ibe}. Here we simply prove a slightly stronger security.
\end{remark}

 \paragraph{Correctness}
 There exists a negligible funcion $\negl(\cdot)$
such that for all $\lambda \in \N$, all $x \in \sX$, $f, g\in \cF, m \in \cM$, when $f(x) = 1$ the following holds:
    \begin{align*}
        \Pr\left[\Dec(\sk_{f}, \ct) = m: \begin{array}{cc}
           (\mpk, \msk) \gets \wmsetup(1^\lambda), \sk_f \gets \KeyGen(\msk, f)   
              \\
            \ct \gets \Enc(\mpk, x, m)
        \end{array}\right] \geq 1-\negl(\lambda)
    \end{align*}

\begin{definition}[ABE security]
\label{def:abe_security}

The standard notion of security for a KP-ABE scheme is that of full or adaptive
security. Specifically, a key-policy attribute-based encryption scheme is said to be fully secure if for every stateful PPT adversary A, there exists a negligible function
$negl(\cdot)$, such that for every $\lambda \in \N$ the following holds:
    \begin{align*}
        \Pr\left[\cA^{\KeyGen(\msk, \cdot)}(\ct) =b : \begin{array}{cc}
           (\mpk, \msk) \gets \wmsetup(1^\lambda)   \\
              ((m_0, m_1), x) \gets \cA^{\KeyGen(\msk, \cdot)}(1^\lambda, \mpk) \\
              b \gets \{0,1\}, \ct \gets \Enc(\mpk, x, m_b)
        \end{array}\right] \leq \frac{1}{2} +\negl(\lambda)
    \end{align*}
$\A$ is admissible if all query $f \in \cF$ made to oracle $\KeyGen(\msk, \cdot)$ satisfies $f(x) = 0$.

It is selectively securre if $\A$ needs to output $x$ before seeing $\mpk$.
\end{definition}

\begin{definition}[$\gamma$-Unremovability]
\label{def:wm_abe_security}
The $\gamma$-Unremovability for a watermarkable ABE scheme
 says, for all $\lambda \in \N$, and for all PPT admissible stateful adversary $\A $,  there exists a negligible function $\negl(\lambda)$ such that:
\begin{align*}
 \Pr \left[ \begin{array}{cc}
\extract(\xk, \mpk, \aux=(m_0, m_1, x) , C) \notin \cQ \\ \wedge C \text{ is } \gamma\text{-good}
\end{array}:  \begin{array}{cc}  (\msk,\mpk, \xk, \mk)  \gets \wmsetup(1^\lambda)  \\ 
(\{m_0, m_1\}, x, C) \gets \cA^{\KeyGen(\msk, \cdot), 
\watermark(\mk, \msk, \cdot)}(1^\lambda, \mpk)
\end{array} \right]  \\ 
 \leq \negl(\lambda).
\end{align*}
where $\cQ$ is the set of marks queried by $\A$ and $C$ is a PPT admissible, stateful $\gamma$-good adversary in the security game $G_\abe(\msk, \mpk, \aux = (m_0, m_1, x), \cdot)$, more specifically:
$$ \Pr \left[ C^{\KeyGen(\msk, \cdot)}(\ct) = b:  \begin{array}{cc}
\ct^* \gets \Enc(\mpk, m_b, x), b \gets \{0,1\}
\end{array} \right] \geq \frac{1}{2} + \gamma.$$
$\A, C$ is admissible if all query $f \in \cF$ made to oracle $\KeyGen(\msk, \cdot)$ satisfies $f(x) = 0$.

$C$ will be given the $\KeyGen$ oracle as $\KeyGen(\msk, \cdot, \Unmarked)$ where it only queries on unmarked keys.

A selective variant is requiring $\A$ output $x$ before seeing $\mpk$.
\end{definition}

\begin{remark}
    As discussed in \Cref{sec:wm_cca_ske_def}, the above definition is equivalent to letting $C$ output $(m_0, m_1)$ because we can view any distribution over $(m_0, m_1)$ used by $C$ as a convex combination of different message pairs $\{(m_0, m_1)_i\})_i$ and its corresponding strategy such that the overall winning probability is $1/2+\gamma$. Thus, we can let $\A$ pick the message-pair and corresponding strategy with the largest winning probability and hardcode them into $C$ instead.

    But $x$ needs to be output by $\A$ instead of $C$, because the $\extract$ algorithm does not get to see queries of $\A$ and thus cannot make sure if $f(x) = 0$ for all $f$ queried.
\end{remark}

\begin{remark}[Modification of \cite{goyal2019watermarking} Watermarkable ABE scheme]
    In \cite{goyal2019watermarking}, the unremovability game allows  $\A$ to be given only one policy-embedded key $\sk_f$ and then allowed to query $\watermark(\mk, \sk_f, \cdot)$ for polynomially many times. In our setting, we let $\A$ query a marked master secret key $\watermark(\mk, \msk, \cdot)$ and the $\KeyGen$ algorithm can derive a functional key from a "marked" master secret key.
    We will show how to adapt \cite{goyal2019watermarking}'s construction for our use.
\end{remark}

\subsection{Preliminaries: Delegatable ABE and Mixed FE}
The \cite{goyal2019watermarking}'s watermakable PKE and ABE building blocks are delegatable $\abe$ and mixed $\fe$.

\subsubsection{Delegatable ABE}
 Delegatable ABE has the same syntax with ABE (see \Cref{def:abe_def}, ignoring the extraction key, marking key and $\extract, \watermark$ algorithm) with an additional algorithm $\delegate$:
 \begin{description}
     \item $\delegate(\sk_f, g) \to sk_{f,g}$: on input $\sk_f$ and a predicate $g \in \cF$, output a delegated key $\sk_{f,g}$.
 \end{description}

 \paragraph{Correctness of Delegation}
 There exists a negligible funcion $\negl(\cdot)$
such that for all $\lambda \in \N$, all $x \in \sX$, $f, g\in \cF, m \in \cM$, when $f(x) = g(x) = 1$ the following holds:
    \begin{align*}
        \Pr\left[\Dec(\sk_{f,g}, \ct) = m: \begin{array}{cc}
           (\mpk, \msk) \gets \wmsetup(1^\lambda), \sk_f \gets \KeyGen(\msk, f)   \\
           \sk_{f,g} \gets \delegate(\sk_f, g)
              \\
            \ct \gets \Enc(\mpk, x, m)
        \end{array}\right] \geq 1-\negl(\lambda)
    \end{align*}

\paragraph{Delegatable ABE Security}
The $\dabe$ security is the same as $\abe$ security except that the adversary $\A$ queries an oracle $\cO(\msk, \cdot)$ that has several modes and takes in a query of the form $(f, \ind, \mode)$:
\begin{itemize}
    \item If $\mode = \mathsf{Storekey}$, the challenger generates a new key $\sk_f \gets \KeyGen(\msk, f)$ and store the generated $(n, f, \sk_f)$ where $n$ is an index; update $n := n+1$.

    \item If $\mode = \mathsf{OutputKey}$: the challenger first checks if there exists a key tuple of the form $(\ind, g, sk_g)$. If no such tuple exists or if $g(x) = 1$, it outputs $\bot$. Otherwise, it replies with
$(\ind,sk_g)$.

\item If $\mode = \mathsf{DelegateKey}$, then the challenger first checks if there exists a key tuple of the form
$(\ind, g, sk_g)$. If no such tuple exists or if $g(x) = f(x) = 1$, it outputs $\bot$. Otherwise, it generates
$sk_{g,f} \gets \delegate(\sk_g, f)$ and replies with $(\ind,sk_{g,f})$.
\end{itemize}
    
\subsubsection{Mixed Functional Encryption}

An $\mfe$ scheme consists of the following algorithms:
\begin{description}
    \item $\setup(1^\lambda) \to (\mpk,\msk)$: on input security parameter $\lambda$, outputs the public parameters/master public key $\mpk$ and master secret key $\msk$.

    \item $\Enc(\mpk) \to \ct$: on master public key, the normal encryption outputs ciphertext $\ct$.

    \item $\skenc(\msk, f) \to \ct$: on master secret key $\msk$ and fucnction $f \in \cF$, outputs a ciphertext $\ct$.

    \item $\KeyGen(\msk, m) \to \sk_m$: on master secret key $\msk$ and message $m$, outputs a key $\sk_m$.

    \item $\Dec(\sk_m, \ct) \to \{0,1\}$: on secret key $\sk_m$ and ciphertext $\ct$, outputs a bit.
\end{description}

\paragraph{Correctness}
A mixed functionl encryption scheme is correct if there exists a negligible function $\negl(\lambda)$ such  that for all $\lambda \in \N, f \in \cF, m \in \cM$:
    \begin{align*}
        \Pr\left[\Dec(\sk_{m}, \ct) = 1: \begin{array}{cc}
           (\mpk, \msk) \gets \setup(1^\lambda), \sk_m \gets \KeyGen(\msk, m)   
              \\
            \ct \gets \Enc(\mpk)
        \end{array}\right] \geq 1-\negl(\lambda) \\
        \Pr\left[\Dec(\sk_{m}, \ct) = f(m): \begin{array}{cc}
           (\mpk, \msk) \gets \setup(1^\lambda), \sk_m \gets \KeyGen(\msk, m)   
              \\
            \ct \gets \skenc(\msk, f)
        \end{array}\right] \geq 1-\negl(\lambda) 
    \end{align*}


\paragraph{$q$-Bounded function indistinguishability}
Let $q = q(\lambda)$  be any fixed polynomial. A mixed
functional encryption scheme  is said to satisfy $q$-bounded function indistinguishability security if for every stateful PPT adversary $\A$, there
exists a negligible function $\negl(\cdot)$, such that for every $\lambda 
\in \N$ the following holds:
   \begin{align*}
        \Pr\left[\cA^{\KeyGen(\msk, \cdot), \skenc(\msk, \cdot)}(\ct) = b : \begin{array}{cc}
           (\mpk, \msk) \gets \setup(1^\lambda)   \\
              (f_0, f_1) \gets \cA^{\KeyGen(\msk, \cdot),\skenc(\msk, \cdot)}(1^\lambda, \mpk) \\
              b \gets \{0,1\}, \ct \gets \skenc(\mpk, f_b)
        \end{array}\right] \leq \frac{1}{2} +\negl(\lambda)
    \end{align*}
    where $\A$ can make at most $q$ queries to $\skenc(\msk, \cdot)$ oracle; every secret key query $m$ made by adversary $\A$ to the $\KeyGen(\msk,\cdot)$ oracle must satisfy $f_0(m) = f_1(m)$.

\paragraph{$q$-Bounded accept indistinguishability}
Let $q = q(\lambda)$  be any fixed polynomial. A mixed
functional encryption scheme  is said to satisfy $q$-bounded accept indistinguishability security if for every stateful PPT adversary $\A$, there
exists a negligible function $\negl(\cdot)$, such that for every $\lambda 
\in \N$ the following holds:
   \begin{align*}
        \Pr\left[\cA^{\KeyGen(\msk, \cdot), \skenc(\msk, \cdot)}(\ct_b) = b : \begin{array}{cc}
           (\mpk, \msk) \gets \setup(1^\lambda)   \\
              f^* \gets \cA^{\KeyGen(\msk, \cdot),\skenc(\msk, \cdot)}(1^\lambda, \mpk) \\
              b \gets \{0,1\}, \ct_1 \gets \skenc(\msk, f^*), \ct_1 \gets \Enc(\mpk)
        \end{array}\right] \leq \frac{1}{2} +\negl(\lambda)
    \end{align*}
    where $\A$ can make at most $q$ queries to $\skenc(\msk, \cdot)$ oracle; every secret key query $m$ made by adversary $\A$ to the $\KeyGen(\msk,\cdot)$ oracle must satisfy $f^*(m) = 1$.

\paragraph{$\qtrace$ Jump-Finding Algorithm}
The \cite{goyal2019watermarking} watermarkable ABE schemes $\extract$ procedure also additionally uses an algorithm $\qtrace$ that is widely used in the traitor-tracing literature.
$\qtrace$ on inpput parameters $(\lambda, N, q, \delta, \gamma)$ and given black box access to a program $Q$, runs in time $t = \poly(\lambda, \log N, q, 1/\delta)$. Since detailed discussions are not essential to our presentation on the modified scheme, we refer the readers to \cite{goyal2019watermarking} Section 4.2.3 for details.

\subsection{Construction}
Now we give the construction for watermarkable implementation of ABE, which is modified from \cite{goyal2019watermarking} Section 4.3.

\begin{itemize}
    \item 
    \ifllncs
    Given a delegatable $\abe$ scheme $\dabe = (\dabe.\setup, \dabe.\KeyGen, \dabe.\Enc,  \\ \dabe.\Dec, \dabe.\delegate)$ 
    
    and a $\mfe$ scheme $(\mfe.\setup, \mfe.\KeyGen, \mfe.\skenc, \mfe.\Enc, \mfe.\Dec)$.
    \else
     Given a delegatable $\abe$ scheme $\dabe = (\dabe.\setup, \dabe.\KeyGen, \dabe.\Enc,  \dabe.\Dec, \\ \dabe.\delegate)$ 
    and a $\mfe$ scheme $(\mfe.\setup, \mfe.\KeyGen, \mfe.\skenc, \mfe.\Enc, \mfe.\Dec)$.
    \fi
    
    \item Let $\cM_\tau$ be the marked space, $\cX = \{0,1\}^{\ell_2}$ the attribute space, and $\cC \subseteq \mathsf{Functions}[\cX, \{0,1\}]$ be thepredicate class for the target watermarkable ABE scheme, with parameters $\ell_1, \ell_2$.
    
    \item Let $\cY = \{0,1\}^{\ell_2+\kappa}$ and $ \cD = \cC\cup \{\mfe.\Dec\} \cup \{\mfe.\Dec \wedge C\}_{C \in \cC}$.

    \item the $\dabe$ scheme has attribute class $\cY$, message space $\cM$, predicate class $cD$. Elements in $\cY$ hav the following format $(x \in \cX, \mfe.\ct)$.
    
    \item Let $\gamma$ be the unremovability parameter.
\end{itemize}

\begin{description}
    \item $\wmsetup(1^\lambda) \to (\mpk, \msk, \mk, \xk)$:
    \begin{enumerate}
        \item compute $\mfe.\mpk, \mfe.\msk \gets \mfe.\setup(1^\lambda)$;

        \item compute $(\dabe.\mpk, \dabe.\msk) \gets \dabe.\setup(1^\lambda)$;

        \item output $\mpk = (\mfe.\mpk, \dabe.\mpk); \msk = (\mfe.\msk, \dabe.\msk);$
        
        $ \mk = \xk = (\mfe.\mpk, \mfe.\msk, \dabe.\msk)$.

   \end{enumerate}

   \item $\KeyGen(\msk, f, \mode, \tau) \to \sk_f$:
   Let $\msk = (\mfe.\msk, \dabe.\msk)$
  
   If $\mode = \Unmarked$:
   \begin{enumerate}
       \item Let $\tilde{f}: \{0,1\}^{\ell_2+\kappa} \to \{0,1\}$ denote the predicate $\tilde{f}(x,c) = f(x), x \in \{0,1\}^{\ell_2}, c \in \{0,1\}^{\kappa}$.

       \item output $\sk_f \gets \dabe.\KeyGen(\msk, \tilde{f})$.
   \end{enumerate}

   If $\mode = \Marked$:
   \begin{enumerate}
   
   \item compute $\sk_\tau \gets \mfe.\KeyGen(\mfe.\msk, \tau)$
       \item Let $g_{\tau}$
denote the mixed FE decryption circuit with $\sk_\tau$ hardwired, i.e. $g_\tau = \mfe.\Dec(\sk_\tau, \cdot)$; 

\item compute $\sk_{\tilde{g_\tau}} \gets \dabe.\KeyGen(\dabe.\msk, \tilde{g_\tau})$ where $\tilde{g_\tau}(x,c) = g_\tau(c), \forall x \in \{0,1\}^{\ell_2}, c \in \{0,1\}^{\kappa}$;

\item compute and output delegated key $\sk_{f, \tilde{g_\tau}} \gets \dabe.\delegate(\sk_{\tilde{g_\tau}}, f)$
   \end{enumerate}

   \item $\Enc(\mpk, x, m)$: 

   Let $\mpk = (\mfe.\mpk, \dabe.\mpk);$ compute $\ct_\mfe \gets \mfe.\Enc(\mfe.\mpk)$.

   Next compute ad output $\ct \gets \dabe.\mpk, (x, \ct_\mfe), m)$ where $(x, \ct_\mfe)$ is the attribute.

   \item $\Dec(\sk, \ct) \to m/\bot$: output $m \gets \dabe.\Dec(\sk, \ct)$.

   \item $\watermark(\mk, \msk, \tau)$:
   \begin{enumerate}
       \item $\mk := (\mfe.\mpk, \mfe.\msk)$;

 \item compute $\sk_\tau \gets \mfe.\KeyGen(\mfe.\msk)$
       \item Let $g_{\tau}$
denote the mixed FE decryption circuit with $\sk_\tau$ hardwired, i.e. $g_\tau = \mfe.\Dec(\sk_\tau, \cdot)$; 

\item compute and output $\sk_{\tilde{g_\tau}} \gets \dabe.\KeyGen(\dabe.\msk, \tilde{g_\tau})$ where $\tilde{g_\tau}(x,c) = g_\tau(c), \forall x \in \{0,1\}^{\ell_2}, c \in \{0,1\}^{\kappa}$
   \end{enumerate}

\item $\extract(\xk, \mpk, \aux = (x, m_0, m_1), C, q)$:

Let $\xk = (\mfe.\mpk, \mfe.\msk, \dabe.\msk); \mpk= (\dabe.\mpk, \mfe.\mpk)$;

The extraction algorithms runs the $\qtrace$ algorithm as $\tau \gets \qtrace^{Q_C}(\lambda, 2^{\ell_1}, q, \delta, \gamma)$ where $q$ is the bounded parameter for the $\mfe$ scheme (upper bound on the number of queries to the $\skenc(\msk, \cdot)$ oracles), $\gamma$ is the unremovability paramter and $\delta = \gamma/(5+2\ell_1 q)$. $Q_C$ is simulated as follows:
\begin{mdframed}
    On input $\tau \gets [0, 2^{\ell_1}]$:

    \begin{itemize}
        \item compute $\mfe$ ciphertext as $\ct_\mfe \gets \mfe.\skenc(\mfe.\msk, \comp_\tau)$, where $\comp_\tau$ is the comparison function that on input $z$, output 1 if and only if $z \geq \tau$.

        \item sample a bit $b \gets \{0,1\}$ and compute $\ct_b \gets \dabe.\Enc(\dabe.\mpk, (x, \ct_\mfe), m_b)$ where $(x, \ct_\mfe)$ is the attribute input.

        \item compute $b' \gets C^{\KeyGen(\msk, \cdot, \Unmarked)}(\ct)$ and output 1 if $b' = b$ and 0 otherwise. $C$ has oracle access to $\KeyGen(\msk, \cdot, \Unmarked)$, which can be simulated using $\dabe.\msk$ provided in $\xk$.
    \end{itemize}
\end{mdframed}

\end{description}

\paragraph{Correctness and Other Properties}
The correctness of decryption, correctness of extraction, functionality-preserving and ABE security can directly follow the same proof as in \cite{goyal2019watermarking} so we omit them here.

We can observe that the above scheme also satisfies the additional extraction key simulation property and extraction syntax also satisfy our requirement on watermarkable implementation: the extraction key has $\dabe.\msk$ embedded so enough to simulate the $\KeyGen(\msk, \cdot, \Unmarked)$ queries for program $C$.  The entire extraction algorithm follows the format that runs the simulated stage-2 security game for $\abe$.

\paragraph{Security Proof}
The security proof will be similar to the proof in \cite{goyal2019watermarking}. Overall, our changes are minor.
To avoid repetitive work of doing a same proof verbatimly, we omit the parts of our proof which are exactly the same as \cite{goyal2019watermarking} and refer the readers to Section 4.3.1 of \cite{goyal2019watermarking} "Proof of Theorem 4.15(Unremovability)" part.
We only discuss the major  changes to make to adjust to our construction:

All discussions and lemmas before Lemma 4.18 will be exactly the same when moved to our case, except some minor changes:
\begin{enumerate}
    \item We adjust Experiment $\mathsf{GetCircuit}_\A(\lambda)$ in Figure 2 of \cite{goyal2019watermarking} to fit our unremovability game: $\A$ is allowed to query $\KeyGen(\msk, \cdot)$ for arbitrarily polynomial times. We will discuss how this oracle is simulated when doing different reductions.

    \item In Lemma 4.16 and 4.17 invokes security of $\mfe$: in our setting, the reduction algorithm will additionally query the $\mfe.\KeyGen(\msk, \cdot)$ oracle when dealing with the adversary $\A$s $\KeyGen(\msk, \cdot, \Marked, \cdot)$ queries. But these queries have exactly the same format as when answering the marking queries $\watermark(\mk, \msk, \tau)$. Moreover, in both the $q$-bounded function indistinguishability and accept indistinguishability games, the number of queries a reduction can make to oracle $\mfe.\KeyGen(\msk, \cdot)$ are unbounded polynomial.
    Therefore, the proof of Lemma 4.16 and 4.17 are unaffected.    
\end{enumerate}

We mainly discuss the changes to Lemma 4.18 and take some steps from \cite{goyal2019watermarking} verbatimly for completeness:
Let $\cB$ be the reduction to security of $\dabe$.
\begin{enumerate}
    \item Adversary $\A$  chooses a challenge attribute $x$ and sends it over to $\cB$, since we consider only selective security.

    \item The DABE challenger samples  a $\dabe$ key pair $(\dabe.\mpk, \dabe.\msk) \gets \dabe.\setup(1^\lambda)$. Algorithm $\cB$ samples a mixed FE
    key pair $(\mfe.\mpk, \mfe.\msk) \gets \mfe.\setup(1^\lambda)$,
and sends $\mfe.\mpk$ to $\cB$.
, and sends the public key $\mpk = (\mfe.\mpk, \dabe.\mpk)$ to $\A$.

\item $\cB$ computes $\ct_\mfe^* \gets \mfe.\skenc(\mfe.\msk, \comp_{2^\ell_1})$ and commits to the challenge attribute $x^* = (x, \ct^*_\mfe)$.

\item When $\A$ makes a query to the $\KeyGen(\msk, \cdot)$ and $\watermark(\mk, \msk, \cdot)$ oracles, $\cB$ simulates the answer as follows:
\begin{enumerate}
    \item If $\A$ makes a marking query on message $\tau$: $\cB$ computes the function $\tilde{g_{\sk_\tau}}$ as described in the construction on its own; using $\mfe.\KeyGen$; then it sends query $(\tilde{g_{\sk_\tau}}, \bot, \mathsf{StoreKey})$ to $\cO(\msk, \cdot)$.

    Note that $\tau \in [0, 2^{\ell_1}]$, so $\tilde{g_{\sk_\tau}}(x^*) =0 $ is satisfied with overwhelming probability.

    $\cB$ gets reply in the form $(\ind, \bot)$; $\cB$ then makes another query $(\tilde{g_{\sk_\tau}}, \ind, \mathsf{OutputKey})$
    and gets back $(\ind, \dabe.\sk_{g_{\sk_\tau}})$.

    $\cB$ makes a table that stores the following information $(\ind, \tau, \sk_{g_{\sk_\tau}})$ in each entry.

    \item If $\A$ makes $\KeyGen$ query of the format $(\Marked, \tau, f)$: $\cB$ first checks if an entry containing $\tau$ is stored in its table in the format $(\ind, \tau, \sk_{g_{\sk_{\tau}}})$; if so $\cB$ prepares $\tilde{f}(x,c) = f(x)$ and queries the $\cO(\msk, \cdot)$ oracle on input $(\ind, f, \mathsf{DelegateKey})$ and gets back $(\ind, \dabe.\sk_{g_{\sk_{\tau}}, \tilde{f}})$. $\cB$ sends $\dabe.\sk_{g_{\sk_{\tau}}, \tilde{f}}$ to $\A$.

    If there is no such entry, $\cB$ will  query $(\tilde{g_{\sk_{\tau'}}}, \bot, \mathsf{StoreKey})$ first and then query $(\ind', f, \mathsf{DelegateKey})$ to get the same output. 

    \item If $\A$ makes a $\KeyGen$ query in the format
    $(\Unmarked, f)$, then 
    $\cB$ prepares $\tilde{f}(x,c) = f(x)$ and first queries $(\tilde{f}, \bot, \mathsf{StoreKey})$ to get back some $(\ind_{\tilde{f}}, \bot)$; then queries again $(\ind_{\tilde{f}}, \tilde{f}, \mathsf{OutputKey})$ for $sk_{\tilde{f}}$.
    Note it is also satisfied that $\tilde{f}(x^*) = f(x) = 0$.   

\end{enumerate}

\item At the end of the game, $\A$ outputs a pair of messages $(m_0, m_1)$ and a circuit $C^*$.

\item $\cB$ then sends $(m_0, m_1)$ to DABE challenger and gets a ciphrtext $\ct^*$. $\cB$ samples a random bit $\beta \gets \{0, 1\}$
and computes a fresh ciphertext $\ct \gets \dabe.\Enc(\dabe.\mpk,(x, \ct_{\mfe}), m_\beta)$, where $\ct_\mfe \gets \mfe.\skenc(\mfe.\msk, \comp_{2^{\ell_1}})$.

$C^*$ will continue to make queries but only allowed to query $\KeyGen(\msk, \Unmarked, \cdot)$ oracle. $\cB$ can  answer these queries by querying $\cO(\msk, \cdot)$ with $(\tilde{f}, \bot, \mathsf{StoreKey})$, because $\cB$ is allowed to make such queries in $\dabe$ game:
\begin{itemize}
    \item  If $C$ makes a $\KeyGen$ query in the format
    $(\Unmarked, f)$, then 
    $\cB$ prepares $\tilde{f}(x,c) = f(x)$ and first queries $(\tilde{f}, \bot, \mathsf{StoreKey})$ to get back some $(\ind_{\tilde{f}}, \bot)$; then queries again $(\ind_{\tilde{f}}, \tilde{f}, \mathsf{OutputKey})$ for $sk_{\tilde{f}}$.
    Note it is also satisfied that $\tilde{f}(x^*) = f(x) = 0$.   
\end{itemize}

\item Finally, $\cB$ runs the decryption circuit $C^*$ on $\ct^*$ and $\ct$, and
if $C^{\KeyGen(\msk, \Unmarked, \cdot)}(\ct^*) = C^{\KeyGen(\msk, \Unmarked, \cdot)}(\ct)$, it outputs $b' = \beta$. Otherwise, it outputs $b' = 1 -\beta$.
\end{enumerate}

The rest of the analysis is the same as in the original proof.

\section{Watermarkable Implementation of Digital Signatures}
\label{sec:wm_signatures}
\subsection{Preliminaries: Constrained Signatures}
\label{def:constrain_sig}

The building block of a bounded collusion resistant watermarkable implementation of digital signatures in \cite{goyal2019watermarking} relies on the following building block. 

\paragraph{Constrained Signatures}
A constrained signature with message space $\cM$ and consraint family $\cF \subseteq \mathsf{Function}[\cM, \{0,1\}]$ is a tuple of algorithms:
\begin{description}
    \item $\setup(1^\lambda) \to (\vk, \msk)$. On input the security parameter $\lambda$, the setup algorithm outputs the verification key $\vk$ and the master secret key $\msk$.

\item $\sign(\msk, m) \to \sig$. On input the master signing key $\msk$ and a message $m \in \cM$, the signing
algorithm outputs a signature $\sig$.

\item $\verify(\vk, m, \sig) \to b$. On input the verification key $\vk$, a message $m \in \cM$, and a signature $\sig$, the
verification algorithm outputs a bit $b \in \{0, 1\}$.

\item $\constrain(\msk, f) \to \sk_f$: On input the master signing key $\msk$ and a function $f \in \cF$, the constrain
algorithm outputs a constrained key $\sk_f$.

\item $\ConstrainSign(\sk_f , m) \to \sig$. On input a constrained key $\sk_f$ and a message $m \in \cM$, the signing
algorithm outputs a signature $\sig$.
\end{description}

\paragraph{Correctness}
A constrained signature scheme is correct if for all messages $m \in \cM$ and key pair
$(vk, msk) \gets \setup(1^\lambda)$:
$$ \Pr[\verify(\vk, m, \sign(\msk, m)) = 1] = 1.$$
In addition, for all constraints $f \in \cF$ where $f(m) = 1$, 
$$ \Pr[\verify(\vk, m, \ConstrainSign(\sk_f, m)) = 1: \sk_f \gets \constrain(\msk, f)] = 1. $$

\paragraph{Constrained Unforgeability}
A constrained signature scheme is secure if for every stateful admissible PPT $\A$, there exists a negligible function $\negl(\cdot)$ such that for all $\lambda in \N$:
$$ \Pr \left[ \verify(\vk, m^*, \sig^*) = 1:  \begin{array}{cc}  (\msk,\vk)  \gets \setup(1^\lambda)  \\ 
(m^*, \sig^*) \gets \A^{\sign(\msk, \cdot), \constrain(\msk, \cdot)}(1^\lambda, \vk) \\
\end{array} \right] \leq \negl(\lambda). $$
where $\A$ is admissible if (1) it does not make a signing query on message $m^*$
; and (2)
it does not make a constrained key query for any function $f \in \cF$ such that $f(m^*) = 1$.



\subsection{Definition: Watermarkable Signatures}
\label{sec:def_wm_signatures}
A watermarkable signatures scheme consists of the following algorithms: 
\begin{itemize}
    \item $\wmsetup(1^\lambda) \to (\vk, \sk, \xk, \mk)$: on input security parameter outputs verification key $\vk$, signing key $\sk$; extraction key $\xk$, marking key $\mk$.

    \item $\watermark(\mk,\sk, \tau) \to \sk_\tau$. On input the marking key $\mk$, a signing key $\sk$, and a mark $\tau \in \cM_\tau$ , the
marking algorithm outputs a marked key $\sk_\tau$.

\item $\sign(\sk, m) \to \sig$. On input a signing key $\sk$ and a message $m \in \cM$, the signing
algorithm outputs a signature $\sig$.
$\verify(\vk, m, \sigma) \to 0/1$. On input a verification key $\vk$, a message $m \in \cM$, and a signature $\sig$, the
verification algorithm outputs a bit to signify whether the signature is valid or not.

\item $\extract(\xk, \vk, C) \to \tau /\bot$. On input the extraction key $\xk$, a verification key $\vk$, and a circuit
$C$, the extraction algorithm either outputs a mark $\tau \in \cM_\tau$ or $\bot$.
    
\end{itemize}

The correctness, meaningfulness and other properties are natural to derive from the general watermarking definition \Cref{sec:def_watermarkable_primitive}. We refer to \cite{goyal2019watermarking} for more details and only present functionality-preserving and security defintions. 

\paragraph{Functionality Preserving}
A watermarkable signature scheme satisfies
the functionality-preserving property if there exists a negligible function $\negl(\cdot)$ such that for all
$\lambda \in \N, (\vk, \sk, \mk, \xk) \gets \wmsetup(1^\lambda)$, $m \in \cM, \tau \in \cM_\tau$, the following holds:
\begin{align*}
      \Pr\left[\verify(\vk, m, \wsecEval(\sk_\tau, m)) = 1: \begin{array}{cc}
           (\vk, \sk, \xk, \mk) \gets \wmsetup(1^\lambda), \\ \sk_\tau \gets \watermark(\mk, \sk, \tau)
        \end{array}\right] \geq 1-\negl(\lambda) 
\end{align*}

\jiahui{subtlety on watermarking game: 1.should $\A$ gets to query signing oracle in the first stage? can get "unmarked" signatures on $m^*$ unless commits $m^*$ in 

2.  }

We provide the unremovability definition, which is different from the one defined in \cite{goyal2019watermarking}, but their scheme can be modified to satisfy this definition.


%

\paragraph{ $\gamma$-Unremovability}
For every stateful $\gamma$-unremovable admissible PPT adversary $\A$, there exists a
negligible function $\negl(\cdot)$ such that for all $\lambda \in \N$, the following holds:
$$ \Pr \left[ \extract(\xk, \vk, C) \notin \cQ \\ \wedge C \text{ is } \gamma\text{-good}:  \begin{array}{cc}  (\sk,\vk, \xk, \mk)  \gets \wmsetup(1^\lambda)  \\ 
 C \gets \cA^{ \watermark(\mk, \sk, \cdot)}(1^\lambda, \vk) \\
\end{array} \right] \leq \negl(\lambda).$$
where $\cQ$ is the set of marks queried by $\A$ and $C$ is said to  be a (PPT admissible, stateful) $\gamma$-good adversary if: 
$$ \Pr \left[ C^{\sign(\sk, \cdot)} \to (m^*, \sig^*):  \begin{array}{cc}
\verify(\vk, m^*, \sig^*) = 1
\end{array} \right] \geq \gamma.$$
$ C$ is admissible if and only if it does not query $\sign(\sk, \cdot)$ on $m^*$.

\paragraph{Discussions on Security Definitions and \cite{goyal2019watermarking} construction}
In the \cite{goyal2019watermarking}'s watermarkable signature scheme, the unmarked $\sign$ function and the circuit used to compute signatures using marked keys have different output distributions, if we let $\A$ query $\sign(\sk, \cdot)$ oracle in the first stage but allow $C$ to choose the challenge message to sign on, then there can be an attack. Because the \cite{goyal2019watermarking} scheme 's $\sign(\sk, \cdot)$  and marked signing circuit have different output distributions. If $C$ can choose its own $m^*$ but the $\extract$ algorithm does not know which queries $\A$ has made in stage 1, then $C$ can be hardcoded with some $m^*, \sig'$ where watermark cannot be extracted. We thus give the above definition, which  suffice for our applications 1. 



\subsection{Construction}

We show the following modified construction from \cite{goyal2019watermarking} based on a constrained signature scheme $\csig$ defined in \Cref{def:constrain_sig}:

\begin{itemize}
    \item Let $\cT' = \cT \cup  \{\bot\}$. For a mark $\tau^* \in \cT$, let $f_{\tau^*}: \cT' \times \cM \to \{0,1\}$ be the function $f_{\tau^*}(\tau,m) = 1$ if $\tau = \tau^*$ and 0 otherwise.
    
\end{itemize}

\begin{description}
    \item $\wmsetup(1^\lambda) \to (\vk,\sk, \mk, \xk)$: 
    outputs a signing/verification key-pair $(\vk,\sk) \gets \CSig.\setup(1^\lambda)$; $\mk = \bot, \xk = \sk$.

    \item $\sign(\sk, m) \to \sig$: 
    On input a signing key $\sk$, and a message $m \in \cM$, the signing algorithm signs  $\sig' \gets \CSig.\sign(\sk,(\bot, m))$, and outputs the signature $\sig = (\bot, \sig')$.

    \item $\verify(\vk, m, \sig) \to b$. On input a verification key $\vk$, a message $m \in \cM$, and a signature $\sig = (\tau', \sig')$; the verification algorithm outputs $b \gets \csig.\verify(\vk,(\tau',m), \sig')$.

    \item $\watermark(\mk,\sk,\tau) \to C$. On input a marking key $\mk = \bot$, a signing key $\sk$, and a mark $\tau \in \cT$ , the
marking algorithm computes $\sk_\tau \gets \csig.\constrain(\sk, f_\tau)$ and outputs a circuit $C_\tau: \cM \to
\mathcal{SIG}$ where $C_\tau(\cdot) := (\tau, \csig.\ConstrainSign(\sk_\tau,(\tau, \cdot)))$.

\item $\extract(\xk, \vk, C) \to \tau/\bot$: on $\xk = \csig.\sk, \vk = \csig.\vk$ and circuit $C: \cM \to \mathcal{SIG}$, perform the following for $T  = \lambda/\gamma$ times where $\gamma$ is the unremovability parameter:
\begin{itemize}
    \item For $i \in [T]$: compute $(m_i = (\tau_i', m_i'), \sig_i) \gets C^{\sign(\sk, \cdot)}$. If $\csig.\verify(\vk, (\tau_i, m_i'), \sig_i) = 1$ and $m_i'$ has not been queried, abort and output $\tau_i'$.
\end{itemize}
    
\end{description}

The correctness, functionlity-preserving, meaningfulness proof will all follow exactly from \cite{goyal2019watermarking}.
Even though our extraction algorithm is different, the unremovability proof will follow similarly: any valid signatures provided by $C$ must contain some $\tau_i \in \cQ$ of the marking queries, otherwise it helps break the constrained unforgeable security of the constrained signature.
The only change is that the reduction 
does not answer any $\sign(\sk, \cdot)$ queries in the first stage, and thus do not need to query $\csig.\sign(\sk, \cdot)$ oracle  when interacting with $\A$, but only when interacting with the circuit $C$ produced by $\A$.

\section{Watermarkable CCA-secure Hybrid Encryption (Key Encapsulation Scheme) }
\label{sec:wm_ccahybrid_enc}

\begin{theorem}
    \label{lem:wm_hybrid_enc}
    Assuming LWE, there exists secure watermarkable implementation of a CCA-secure hybrid encryption scheme.
\end{theorem}

The assumption of LWE comes from the need of watermarkable CCA-secure PKE.

\paragraph{Construction}
A watermarkable implementation of CCA-secure hybrid encryption $\whe = (\KeyGen, \\ \Enc, \Dec, \watermark, \extract)$  can be built from a CCA-secure secret key encryption scheme $\ske = (\ske.\KeyGen, \ske.\Enc,  \ske.\Dec)$ and a watermarkable implementatin of CCA-secure public key encryption $\wpke = (\wpke.\wmsetup, \wpke.\Enc,  \wpke.\Dec, \wpke.\watermark, \wpke.\extract)$ as follows:

\begin{itemize}
    \item $\KeyGen(\lambda):$ compute  $(\wpke.\pk, \wpke.\sk, \wpke.\xk, \wpke.\mk) \gets \wpke.\wmsetup(\lambda)$. Output $\sk = \wpke.\sk; \pk = \wpke.\pk; \xk = \wpke.\xk; \mk = \wpke.\mk $ 

    \item $\Enc(\pk, m)$: compute $\ske.\sk \gets \ske.\KeyGen(\lambda)$; then compute $\ct_1 \gets \wpke.\Enc(\pk, \ske.\sk)$; compute $\ct_2 \gets \ske.\Enc(\ske.\sk, m)$. Output $\ct = (\ct_1, \ct_2)$.

    \item $\Dec(\sk, \ct)$: parse $\ct = (\ct_1, \ct_2)$; compute $\sk' \gets \wpke.\Dec(\sk, \ct_1)$ and then $m' \gets \ske.\Dec(\sk', \ct_2)$. Output $m'$.

    \item $\watermark(\mk, \sk, \tau)$: parse $\mk = \wpke.\mk; \sk = \wpke.\sk$; output $\sk_\tau \gets \wpke.\watermark(\sk, \tau)$;
    \item $\extract(\xk, \pk, \aux = \bot, C)$: \begin{itemize}
        \item parse $\xk := \wpke.\xk; \pk := \wpke.\pk$. Create the following circuit $C_\pke$ with black-box access to $C$:
        \begin{itemize}
            \item $C_\pke$ is hardcoded with $\pk$ and simulates the CCA-PKE game for $C$ as follows: 
            
            \item give $\pk$ to $C$; for $C$'s decryption queries $\ct = (\ct_1, \ct_2)$: query an external oracle $\sk' \gets \wpke.\Dec(\wpke.\sk, \ct_1)$;
            output $m' \gets \ske.\Dec(\sk', \ct_2)$.

            \item $C$ submits challenge messages $(m_0, m_1)$; $C_\pke$ submits $(\sk_0 \gets \ske.\KeyGen(1^\lambda), \sk_1 \gets \ske.\KeyGen(1^\lambda))$ to the external challenger; $C_\pke$ receives challenge ciphertext $\ct^*$ from the challenger and $C_\pke$ sends 
            $\ct^{**} = (\ct^*, \ct^*_2 \gets \ske.\Enc(\sk_0, m_b))$ to $C$,  
            where $m_b \gets \{m_0, m_1\}$.
            \item $C$ continues to simulate the decryption oracle as above to decrypt only valid ciphertexts.
            
            \item If $C$ outputs $b' = b$, then $C_\pke$ outputs guess $0$; else it outputs guess 1.
        \end{itemize}
        \item output $\tau/\bot \gets \wpke.\extract(\wpke.\xk, \wpke.\pk, \aux = \bot, C_\pke)$.
    \end{itemize}
\end{itemize}

Now we prove \Cref{lem:wm_hybrid_enc}
\begin{proof}

   Suppose there exists adversary $\A$ that breaks the $\gamma$-unremovability of watermarkable CCA-secure PKE for some non-negligible $\gamma$, i.e. $\A^{\Dec(\sk, \cdot),\watermark(\mk, \sk,\cdot)(\pk}$ produces some program $C$ such that  $\Pr[\extract(\xk, \pk,  C) \notin\cQ] \geq \epsilon$ for some non-negligible $\epsilon$, whereas $C$ satisfies $\Pr[G_{CCA}(\sk, \pk, C) = 1] \geq \frac{1}{2}+\gamma$ (for $G_{CCA}$ see \Cref{sec:wm_cca_ske_def}).
   The probability $\Pr[\extract(\xk, \pk, C) \notin\cQ]$  is taken over the randomness in $\wmsetup$ and the probability  $\Pr[G_{CCA}(\sk, \pk, C) = 1]$ is take over the randomness used in $G_{CCA}$.

    We will show that we can either break $\gamma'$-unremovability of the watermarkable $\wpke$ or CCA-security of $\ske$.

   For any $(\sk, \pk, \xk, \mk)$ generated by $\wmsetup(1^\lambda)$ and any $\A^{\Dec(\sk, \cdot), \watermark(\mk,  \sk, \cdot)}(\pk)$ producing a $\gamma$-good $C$ such that $\Pr[\extract(\xk, \pk, C) \notin \cQ] \geq \epsilon$, one of the following cases must hold:

    
\begin{itemize}
    \item \textbf{Case 1: the program $C_\pke$ created during the execution of $\extract(\xk, C)$ is a $\gamma'$-good adversary for the stage-2 game $G_{\ccapke}(\sk, \pk, \cdot)$ for some non-negligible $\gamma'$. }
    The reduction $\cB_\wpke$ works as follows: 
    \begin{itemize}
        \item $\cB_\wpke$ receives $\wpke.\pk$ from the challenger and sends to $\A$.
        $\cB_\wpke$;
        \item for $\A$'s marking queries : query the challenger's marking oracle $\watermark(\wpke.\sk, \cdot)$ and forward the output.
    \end{itemize}
    After $\A$ outputs program $C$, $\cB_\wpke$ creates circuit $C_\wpke'$ with black-box access to $C$: 
    \begin{itemize}
        \item $C_\wpke'$ is hardcoded with $\pk$ and makes external queries to simulate the decryption oracle for $C$:
        \begin{itemize}
            \item for $C$'s decryption queries $\ct = (\ct_1, \ct_2)$: query the challenger's decryption oracle $\sk' \gets \wpke.\Dec(\wpke.\sk, \ct_1)$;    output $m' \gets \ske.\Dec(\sk', \ct_2)$.
        \end{itemize}

      \item $C$ submits challenge messages $(m_0, m_1)$; $C_\pke$ submits $(\sk_0 \gets \ske.\KeyGen(1^\lambda), \sk_1 \gets \ske.\KeyGen(1^\lambda))$ to the external challenger; $C_\pke$ receives challenge ciphertext $\ct^*$ from the challenger and $C_\pke$ sends 
            $\ct^{**} = (\ct^*, \ct^*_2 \gets \ske.\Enc(\sk_0, m_b))$ to $C$,  
            where $m_b \gets \{m_0, m_1\}$.
            \item $C$ continues to simulate the decryption oracle as above to decrypt only valid ciphertexts.
            
            \item If $C$ outputs $b' = b$, then $C_\pke$ outputs guess $0$; else it outputs guess 1.
    \end{itemize}
   By the design of out $\extract$ and by our assumption, we must have that $\Pr[\wpke.\extract(\wpke.\xk,\\ \wpke.\pk, \aux,  C_\pke) \notin \cQ] \geq \epsilon$ where $C_\pke$ is the program created dring $\extract(\xk, \pk,C)$.

For any $(\xk, \mk, \sk,\pk)$, it is easy to see that program $C_\pke$ created by $\extract(\xk, \pk, C)$'s input-output behavior is the same as $C_\wpke'$ above. Therefore, if $\Pr[\extract(\xk, \pk, C_\pke) \in\cQ]$
  is $\gamma'$-good and $\Pr[\extract(\xk, \pk, C_\pke)] \geq \epsilon$, so will $C_\wpke'$ satisfy these two conditions. Thus $\cB_\wpke$ breaks the $\gamma'$-unremovability.

\vspace{\baselineskip}

      \item \textbf{Case 2: the program $C_\pke$ created during the execution of $\extract(\xk, C)$ is not a $\gamma'$-good adversary for the stage-2 game $G_{\ccapke}(\sk, \pk, \cdot)$ for any non-negligible $\gamma'$. } 
         In this case, we must have that $\vert \Pr[C_\pke \to 0 \vert \ct^* = \pke.\Enc(\pk, \sk_0)] -\Pr[C_\pke \to 0 \vert \ct^* = \pke.\Enc(\pk, \sk_1)]  = \vert \Pr[C \text{ outputs } b' = b \vert \ct^* = \pke.\Enc(\pk, \sk_0)] -\Pr[C \text{ outputs } b' = b \vert \ct^* = \pke.\Enc(\pk, \sk_1)]  \vert \leq \negl(\lambda)$; 

        Thus we can switch to doing the following in running $\extract(
        \xk, \pk, C)$: when  $C$ submits challenge messages $(m_0, m_1)$; 
        $C_\pke$ sends 
            $\ct^{**} = (\ct_1^* = \pke.\Enc(\pke.\pk,\sk_1), \ct^*_2 \gets \ske.\Enc(\sk_0, m_b))$ to $C$,  
            where $m_b \gets \{m_0, m_1\}$ where $\ct^*$ is always $\pke.\Enc(\pke.\pk,\sk_1)$.

            Since $C$ is still $(\gamma-\negl(\lambda))$ good, we will build a reduction to break IND-CCA security of 
            $\ske$.

            The reduction $\cB_\ske$ samples $\wpke$'s keys on its own and can query the oracles $\ske.\Enc(\sk_0, \cdot), \\\ske.\Dec(\sk_0, \cdot)$ provided by the $\ske$ challenger. 

            After $\A$ outputs $C$,
            $\cB_\ske$ enters stage-2 of reduction and  simulates the encryption and decryption oracles by submitting the decryption queries to the encryption, decryption oracles $\ske.\Enc(\sk_0, \cdot),\ske.\Dec(\sk_0, \cdot)$ in the CCA-security game of SKE.
            In the challenge phase, 
            $C$ submits challenge messages $(m_0, m_1)$; $\cB_\ske^2$ submits $(m_0, m_1)$ to the external challenger; and samples $\sk_1 \gets \ske.\KeyGen(1^\lambda), $;  $C_\pke$ receives challenge ciphertext $\ct^*_2 \gets \ske.\Enc(\sk_0, m_b)$ from the challenger, where $\sk_0$ is only known to the challenger and $C_\pke$ sends 
            $\ct^{**} = (\ct^*_1 = \wpke.\Enc(\pk, \sk_1), \ct^*_2 \gets \ske.\Enc(\sk_0, m_b))$ to $C$.
            $C$ continues to simulate the encryption and decryption oracle as above to decrypt only valid ciphertexts. 
            In the end, if $C$ outputs $b'$, then $\cB_\ske^2$ outputs the same $b'$.

\end{itemize}

\end{proof}

\section{Watermarkable Functional Encryption from Watermarkable Attribute-Based Encryption}
\label{sec:wm_fe_from_abe}


We give a high-level description on how to turn the \cite{goldwasser2013reusable} functional encryption construction from ABE, FHE, garbled circuits into a watermarkable FE based on watermarkable ABE. We will not elaborate details of the construction as we cannot possibly cover all explicit constructions for watermarking-compositions of existing schemes, and our main goal here is to give reference to another example of watermarking-composition when applied to an advanced encryption scheme.

All the main algorithms are the same as the \cite{goldwasser2013reusable} construction. We will mainly remark on how we do extraction and refer interested readers to \cite{goldwasser2013reusable} for the detailed construction.

\paragraph{Watermarkable Functional Encryption: Definitions} A watermarkable $\fe$ scheme can be given different 
watermarking definitions. One is to watermark the master secret key and given out an unmarked functional key to the adversarial program $\A$; the other is to let $\A$ submit a function $f$ in the first stage and gets to query on marked versions of the functional
key $\sk_f$. Both notions are interesting.

For the first notion, we can consider the following security game: 
The $\gamma$-Unremovability for a watermarkable single-key FE scheme
 says, for all $\lambda \in \N$, and for all PPT admissible stateful adversary $\A $,  there exists a negligible function $\negl(\lambda)$ such that:
$$ \Pr \left[ \begin{array}{cc}
\extract(\xk, \mpk, \aux= x , C) \notin \cQ \\ \wedge C \text{ is } \gamma\text{-good}
\end{array}:  \begin{array}{cc}  (\msk,\mpk, \xk, \mk)  \gets \wmsetup(1^\lambda)  \\ 
(x, C) \gets \cA^{\watermark(\mk, \msk, \cdot)}(1^\lambda, \mpk) \\
\end{array} \right] \leq \negl(\lambda).$$
where $\cQ$ is the set of marks queried by $\A$ and $C = (C_1, C_2, D)$ is a PPT admissible, stateful $\gamma$-good adversary in the security game $G_\abe(\msk, \mpk, \aux = (x), \cdot)$ if:
$$ \Pr \left[ D(\st, \ct_b) = b:  \begin{array}{cc}
(f, \st_1) \gets C_1(\mpk) \\
\sk_f \gets \fe.\KeyGen(\msk, f) \\
\ct_0 \gets \fe.\Enc(\mpk, x), \ct_1 \gets \Sim(\mpk, \sk_f, f, f(x), 1^{|x|}), b \gets \{0,1\}\\
 \st \gets C_2(\st_1, \ct_b)
\end{array} \right] \geq \frac{1}{2} + \gamma.$$
where $\Sim$ is some PPT simulator.


A selective variant is requiring $\A$ output $x$ before seeing $\mpk$. Note that we need $\A$ to output $x$ in the first stage because by relying out construction on ABE, the choice of message $x$
will influence the choice of the ABE reduction's choice of challenge attribute, which must be committed to before entering stage 2. 
We can also let $\A$ outputs $f$ first so that $\A$ sees $\sk_f$ before outputing $x$ (which matches the semantics of a fully secure FE).
But then the $\extract$ algorithm will have to take in $f$ as an auxiliary inputs. 

The case where $\sk_f$ gets watermarked is slightly trickier to define. We need to work with an indistinguishability-based FE security where the function family $\cF$ is a "high-entropy" function family. Also, the $\extract
$ needs to take in the challenge $f$
as an auxiliary input. We refer to \cite{goyal2019watermarking} Remark 4.5 for more discussions. \cite{goyal2019watermarking}'s watermarkable predicate encryption
scheme can be extended to a watermarkable FE scheme for watermarking a functional key when working with the above restricted function class.

\subsection{Two-Outcome Attribute-Based Encryption}
The first building block for \cite{goldwasser2013reusable} FE is a two-outcome ABE, called $\abe_2$. It has the similar syntax to a normal ABE except that: the encryption scheme takes in $\mpk, x, m_0, m_1$; when decrypting with a functional key $\sk_f$ on ciphertext $\ct_{x, m_0, m_1}$, the decryption will output $m_b$ when $f(x) = b$.

The single-key $\abe_2$ security game is exactly like a single-key ABE security game except that $\A$ provides 3 challenge messages $(m, m_0, m_1)$ and attribute $x$, then $\A$ is given either $\abe_2.\Enc(\mpk, x, (m_0, m))$ if $f(x) = 1$ for the key generation query $f$,  or 
$\abe_2.\Enc(\mpk, x, (m, m_1))$  if $f(x) = 0$, and asked to output $b$.

A single-key $\abe_2$ scheme can be built with a simple construction from a regular single-key $\abe$ .
We refer the readers to \cite{goldwasser2013reusable} Appendix B for details and give a high level idea:
the scheme generates two pairs regular ABE keys; the encryption generates two encryptions of the same message and attribute, each under one of the master public keys respectively; the decryption algorithm on $\sk_f$ and $\ct = (\ct_0, \ct_1)$ will try to decrypt both ciphertexts, and output that one that has a valid decryption result.

In the watermarking construction, we watermark both the master secret keys of the two $\abe$ (regarding single-key watermarkable ABE security: see \Cref{def:wm_abe_security} except that $\A$ is not give $\KeyGen(\msk,\cdot)$ oracle and $C$ can only query $\KeyGen(\msk,\cdot)$ once)).
The $\extract$ algorithm in the watermarkable $\abe_2$ will turn the input program $C$ into a circuit $C_i$ to break the single-key ABE security for each of the keys in the construction. The reduction in \cite{goldwasser2013reusable} Appendix B would go through because 
each stage of $\abe_2$ in the unremovability game would map to the stage of the regular $\abe$ unremovability game. That is, making some $\msk$ marking queries after receiving $\mpk$, and outputs challenge messages, attribute and program $C$; then $C$ gets to query a single $sk_f$ and challenged with the challenge ciphertext (If we choose to work with a watermarkable FE scheme where we let $\A$ choose $f$, we can also let the adversary in watermarkble $\abe$ scheme query one $\sk_f$ and then the program $C$ is not allowed to make any queries). 

\paragraph{Watermarkable Functional Encryption Construction, Extraction and Security}
We refer the readers to \cite{goldwasser2013reusable} Section 3.1 for details on the construction and its preliminaries for definitions of FHE and garbled circuit, since we only plan to give an idea here.

The construction only needs to rely on a watermarkable $\abe_2$ scheme (which we have shown above can be based on watermarkable regular $\abe$). The FHE and garbled circuits do not have to be watermarked because they are only generated freshly on each evaluation.

The $\wmsetup$ algorithm generates all $\lambda$ number of $(\mpk_i, \msk_i, \mk_i, \xk_i)$ from the setup of a watermarkable $\abe_2$ scheme.

The $\KeyGen(\msk, f)$ algorithm generates a $sk_i \gets \abe_2.\KeyGen(\msk_i, \fhe.\eval_f^i)$ where $\fhe.\eval_f^i$ is outputing the $i$-th bit after the FHE evaluation on some FHE $\eval$ key $\fhe.\pk$, the input function $f$ and some FHE ciphertexts $c_1, \cdots, c_n$.

The $\Enc(\mpk, x)$ algorithm takes input attribute $x = x_1 \cdots x_n$, generates a fresh $\fhe$ key pair, encrypt each $x_i$ to get FHE ciphertext $c_i$. Then produce garbled circuit for the FHE decryption algorithm  $\fhe.\Dec(\fhe.\sk, \cdot)$ with $2\lambda$ labels $\{L_i^0, L_i^1\}_{i \in [\lambda]}$. Then we produce $\abe_2$ ciphertext $\ct_i \gets \abe_2.\Enc(\mpk_i, (\fhe.\pk,\vec{c}), L_i^0, L_i^1)$ for all $i \in [\lambda]$ where $\vec{c} = (c_1, \cdots, c_n)$.
Output the ciphertext $(\ct_1, \cdots \ct_\lambda)$ and the garbled circuit.

The decryption algorithm $\Dec(\sk_f, \ct)$ runs the $\abe_2$ decryption on each$\ct_i$ using $\sk_i$ to obtain a label $L_i^{d_i}$, for each $i \in [\lambda]$.
Then one can recover $f(x) \gets \fhe.\Dec(\fhe.\sk, d_1 \cdots d_\lambda)$ by using the labels  $L_i^{d_i}$ and the garbled circuit.

Now we can talk about how we watermark and extract: the marking scheme simply marks each $\abe_2.\msk_i$ for $i \in [\lambda]$. To see how we extract, we roughly open up the security proof of the above scheme: On an input program $C$, the extraction interacts with $C$ as in the stage-2 FE game in our unremovability definition. Recall that in the end of the original security game, we need to prepare either real ciphertext $\ct \gets \Enc(\mpk, x)$ or a simulated ciphertext, and ask the final stage adversary $D$ to distinguish. But in our extraction algorithm that tries to create circuit $C_i$ to break the unremovability of the $i$-the $\abe_2$ scheme, we instead sample either of the following: (1) generate the ciphertext as in the original scheme; (2) use the label  $L_i^{d_i}$ twice in the ciphertext $\ct_i' \gets \abe_2.\Enc(\mpk_i, (\fhe.\pk, \vec{c}), L_i^{d_i}, L_i^{d_i})$, where $d_i = \fhe.\eval_f^i(\fhe.\pk, \vec{c})$. 
More specifically, circuit $C_i$ obtains one of these ciphetexts from some external $\abe_2$ challenger and uses $C$'s output to distinguish them.

Why would a $\gamma$-good $C$ in the original security game still be $\gamma$-good when we simulate the game as in the above extraction?
This relies on the hybrid arguments invoking the IND-CPA security of $\fhe$ and then  the security of garbled circuit. We refer the details to the proofs for Hybrid 0 to Hybrid 2 in Section 3.2 of \cite{goldwasser2013reusable}. By invoking these security properties, we move from giving a completely simulated ciphertext, to give the fake ciphertext described in the extraction algorithm (when flipping the coin to $1$), while $C$ should still be $\gamma$-good with some noticeable probability.

Then we can argue that at least one of $C_i$ is $\gamma_i$-good for some non negligible $\gamma_i$, we can break unremovability of $\abe_2$ with the $i$-th key pair.
 In the marking stage, the $i$-th $\abe_2$ reduction can simulate the marking queries by making marking query to the $\abe_2$ marking oracle (and mark the other keys on its own). 
 After $\A$ commits to chellenge message $x$, the reduction can also commit the challenge attribute $(\fhe.\pk, \vec{c})$ where $\vec{c}$ is computed from $x$ as in the construction. 
 After $\A$ outputs $C$, the reduction creates circuit $C_i$ as in the extraction algorithm, and $C_i$ can answer $C$'s single key query by making a single query to the $\abe_2$ $\KeyGen(\msk, \cdot)$ oracle. 


\section{Watermarkable Non-malleable PKE from Watermarkable CPA-secure PKE, NIZK and Signatures}
\label{sec:ddn91_wm}

\subsection{Watermarkable Implementation of Non-malleable CCA2 PKE}

We present that an alternative construction to \Cref{sec:naor_yung_wm}, the CCA2 non-malleable PKE in \cite{dolev1991non} also has a watermarking implementation. We will omit the security proof due to the similarity to \Cref{sec:naor_yung_wm}.

\subsection{Construction and Security}

Our construction is based on the non-malleable encryption scheme in \cite{dolev1991non}.

Given a watermarkable implementation of a CPA-secure PKE $\wpke = (\wmsetup, \\ \Enc, \Dec, \watermark, \extract)$, a NIZK scheme $\nizk = (\setup, \prove, \verify)$ and a signature scheme $\DS = (\KeyGen,\sign, \verify)$. where the verification key size of $\DS$ is $n$.

\paragraph{Construction}

\begin{description}
    \item $\KeyGen(1^\lambda, 1^n):$
    For $i \in [n], b \in \{0,1\}$: 
     compute  $(\wpke.\pk_{i,b}, \wpke.\sk_{i,b}, \xk = \wpke.\xk_{i,b}; \mk = \wpke.\mk_{i,b}) \gets \wpke.\wmsetup(\lambda)$;

     Compute $\crs \gets \nizk.\setup(1^\lambda)$;
    Output $\sk = (\wpke.\sk_{i,b})_{i\in [n], b \in \{0,1\}}; \pk = (\{\wpke.\pk_{i,b}\}_{i\in [n], b \in \{0,1\}}, \crs); \xk = \{\wpke.\pk_{i,b}\}_{i\in [n], b \in \{0,1\}}; \mk = \{\wpke.\pk_{i,b}\}_{i\in [n], b \in \{0,1\}}$ 
    
    \item $\Enc(\pk, m)$: 
    \begin{itemize}
        \item parse $\pk := \{\pk_{i, b}\}_{i\in [n], b \in \{0,1\}}, \crs$; compute $(\DS.\sk, \DS.\vk) \gets \DS.\KeyGen(1^\lambda)$;
        Let us denote $i$-th bit of $\DS.\vk$ as $\vk_i$.
        
        \item For each $i\in [n]$: compute 
        $\ct_i \gets \wpke.\Enc(\pk_{i, \vk_i}, m)$; 

        \item compute $\pi \gets \nizk.\prove(\crs, (\ct_1, \ct_2, \cdots, \ct_n, \DS.\vk), (r_1, r_2, \cdots, r_n, m))$ for the following statement:

        $\exists$ witness $(r_1, r_2 , \cdots, r_n, m)$ such that $\ct_i = \wpke.\Enc(\pk_{i, \vk_i}, m; r_i)$ for all $i \in [n]$, where $r_i$ is the randomness used in encryption.
        
        \item compute a signature $\sig \gets \DS.\sign(\DS.\sk, (\{\ct_i\}_{i \in n}, \pi))$
        \item Output $\ct = (\{\ct_i\}_{i \in n}, \DS.\vk,  \pi, \sig)$.
    \end{itemize}

    \item $\Dec(\sk, \ct)$: 
    \begin{itemize}
        \item parse $\ct = (\{\ct_i\}_{i \in n}, \DS.\vk, \pi); \sk = \{\sk_{i,b}\}_{i\in [n], b \in \{0,1\}}$;
        \item if $\DS.\verify(\DS.\vk, (\{\ct_i\}_{i \in n}, \pi), \sig) = 1$ continue; else abort and output $\bot$.
        \item if $\nizk.\verify(\crs, \pi, ((\{\ct_i\}_{i \in n}, \DS.\vk)) = 1$, continue; else abort and output $\bot$.

        \item compute $m \gets \wpke.\Dec(\sk_{1, \vk_1}, \ct_1)$;
        output $m'$. 
    \end{itemize}

    \item $\watermark(\mk, \sk, \tau)$: parse $\mk = \{\mk_{i, b}\}_{i\in [n], b \in \{0, 1\}}; \sk =  \{\sk_{i, b}\}_{i\in [n], b \in \{0, 1\}}$; output $ \{\sk_{i, b, \tau} \gets \wpke_{i,b}.\watermark(\sk_{i,b}, \tau)\}\}_{i\in [n], b \in \{0, 1\}}$;
    
    \item $\extract(\xk, \pk, \aux = (m_0, m_1), C)$: \begin{itemize}
        \item parse $\xk := ( \{\xk_{i, b}\}_{i\in [n], b \in \{0, 1\}}, \nizk.\td); \pk := (\{\pk_{i, b}\}_{i\in [n], b \in \{0, 1\}}, \crs)$. 
        Initialize an empty tuple $\vec{\tau}$;
        
        For $i^* = 1, \cdots, n; b^* = 0, 1$:
        Create the following circuit $C_{\istar,b^*}$ with black-box access to $C$:
        \begin{itemize}
            \item $C_{\istar,b^*}$ is hardcoded with $(\{\pk_{i,b}\}_{i \in [n], b\in \{0,1\}}, \{\xk_{i,b}\}{i \neq \istar \vee b \neq b^*},  \td, \crs)$
            
            $C_{\istar,b^*}$ acts as a stage-2 reduction from CCA2-PKE to CPA-PKE with the keys $(\sk_{\istar, b^*}, \pk_{\istar, b^*} )$ simulates the CCA2-PKE stage-2 game for $C$ as follows: 
            
            \item For $C$'s decryption queries $\ct = (\{\ct_{i,\vk_i}\}, \vk, \pi, \sig) $: 
            \begin{itemize}
                \item First check if $\nizk.\verify(\crs, \pi) = 1$, if 0 output $\bot$; else continue;

                \item By the extraction key simulation property, since $\xk_{i \neq \istar \vee b \neq b^*}$ can be used to simulate the oracles used in $G_{CPA}(\pk_i, \sk_j, \cdot)$ for $i \neq \istar$ or $b \neq b^*$, $C_{\istar, b^*}$ can simulate the oracle $\wpke.\Dec(\sk_{\istar, 1-b^*}, \cdot)$ and thus decrypt $\ct_{\istar, 1-b^*}$ in the ciphertext $\ct$ to output $m$.
            \end{itemize}

            \item 
            $C_{\istar, b^*}$ submits $(m_0, m_1)$ to the external challenger; $C_{\istar, b^*}$ receives challenge ciphertext $\ct_{\istar, b^*}^* = \wpke.\Enc(\pk_{\istar, b^*}, m_\delta), \delta \gets \{0,1\}$ from the challenger.

            \item
            $C_{\istar, b^*}$ prepares the following ciphertext:

            \begin{enumerate}
                \item   Sample $(\DS.\vk^*, \DS.\sk^*) \gets \DS.\KeyGen(1^\lambda)$ so that the $i^*$-th bit of $\vk^*$ is $b^*$.

                \item   Compute $\ct_{i,b_i}^* \gets \wpke.\Enc(\pk_{i,b_i}, m_{\delta'}), \delta' \gets \{0,1\}$ for all $i \neq \istar$, where $b_i = \vk^*_i$.

                \item  \underline{ compute $\widehat{\pi} \gets \Sim(\td, \crs,$,$ (\{\ct_{i,b_i}^*\}_{i \in [n], b_i = \vk^*_i}))$ } \\ \underline{ where $\Sim$ is the simulator algorithm for $\nizk$.}
            \end{enumerate}

            Then it sends 
            $\ct^{*} = (\{\ct_{i,b_i}^*\}_{i \in [n], b_i = \vk^*_i}), \vk^*, \widehat{\pi})$ to $C$.
            
            \item $C$ continues to simulate the decryption oracle as above to decrypt only valid ciphertexts $\ct \neq \ct^*$, except adding the following check:
            \begin{itemize}
                \item \underline{Check if $\vk = \vk^*$, if yes, output $\bot$; else continue to decrypt.}
            \end{itemize}

            \item $C$ outputs $\ell$ ciphertexts $\ct_1, \cdots, \ct_\ell$ and $C_i$ computes $d_j \gets \Dec(\sk, \ct_j)$ if $ct_j \neq \ct^*$, else $d_j = \bot$.
            
            \item In the end, feed $(d_1, \cdots, d_\ell)$ to $C$ and $C$ outputs a bit b'.
            $C_i$ outputs the same as
            $C$ outputs.

            \item Add $\tau/\bot \gets \wpke.\extract(\wpke.\xk_{i^*, b^*}, \wpke.\pk_{i^*, b^*}, \aux = (m_0, m_1), C_{i^*, b^*})$ to $\vec{\tau}$.
        \end{itemize}
        \item Output $\vec{\tau}$.
    \end{itemize}
\end{description}

The security proof is highly similar to \Cref{sec:naor_yung_wm} and we omit it here. 

\fi



\vspace{\baselineskip}

\end{document}